\catcode`\@=11
\newif\if@fewtab\@fewtabtrue
{\count255=\time\divide\count255 by 60
\xdef\hourmin{\number\count255}
\multiply\count255 by-60\advance\count255 by\time
\xdef\hourmin{\hourmin:\ifnum\count255<10 0\fi\the\count255}}
\def\ps@draft{\let\@mkboth\@gobbletwo
    \def\@oddfoot{\hbox to 7 cm{\tiny \versionno
       \hfil}\hskip -7cm\hfil\rm\thepage \hfil {\tiny\draftdate}}
    \def\@oddhead{}
    \def\@evenhead{}\let\@evenfoot\@oddfoot}
\def\draftdate{\number\month/\number\day/\number\year\ \ \ \hourmin }

\def\citen#1{\if@filesw \immediate\write \@auxout {\string\citation{#1}}\fi%
\@tempcntb\m@ne \let\@h@ld\relax \def\@citea{}%
\@for \@citeb:=#1\do {\@ifundefined {b@\@citeb}%
    {\@h@ld\@citea\@tempcntb\m@ne{\bf ?}%
    \@warning {Citation `\@citeb ' on page \thepage \space undefined}}%
    {\@tempcnta\@tempcntb \advance\@tempcnta\@ne
    \setbox\z@\hbox\bgroup\ifcat0\csname b@\@citeb \endcsname \relax
    \egroup \@tempcntb\number\csname b@\@citeb \endcsname \relax
    \else \egroup \@tempcntb\m@ne \fi \ifnum\@tempcnta=\@tempcntb
    \ifx\@h@ld\relax \edef \@h@ld{\@citea\csname b@\@citeb\endcsname}%
    \else \edef\@h@ld{\hbox{--}\penalty\@highpenalty
    \csname b@\@citeb\endcsname}\fi
    \else \@h@ld\@citea\csname b@\@citeb \endcsname \let\@h@ld\relax \fi}%
\def\@citea{,\penalty\@highpenalty\hskip.13em plus.13em minus.13em}}\@h@ld}
\def\@citex[#1]#2{\@cite{\citen{#2}}{#1}}%
\def\@cite#1#2{\leavevmode\unskip\ifnum\lastpenalty=\z@\penalty\@highpenalty\fi%
  \ [{\multiply\@highpenalty 3 #1%
  \if@tempswa,\penalty\@highpenalty\ #2\fi}]}   %
\makeatother % end of CITE.STY
\catcode`\@=12

\def\be            {\begin{equation}}
\def\bearl         {\begin{array}{l}}
\def\bearll        {\begin{array}{ll}}
\def\beginitemize  {\def\leftmargini{1.12em}~\\[-1.55em]\begin{itemize}\addtolength{\itemsep}{-6pt}}
\def\boti          {\,{\boxtimes}\,}
\def\Bun           {\mathrm{Bun}}
\def\C             {{\ensuremath{\mathcal C}}}
\def\cala          {{\ensuremath{\mathcal A}}}

\def\calc          {{\ensuremath{\mathcal C}}}
\def\cald          {{\ensuremath{\mathcal D}}}
\def\calf          {{\ensuremath{\mathcal F}}}
\def\calm          {{\ensuremath{\mathcal M}}}
\def\calw          {{\ensuremath{\mathcal W}}}

\def\cdo           {\,{\cdot}\,}

\def\cir           {\,{\circ}\,}
\def\cobord        {\mathrm{cobord}}
\def\complex       {{\ensuremath{\mathbb C}}}
\def\complexx      {{\ensuremath{\mathbb C}^\times_{}}}

\def\ee            {\end{equation}}

\def\eear          {\end{array}}
\def\End           {{\ensuremath{\mathrm{End}}}}

\def\erf           {\eqref }
\def\eq            {\,{=}\,}
\def\ev            {\mathrm{ev}}

\def\findim        {fini\-te-di\-men\-si\-o\-nal}

\def\Fun           {{\ensuremath{\mathrm{Fun}}}}
\def\gam           {\gamma}
\def\Gl            {G_{\rm l}}
\def\Gr            {G_{\rm r}}
\def\GVect         {\ensuremath G\mbox{-}\ensuremath{\mathrm{vect}}}
\def\Hom           {{\ensuremath{\mathrm{Hom}}}}

\newcommand\hsp[1] {\mbox{\hspace{#1 em}}}
\def\id            {\mbox{\sl id}}

\def\ii            {{\rm i}}

\def\iN            {\,{\in}\,}
\def\Ind           {{\ensuremath{\mathrm{Ind}}}}

\def\ko            {{\ensuremath{\Bbbk}}}
\newcommand\labl[1]{\label{#1}\ee}
\def\le            {\,{\leq}\,}
\def\Map           {{\ensuremath{\mathrm{Map}}}}
\def\Mod           {\mbox{-mod}}
\def\Modf          {\mbox{\footnotesize -mod}}

\newcommand\Nxl[1] {\\[-1.3em]\\[#1mm]}

\def\ohr           {\reflectbox{$\rho$}}

\def\onedim        {one-di\-men\-si\-o\-nal}
\def\op            {^{\mathrm{op}}}
\def\opcl          {^{\mathrm{op/cl}}}
\def\oti           {\,{\otimes}\,}

\def\qquand        {\qquad{\rm and}\qquad}

\def\rmd           {{\rm d}}
\def\rev           {{\mathrm{rev}}}
\def\rog           {\varrho_\gamma}
\def\sll           {\,{\backslash}\hspace{-2.5pt}{\backslash}\hspace{.6pt}} 

\def\srr           {\hspace{-.2pt}\reflectbox{$\backslash$}\hspace{-2.5pt}\reflectbox{$\backslash$}\hspace{-.6pt}}
\def\srrad         {\srr_{\!_{\text{ad}}}}          
\def\tBun          {\widetilde{\mathrm{Bun}}}
\def\tft           {\mathrm{tft}}

\def\thetagamma    {\vartheta_\gamma}

\def\Times         {\,{\times}\,}
\def\To            {\,{\to}\,}

\def\Vect          {\ensuremath{\mathrm{vect}}}
\def\Vectc         {\ensuremath{\mathrm{vect}_\complex}}
 \newcommand\void[1]{}
\def\wht           {\ensuremath{\mathcal W_{H,\theta}}}
\def\Z             {{\ensuremath{\mathcal Z}}}
\def\zet           {{\ensuremath{\mathbb Z}}}

\documentclass[12pt]{article}
\usepackage{latexsym, amsmath, amsthm, amsfonts} 
\usepackage{enumerate, amssymb, xspace, xypic}
\usepackage[all]{xy}
\usepackage[mathscr]{eucal}
\usepackage{graphicx} \usepackage{rotating}
\usepackage{epstopdf,hyperref}
\usepackage{color}

\newtheorem{thm}{Theorem}

\newtheorem{lem}[thm]{Lemma}

\newtheorem{prop}[thm]{Proposition}

\theoremstyle{definition}
\newtheorem{Example}[thm]{Example}
\newtheorem{defi}[thm]{Definition}
\newtheorem{Definition}[thm]{Definition}
\newtheorem{rem}[thm]{Remark}

\newcommand\eqpic[4]{\begin{eqnarray}
                   \begin{picture}(#2,#3){}\end{picture}\nonumber\\
                   \raisebox{-#3pt}{ \begin{picture}(#2,#3) #4 \end{picture} }
                   \label{#1} \\~\nonumber \end{eqnarray} }
\newcommand\Eqpic[4]{\begin{eqnarray}
                   \begin{picture}(#2,#3){}\end{picture}\nonumber\\
                   \raisebox{-#3pt}{ \begin{picture}(#2,#3) #4 \end{picture} }
                   \nonumber \\[3pt]~\label{#1} \end{eqnarray} }
\newcommand\Includepic[1]   {{\begin{picture}(0,0)(0,0)
                            \scalebox{.32}{\includegraphics{imgs/fusV2_#1.eps}}\end{picture}}}

\begin{document}

\def\cir{\,{\circ}\,} 
\numberwithin{equation}{section}
\numberwithin{thm}{section}
                                     
\begin{flushright}
   {\sf ZMP-HH/13-13}\\
   {\sf Hamburger$\;$Beitr\"age$\;$zur$\;$Mathematik$\;$Nr.$\;$484}\\[2mm]
  July 2013
\end{flushright}
\vskip 3.5em
                                   
\begin{center}
\begin{tabular}c \Large\bf 
A geometric approach to boundaries and \\[3mm] \Large\bf 
surface defects in Dijkgraaf-Witten theories   
\end{tabular}\vskip 2.5em
                                     
  ~J\"urgen Fuchs\,$^{\,a}$,~
  ~Christoph Schweigert\,$^{\,b}$,~
  ~Alessandro Valentino\,$^{\,b}$

\vskip 9mm

  \it$^a$
  Teoretisk fysik, \ Karlstads Universitet\\
  Universitetsgatan 21, \ S\,--\,651\,88\, Karlstad \\[4pt]
  \it$^b$
  Fachbereich Mathematik, \ Universit\"at Hamburg\\
  Bereich Algebra und Zahlentheorie\\
  Bundesstra\ss e 55, \ D\,--\,20\,146\, Hamburg
                   
\end{center}
                     
\vskip 5.3em

\noindent{\sc Abstract}\\[3pt]
Dijkgraaf-Witten theories are extended three-dimensional topological
field theories of Turaev-Viro type. They can be constructed geometrically
from categories of bundles via linearization.
Boundaries and surface defects or interfaces in quantum
field theories are of interest in various applications and provide structural
insight. We perform a geometric study of boundary conditions and surface
defects in Dijkgraaf-Witten theories. A crucial tool is the linearization of
categories of relative bundles.
We present the categories of generalized Wilson lines produced by such a
linearization procedure. We establish that they agree with the
Wilson line categories that are predicted by
the general formalism for boundary conditions
and surface defects in three-dimensional topological field theories
that has been developed in \cite{fusV}.

\newpage

\section{Introduction}\label{intro}

For more than two decades, Dijkgraaf-Witten theories have provided a 
laboratory for new ideas in mathematical physics. 
They form a particularly tractable subclass of 
three-dimensional topological field theories. Since they have a Lagrangian
description in which path integrals reduce to counting measures, they
also serve as toy models for more complicated classes of topological
field theories like Chern-Simons theories.

The defining data of a Dijkgraaf-Witten theory are a finite
group $G$ and a 3-cocycle $\omega\iN Z^3(G,\complexx)$. Given these
data, a clear geometric construction \cite{free2,mort6} describes the theory
in terms of a linearization of categories of spans of $G$-bundles.
In the present paper we extend this approach by a geometric study of 
Dijkgraaf-Witten theories on manifolds with boundaries and defects.
More specifically, we consider the class of boundary conditions
and defects for three-dimensional topological field theories that was 
investigated in \cite{fusV}. Besides providing new structural insight, such 
boundary conditions and surface defects are relevant to various 
applications, ranging from a geometric visualization of the TFT approach to 
RCFT correlators to universality classes of 
gapped boundaries and defects in condensed matter systems 
that are of interest in many areas.

A crucial input in our construction are the concepts of relative manifolds 
and relative bundles. Via the linearization of relative bundles we obtain
categories of generalized Wilson lines for Dijkgraaf-Witten theories with
boundaries and defects. Our results perfectly match the general analysis of 
\cite{fusV}, combined with Ostrik's explicit description \cite{ostr5}
of module categories over the categories of $G$-graded vector spaces.

The rest of this paper is organized as follows.
In Section \ref{s:2} we collect pertinent background information. We 
start in Section \ref{ss2.1} with a summary of the geometric construction 
of Dijkgraaf-Witten theories, with emphasis on the implementation of 
locality, which naturally leads to the use of bicategories. We then present
some facts about relative bundles (Section \ref{ssec:mb}), about groupoid 
cohomology (Section \ref{ss:2.4}), and about module categories over the 
monoidal category $\GVect^\omega$ of $G$-graded vector spaces with 
associativity constraint twisted by the cocycle $\omega$ (Section \ref{ss:mod}).

Section \ref{s:3} contains our results for categories of generalized Wilson 
lines in Dijkgraaf-Witten theories with defects and boundaries. 
These categories are associated to one-dimensional manifolds with additional 
data. In the present paper, we restrict our attention to one-dimensional 
manifolds, leaving the case of two-dimensional manifolds with boundaries and of
three-dimensi\-o\-nal manifolds with corners to future work. (The results for
two- and three-dimensional manifolds will allow us to make statements about
generalized partition functions.)
In Section \ref{ss3.1} we discuss the relevant concepts of decorated 
one-dimensional manifolds and of categories of generalized bundles and use 
them to obtain the groupoids for the geometric situations of our interest. 
Afterwards we introduce in Section \ref{ss3.2} the additional data
from groupoid cohomology that are needed for the linearization process.
{}From the perspective of Lagrangian field theory, these data
are a topological bulk Lagrangian and compatible boundary terms; accordingly 
we refer to them as Lagrangian data.  In section \ref{ss3.3} we explain how 
to get 2-cocycles for the groupoids obtained in Section \ref{ss3.1} from 
Lagrangian data assigned to intervals and circles.

Invoking fusion of defects, all one-dimensional manifolds arising
from boundaries and defects can be reduced to two building
blocks: the interval without interior marked points, and the circle
with a single marked point. The linearization of the groupoids for
these two basic situations is described in detail in Section \ref{ss:int}
and \ref{ss:cir}, respectively. A convenient tool in these calculations is a 
graphical calculus for groupoid cocycles which is inspired by \cite{willS}. 
It is introduced in Section \ref{ss:gcc}. Another input is a concrete 
description of the transparent surface defect; this is obtained in Section 
\ref{ss:tpd}, based crucially on the invariance of the graphical calculus 
under Pachner moves.

In the considerations in Sections \ref{ss:int} and \ref{ss:cir} we concentrate 
on the situation that the relevant group homomorphisms are subgroup embeddings;
these lead to indecomposable module categories over $\GVect^\omega$.
Without this restriction, one obtains decomposable module categories;
this is discussed in the Appendix.

\section{Background material}\label{s:2}

In this section we summarize some background material on
the geometric construction of Dijkgraaf-Witten theories and on
boundaries and surface defects in three-dimensional topological field theories,
and on some aspects of relative bundles.

We fix the following conventions.
By $\Vect_\ko$ we denote the category of \findim\ vector spaces over a field 
\ko; In the present paper we only consider the case of complex vector spaces, 
$\ko \eq \complex$. 
A group is assumed to be finite. Manifolds, including manifolds with boundaries 
and manifolds with corners, are smooth.

For a finite group $G$ and a smooth manifold $X$ of any dimension, we denote 
by $\Bun_G(X)$ the category of smooth $G$-principal bundles, which has maps 
covering the identity as morphisms. We adopt the convention that the $G$-action 
on the fiber of a principal $G$-bundle is a right action.  In particular, a 
$G$-bundle over a point is just a right $G$-torsor. Morphisms of the category 
$\Bun_G(X)$ are morphisms of $G$-bundles covering the identity. They are all 
invertible, i.e.\ $\Bun_G(X)$ is a groupoid. Diffeomorphisms $f\colon X \To Y$ 
relate the groupoids by pullback functors, $f^*\colon \Bun_G(Y)\To \Bun_G(X)$. 
We note that with respect to e.g.\ surjective submersions, $\Bun_G$
becomes a stack on the category of smooth manifolds; we will not
use the language of stacks in this paper, though.

\subsection{The geometric construction of Dijkgraaf-Witten theories}\label{ss2.1}

A classic definition by Atiyah characterizes $d$-dimensional
topological field theories as symmetric monoidal functors from a geometric 
category, the symmetric monoidal category $\cobord_{d,d-1}$ of $d$-dimensional
cobordisms, to some linear category, e.g.\ to the symmetric monoidal category 
\Vectc. A classic result states that for $d \eq 2$ the functor given by
  \be
  \tft \,\longmapsto\, \tft(S^1)
  \ee
is an equivalence between the category of topological field
theories and the category of complex commutative Frobenius algebras.

Dijkgraaf-Witten theories are three-dimensional topological field theories.
The Dijkgraaf-Wit\-ten theory
  \be
  \tft_G :\quad \cobord_{3,2} \,\to\, \Vectc
  \ee
based on a finite group $G$ can be characterized as follows. The functor
$\tft_G$ associates to a closed oriented surface $\Sigma$ the
vector space $\tft_G(\Sigma)$ freely generated by the set of isomorphism 
classes of principal $G$-bundles on $\Sigma$.  To a cobordism
  \be
  \xymatrix @R-8pt {
  & M & \\ \Sigma\ar[ur]&&\ar[ul]\Sigma'
  } \ee
it associates a linear map $\tft_G(\Sigma)\To \tft_G(\Sigma')$ whose 
matrix element for principal $G$-bundles $P$ on $\Sigma$ and $P'$ on 
$\Sigma'$ is the number $|\Bun_G(M,P,P')|$. Here $\Bun_G(M,P,P')$ is 
the groupoid of $G$-bundles on $M$ that restrict to a given $G$-bundle $P$ 
on $\Sigma$ and to $P'$ on $\Sigma'$, and  for any groupoid  $\Gamma$ 
we denote by $|\Gamma|$ the groupoid cardinality, which is
the rational number
  \be
  |\Gamma|:= \sum_{\gamma\in \pi_0^{}(\Gamma)} \frac1{|{\rm Aut}(\gamma)|}
  \ee
obtained by summing over the set $\pi_0(\Gamma)$ of isomorphism classes
of objects of $\Gamma$.
 
The introduction of $d{-}1$-dimensional manifolds can be seen as a first
step towards implementing locality in topological field theories:
These submanifolds can be used to cut the $d$-dimensional manifold into 
smaller and simpler pieces, which are manifolds with boundary.
The boundaries of cobordisms are thus to be thought of as ``cut-and-paste
boundaries''. They must not be mixed up with physical boundaries
to be discussed in section \ref{sec:b+defTFT}.

Our analysis uses a framework which goes one step further in the
implementation of locality and naturally leads to the use of bicategories.
We need the following concepts:

\begin{Definition} ~
\\[2pt]
(i)~
The bicategory 2-\Vectc\ of complex 2-vector spaces is the bicategory of
\complex-linear finitely semisimple abelian categories. 
\\
The Deligne product of \complex-linear categories endows this
bicategory with the structure of a symmetric monoidal bicategory.
\\[2pt]
(ii)~
The symmetric monoidal category $\cobord_{3,2,1}$ has as objects
compact oriented smooth \onedim\ manifolds. 1-morphisms are two-dimensional
manifolds with boundary; 2-morphisms are three-manifolds with corners,
up to diffeomorphisms preserving the orientation and the boundary.
(For brevity we suppress collars in our discussion.)
\\[2pt]
(iii)~
An extended three-dimensional topological field theory is a symmetric 
monoidal functor
  \be
  \tft:\quad \cobord_{3,2,1} \,\to\, 2\mbox{-}\Vectc \,.
  \ee
\end{Definition}

We note that, as a consequence of the axioms, 
  \be
  \tft(S\,{\sqcup}\, S') \,\cong\, \tft(S)\boxtimes \tft(S')
  \ee
for any pair $(S,S')$ of one-dimensional manifolds, and
$\tft(\emptyset) \eq \Vectc$, where $\emptyset$ is considered
as a one-dimensional manifold and monoidal unit of $\cobord_{3,2,1}$.

The Dijkgraaf-Witten theory based on a finite group $G$ is
in fact an extended topological field theory \cite{free2,mort6}. 
It assigns to a one-dimensional manifold $S$ the category
  \be
  \tft_G(S) := [ \Bun_G(S),\Vectc ] \,.
  \ee
Here by $[\,\C_1 \,,\C_2 \,]$ we denote the category of
functors between two (essentially small) categories $\C_1$ and $\C_2$.

This formula already gives a hint on the general construction of
the theory: In a first step, one uses the functor $\Bun_G$ 
that associates to a smooth manifold the groupoid of $G$-bundles
to construct a bifunctor
  \be
  \xymatrix{
  \cobord_{3,2,1} \ar[rr]^{\tBun_G} &&
  \mathrm{SpanGrp}
  } \ee
to a bicategory of spans of groupoids. In a second step one
linearizes by taking functor categories with values in \Vectc,
  \be
  \xymatrix{\tft_G:\quad
  \cobord_{3,2,1} \ar[rr]^{~~~~~\tBun_G} &&
  \,\mathrm{SpanGrp}\, \ar[rr]^{[-,\Vectc]~} && \,2\mbox{-}\Vectc \,.
  }
  \ee
The non-extended topological field theory can be obtained from this extended 
topological field theory by restricting to the endomorphism categories of the 
monoidal units of $\cobord_{3,2,1}$ and $2\mbox{-}\Vectc$, since
$\End_{\cobord_{3,2,1}}(\emptyset) \,{\cong}\, \cobord_{3,2}$ and
$\End_{2\mbox{\footnotesize-}\Vectc}(\Vectc) \,{\cong}\, \Vectc$.
 
The fact that $\tft_G$ involves pure counting measures amounts to considering 
vanishing Lagrangians. Dijkgraaf and Witten \cite{diwi2} introduced
the following generalization, in which the linearization is only projective.
Select a cocycle $\omega$ representing a class $[\omega]\iN H^3(G,\complexx)$ 
in group cohomology. One may think about this class as a 2-gerbe \cite{willS}
on the classifying space $BG$
of $G$-bundles, which we represent by the action 
groupoid $* \srr G$ of $G$ acting on a single object $*$.
A $G$-bundle on a 3-manifold $M$ corresponds to a map into this classifying 
space. Pulling back the 2-gerbe along this map to $M$ we get a 2-gerbe on
$M$, which for dimensional reasons is trivial. It therefore gives rise to
a 3-manifold holonomy, which should be seen as the value of a topological 
Lagrangian. For this reason, we refer to the cocycle $\omega$ (and later 
on to similar quantities) as a Lagrangian datum.

The second step of the construction of Dijkgraaf-Witten models consists of 
a linearization of the groupoids obtained in the first step. In general, 
such a linearization is only projective. The relevant 2-cocycle on the 
groupoids must be derived from the Lagrangian data. In the case at hand, 
the 3-cocycle $\omega$ can be transgressed \cite{willS} to a cocycle
$\tau(\omega)$ representing a class in $H^2(G \srr G,\complexx)$, the
groupoid cohomology for the action groupoid $G \srrad G$ with $G$
acting on itself by the adjoint action.

Direct calculation now yields \cite{mort6} 
$\tft_G(S^1) \eq \cald^\omega(G)\Mod$, i.e.\ the category associated to the 
circle is the modular tensor category of modules over the twisted Drinfeld 
double \cite{dipR} of the category of $G$-graded vector spaces -- or, 
equivalently, of complex representations of the finite group $G$.
This category is the category of bulk Wilson lines. The goal of the present 
paper is to generalize this construction to more general cobordism categories 
and to consistently obtain categories of generalized Wilson lines:
both bulk and boundary Wilson lines. Our construction requires the use of 
more general categories of bundles on smooth manifolds.

\subsection{Boundaries and defects in three-dimensional TFT}\label{sec:b+defTFT}

The structure of boundary conditions in two-dimensional topological field 
theories is well understood \cite{laPf,moSe} in the framework of open/closed 
topological field theories.  In this setting one considers a larger cobordism
category $\cobord_{2,1}\opcl$. Its objects are one-dimensional smooth 
manifolds with boundary, with a suitable boundary condition fixed for each 
connected component of the (physical) boundary. Morphisms are now cobordisms 
with boundary, with each boundary component partitioned into segments each of 
which is either a physical boundary or a cut-and-paste boundary. 
An open/closed topological field theory is then a symmetric monoidal 
functor $ \cobord_{2,1}\opcl \To \Vectc$. It turns out that a boundary 
condition $a$ gives rise to a (not necessarily commutative) Frobenius algebra 
$W_a$ whose center is the commutative Frobenius algebra $\tft(S^1)$. Explicitly,
a boundary condition is thus a pair consisting of a Frobenius algebra
$W_a$ and an isomorphism 
  \be
  \tft(S^1) \stackrel\cong\longrightarrow Z(W_a)
  \ee
of commutative associative algebras.
Once such a Frobenius algebra $W_a$ has been determined, the category
of boundary conditions can be described as the category $W_a\Mod$. 

We pause for two comments. First, we allow for point-like insertions on
boundaries that separate possibly different boundary conditions. As a 
consequence, boundary conditions form a category rather than a set: The space
$\Hom_{W_a\Modf}(M_c,M_d)$ of morphisms between two boundary conditions
$M_c,M_d\iN W_a\Mod$ is the vector space of labels for insertions that 
separate the boundary condition $M_c$ from the boundary condition $M_d$.
Second, distinguishing one boundary condition in the discussion could be 
avoided, but at the price of using a higher-categorical language: the one
of module categories over \Vectc. For the three-dimensional topological field 
theories we are interested in, a Morita invariant treatment would amount to 
working with three-categories; we prefer an approach that avoids this. For 
a more detailed analysis of two-dimensional open/closed topological field 
theories we refer to the literature, in particular to \cite{laPf}.

Once one allows for manifolds with boundary, codimension-one
defects that partition a manifold into cells supporting possibly
different topological field theories are a natural extension of the picture 
described above. For two-dimensional theories such defects provide a lot of 
additional insight, in particular about symmetries and dualities \cite{ffrs5}.

\medskip

In three-dimensional topological field theories, boundary conditions and defects 
have been studied only recently. In this case, codimension-one defects are
surface defects. Boundaries and surface defects in three-dimensional 
topological field theories of Reshetikhin-Turaev type appear in a geometric 
interpretation \cite{kaSau3} of the TFT approach \cite{scfr2} to RCFT 
correlators and as models for universality classes of gapped boundaries and 
gapped interfaces for topological phases 
(see e.g.\ \cite{kiKon,waWe,levi4,bajQ4,kapus9}), which arise for instance 
in the study of $2{+}1$-dimensional electron fluids,
including certain fractional quantum Hall states.

A model-independent study of boundary conditions and surface
defects in such theories \cite{fusV} yields the following results, which 
can be regarded as a categorified version of the results in two dimensions
described above. To any boundary condition $a$ there is associated a fusion 
category $\calw_a$. It describes boundary Wilson lines, i.e.\ Wilson lines that 
are confined to the boundary with boundary condition $a$. Let us recall that, 
depending on the chosen formalism, Wilson lines are embedded ribbons or tubes 
with a marked line at the boundary of the tube. In a similar spirit, boundary 
Wilson lines should be described by half-tubes extending into the 
three-dimensional bulk, as illustrated by the following picture:
  \eqpic{Figure13+Figure14}{420}{69} {
  \put(-15,0) {\Includepic{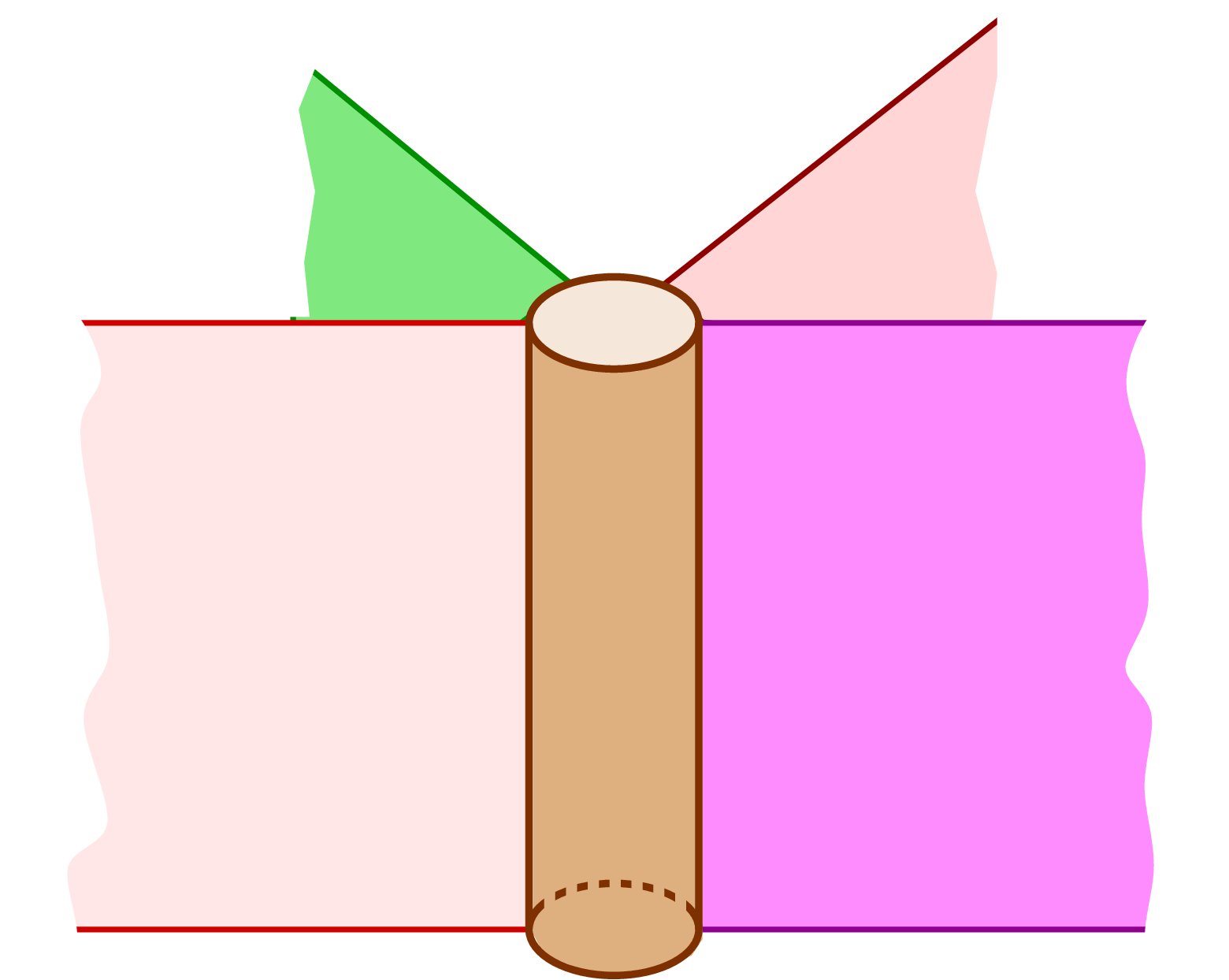}
  }
  \put(205,0) {\Includepic{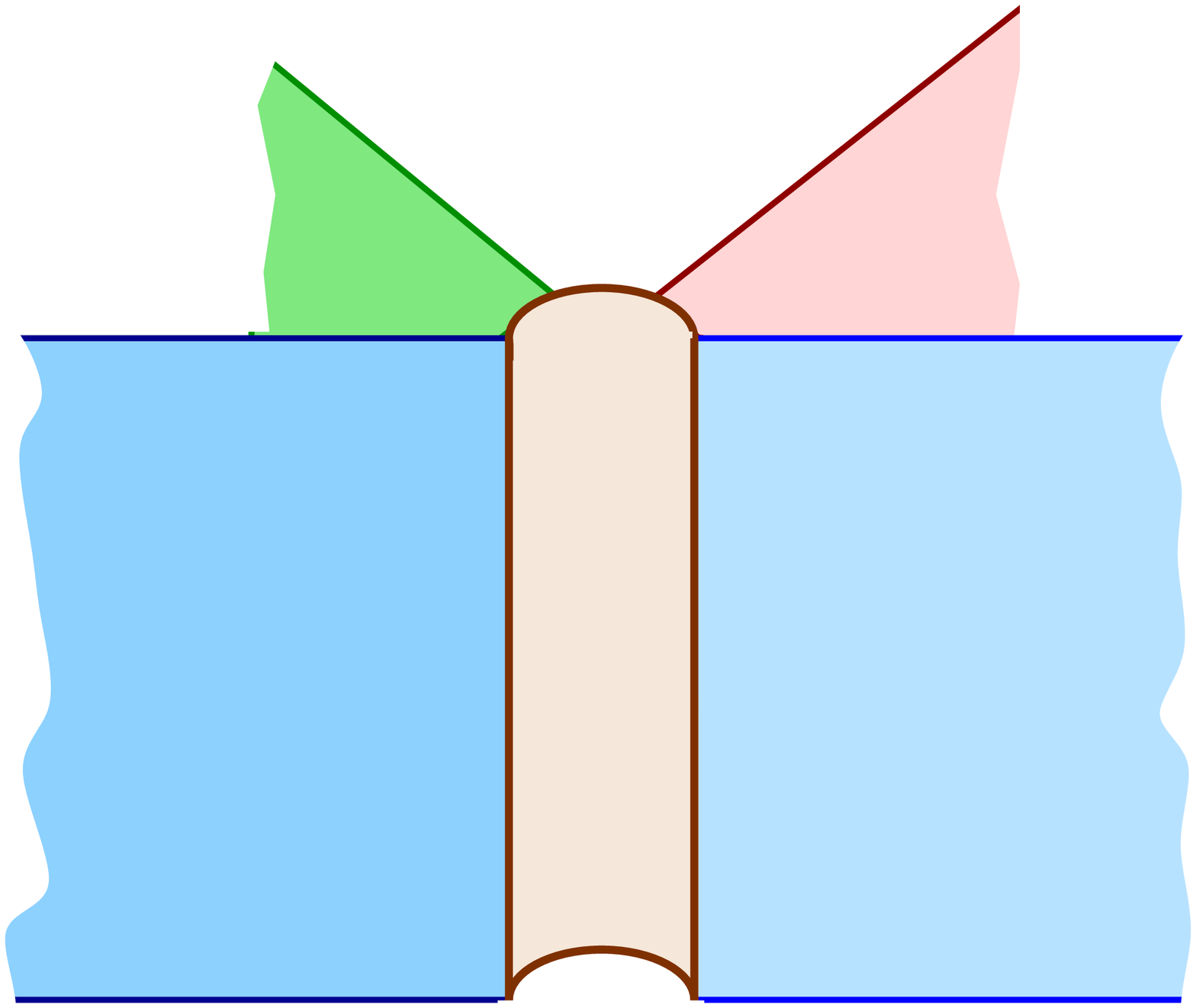}
  \put(45,82)   {$ a $}
  \put(63,120)  {$ d $}
  \put(141,123) {$ d' $}
  \put(149,82)  {$ a' $}
  } }
Here the figure on the right shows a boundary Wilson line in the form of a 
half-tube separating two (possibly different) boundary conditions $a$ and 
$a'$ and at which two surface defects $d$ and $d'$ end, while the
left figure shows a bulk Wilson line in the form of a tube at which four 
surface defects end.

Since boundary Wilson lines are objects in a two-dimensional theory, the 
category $\calw_a$ is not braided. A boundary condition can now be defined as 
a pair consisting of a fusion category $\calw_a$ and a braided equivalence
  \be
  \calc = \tft(S^1) \stackrel\simeq\longrightarrow \Z(\calw_a) \,,
  \labl{Wtriv}
where $\Z$ denotes the Drinfeld center of the fusion category $\calw_a$, 
which is a braided monoidal category. We refer to an equivalence of the type 
\erf{Wtriv} as a \emph{Witt trivialization} of \C. One should note that not 
any braided category is equivalent to a Drinfeld center. In general 
three-dimensional topological field theories this is a source of
obstructions. But in the case of Dijkgraaf-Witten theories the relevant modular 
tensor category $\calc$ indeed is a Drinfeld double, namely the Drinfeld double 
of the fusion category $\GVect^\omega$ of $G$-graded vector spaces 
with associator twisted by $\omega$ (see Section \ref{ss:mod})
  \be
  \calc = \Z(\GVect^\omega)\,.
  \ee
As a consequence, in the case of our interest the existence of boundary 
conditions is not obstructed. 

The collection of all boundary conditions now has the structure
of a bicategory: the bicategory of all module categories over the fusion 
category $\calw_a$. (Module categories over a fusion category are a 
categorification of the notion of a module over a ring; we refer to \cite{ostr} 
for details.) The category of boundary Wilson lines separating two boundary 
conditions $c$ and $d$ that are given by two $\calw_a$-module categories 
$\calm_c$ and $\calm_d$, respectively, is the abelian $\complex$-linear category
  \be
  \Fun_{\calw_a\Modf}(\calm_c,\calm_d)
  \ee
of $\calw_a$-module functors.

A similar analysis can be carried out for surface defects that separate
two topological field theories of Reshetikhin-Turaev type, which are labeled
by modular tensor categories $\calc_1$ and $\calc_2$. The category
of Wilson lines in a surface defect of type $d$ is now a fusion
category $\calw_d$ together with a braided equivalence
  \be
  \calc_1^{}\boxtimes\calc_2^\rev \stackrel\simeq\longrightarrow
  \Z(\calw_d) \,.
  \ee
Since the modular categories relevant for Dijkgraaf-Witten theories are 
already Drinfeld centers themselves, the existence of surface defects between 
any two Dijkgraaf-Witten theories is not obstructed.
The category of Wilson lines separating surface defects that are given by
two $\calw_d$-module categories $\calm_c$ and $\calm_d$, respectively,
is the abelian \complex-linear category
  \be
  \Fun_{\calw_a\Modf}(\calm_c,\calm_d)
  \ee
of $\calw_d$-module functors.

In the special case of defects separating a modular tensor category \C\
from itself, we can work with the canonical Witt trivialization
  \be
  \mathrm{can}:\quad \C\boti\C^\rev_{} \stackrel\simeq\longrightarrow
  \Z(\calc) \,.
  \labl{Wittcan}
This functor maps the object $U\boti V \iN \C\boti\C^\rev_{}$ to the object 
$U\oti V\iN\C$ endowed with a half braiding $e_{U\otimes V}^{}$ given by
\cite[Eq.\,(4.2)]{etno2}
  \be
  e_{U\otimes V}^{}(X): \quad U\oti V\oti X \stackrel{c^{-1}}\longrightarrow 
  U\oti X\oti V \stackrel{c}\longrightarrow X\oti U\oti V .
  \ee
With respect to the canonical Witt trivialization \erf{Wittcan}, we describe 
a defect separating \C\ from itself by a \C-module category. Now \C\ has a
natural structure of module category over itself. This specific \C-module 
category describes a particularly important surface defect, the 
\emph{transparent} (or \emph{invisible}) surface defect. In fact, one expects 
a notion of a fusion product of defects, so that the bicategory of surface 
defects is even a monoidal bicategory. The transparent defect is then the 
tensor unit of the monoidal bicategory of defects. 
(At one step lower in the categorical ladder,
the tensor unit of the monoidal category of 
endofunctors of any given defect category describes a Wilson line that is 
invisible inside the surface. The category of endofunctors of \C\ describes 
Wilson lines inside the transparent defect; these are ordinary bulk Wilson 
lines. In particular, the tensor unit of this monoidal
category is the invisible bulk Wilson line.)

\medskip

Our goal in this paper is to achieve a concrete geometric, Lagrangian
construction of some of the categories describing Wilson lines in the 
presence of boundaries and surface defects in Dijkgraaf-Witten theories
in the spirit of \cite{free2,mort6}.
To this end, we need the appropriate geometric objects that
form categories whose linearizations enter in the topological field theory.

\subsection{Relative bundles}\label{ssec:mb}

In this section we review the notion of a relative bundle. We restrict our 
attention to finite groups, which is sufficient for our construction.

\begin{Definition} ~\\[1pt]
Let $G$ and $H$ be finite groups, $\iota\colon H\To G$ a morphism of 
finite groups, and $X$ a smooth manifold, Then the functor
  \be
  \Ind_\iota:\quad \Bun_H(X) \to \Bun_G(X)
  \ee
is the one that acts on objects as $P_H \,{\mapsto}\, P_H\,{\times_{\!H}}\, G$.
\end{Definition}

\begin{rem}\label{rem2.3}~\\[2pt]
(i)~
If the group homomorphism $\iota$ injective, then the functor $\Ind_\iota$ is 
injective on morphisms. 
\\[2pt]
Indeed, suppose $f_1,f_2\colon P_H\To P_H'$ are two different morphisms
of $H$-bundles on $X$. Then there exist points $x\iN X$ and
$p$ in the fiber of $P_H$ over $x$ such that
$f_1(p) \,{\ne}\, f_2(p)$. Since both $f_1(p)$ and $f_2(p)$ are in the fiber
of $P_H'$ over $x$, we have a unique $h\iN H\,{\setminus}\,\{e\}$
such that $f_1(p) \eq f_2(p).h$. Suppose that after induction
$[f_1(p),g] \eq [f_2(p),g]$ for some $g\iN G$. Then
  \be
  [f_1(p),g] = [f_2(p),g] = [f_1(p).h,g] = [f_1(p),\iota(h)\cdot g] \,.
  \ee
Equality of the left and right hand sides implies $\iota(h)\cdot g \eq g$, i.e.\
$\iota(h)\eq e$. If $\iota$ is injective, this is impossible for $h\,{\ne}\, e$.
\\[2pt]
(ii)~
Induction commutes with pullback: if $f\colon X_1\To X_2$ is a morphism of
smooth manifolds and if $P_H^{(2)}$ is a $H$-bundle on $X_2$, then
  \be
  \Ind_\iota\, f^*P_H^{(2)} = f^*\, \Ind_\iota P_H^{(2)} .
  \ee
More abstractly, for any finite group $G$ we have the stack $\Bun_G(-)$ of 
$G$-bundles on the category of smooth manifolds with topology given by 
surjective submersions. Induction is also compatible with descent. Thus
$\Ind_\iota$ gives a morphism $\Ind_\iota\colon \Bun_H\To\Bun_G$ of stacks.
\end{rem}

A crucial ingredient for our construction is the notion of relative smooth 
manifolds and relative bundles.  This is as follows, see e.g.\ \cite{STee}.
 
\begin{Definition}~\\[2pt]
(i)~
A \emph{relative} (smooth) \emph{manifold} $Y\,{\stackrel j\to}\, X$ 
consists of a pair $Y,X$ of 
smooth manifolds and a morphism $j\colon Y \To X $ of smooth manifolds. 
 \\
A morphism $(Y_1 {\stackrel{j_1^{}}\to} X_1)\,{\longrightarrow}\, 
(Y_2{\stackrel{ j_2^{}}\to} X_2)$ of  relative
smooth manifolds is a commuting diagram
  \be
  \xymatrix{
  Y_1\ar_{j_1^{}}[d]\ar^{f_Y^{}}[r] & Y_2\ar^{j_2^{}}[d]\\
  X_1\ar_{f_X^{}}[r]&X_2
  } \ee
in the category of smooth manifolds.
\\[2pt]
(ii)~
Let $\iota\colon H \To G$ be a homomorphism of finite groups. A \emph{relative 
$(G,H)$-bundle} on the relative manifold $Y\,{\stackrel j\to}\, X$ 
is a triple consisting of a $G$-bundle $P_G$ on $X$,
an $H$-bundle $P_H$ on $Y$, and an isomorphism
  \be
  \alpha:\quad \Ind_\iota(P_H) \stackrel\simeq\longrightarrow j^*(P_G)
  \ee
of $G$-bundles on $Y$.
\\[2pt]
(iii)~
A morphism $(P_G,P_H,\alpha)\To (P'_G,P'_H,\alpha')$ of relative 
$(G,H)$-bundles on a relative smooth manifold $Y\,{\stackrel j\to}\, X$
consists of a morphism
  \be
  \varphi_G^{}:\quad P_G \to P_G'
  \ee
of $G$-bundles on $X$ and of a morphism
  \be
  \varphi_H^{}:\quad P_H \to P_H'
  \ee
of $H$-bundles on $Y$ such that the diagram 
  \be
  \xymatrix  @R+8pt {
  \Ind_\iota(P_H)\ar^{~~~\alpha}[r]\ar_{\Ind_\iota\varphi_H^{}}[d]
  & j^*P_G\ar^{j^*\varphi_G^{}}[d] \\
  \Ind_\iota(P_H')\ar^{~~~\alpha'}[r] & j^*P'_G
  }
  \label{morphdia}\ee
of morphisms of $G$-bundles on $Y$ commutes. 
\\[2pt]
The category of relative $(G,H)$-bundles on $(X,Y)$ is denoted
by $\Bun_{(G,H)}(Y{\to} X)$.
\end{Definition}

\begin{rem}~\\[2pt]
(i)~
The category $\Bun_{(G,H)}(Y{\to} X)$ depends the group homomorphism
$\iota\colon H\To G$. The notation $\Bun_{(G,H)}(Y{\to} X)$ suppresses 
this dependence and is thus slightly inappropriate. 
\\[2pt]
(ii)~
The category $\Bun_{(G,H)}(Y{\to} X)$ inherits from the category
of principal bundles the property of being a groupoid:
all morphisms of relative bundles are invertible.
\\[2pt]
(iii)~
For the special case that $j\eq \id_X$ is the identity on $X \eq Y$, we
obtain the notion of a reduction of a $G$-bundle to an $H$-bundle 
along the group homomorphism $\iota$.
\\[2pt]
(iv)~
As an object, a relative bundle is thus a $G$-bundle $P_G$ on $X$ together 
with a  reduction of its pullback $j^*P_G$ to an $H$-bundle along the group 
homomorphism $\iota$. One should note, however, that the {\em morphisms} 
in $\Bun_{G,H}(X,Y)$ are \emph{not} simply morphisms of reductions, which 
would only involve a morphism of $G$-bundles on the manifold $Y$. Rather, 
also a $G$-morphism on the manifold $X$ is required. (Later on, $Y$ will 
typically be a submanifold of $X$; hence we require a morphism on a larger 
manifold in that case.) In gauge theory terminology, the morphisms are 
thus gauge transformations on $Y$ and on $X$, respectively.
\\[2pt]
(v)~
If the group homomorphism $\iota$ is injective, then by remark \ref{rem2.3}(i)
the morphism $\varphi_H$ of $H$-bundles is determined uniquely by $\varphi_G$, 
provided it exists. It is thus not an extra datum. The morphisms of relative 
$(G,H)$-bundles are in this situation morphisms of $G$-bundles 
that are compatible with the reductions.
\\[2pt]
(vi)~
Fix a homomorphism $\iota\colon H\To G$ of finite groups and consider a 
relative bundle $(P_G^2,P_H^2,\alpha^2)$ on the relative manifold 
$Y_2\,{\stackrel{j_2^{}}\to}\, X_2$. 
We define a pullback of relative bundles along the morphism
  \be
  \xymatrix{
  Y_1\ar_{j_1^{}}[d]\ar^{f_Y^{}}[r] & Y_2\ar^{j_2^{}}[d]\\
  X_1\ar_{f_X^{}}[r]&X_2
  }
  \ee
of relative manifolds. Since induction and pullback commute 
by remark \ref{rem2.3}(ii), we have a canonical isomorphism 
  \be
  \Ind_\iota(f_Y^* P_H^2)\cong f_Y^*\, \Ind_\iota P_H^2
  \ee
of bundles. Noting that $f_X\cir j_1 \eq j_2\cir f_Y$, we also have another 
isomorphism 
  \be
  j_1^*f_X^* P_G^2 \cong f_Y^*j_2^* P_G^2
  \ee
of $G$-bundles, and thus an isomorphism
  \be
  f_Y^*(\alpha):\quad \, \Ind_\iota(f_Y^* P_H^2) \to 
  f_Y^*\Ind_\iota P_H^2 \to f_Y^*f_2^*P_G^2 \to j_1^*f_X^* P_G^2 
  \ee
of $G$-bundles on $Y_1$. Hence $(f_X^*P^2_G,f_Y^*P_H^2, f_Y^*(\alpha))$ is a 
relative $(G,H)$-bundle on $(X_1,Y_1)$.
\\
We have thus a bifunctor $\Bun_{\iota:\, H\to G}$ from the category opposite
to the category of relative manifolds to the bicategory of groupoids, 
i.e.\ a prestack $\Bun_{(G,H)}$ on the category of relative manifolds.
\end{rem}

It should be appreciated that we do not require the group homomorphism 
$\iota\colon H \To G$ to be injective.  For later use, we will consider 
two examples.

\begin{Example}\label{ex1} ~\\[2pt]
Consider the case that $X\eq Y$ is a point. Bundles are then torsors 
$\underline H$ and $\underline G$, respectively, which are unique up to
isomorphism. The additional datum characterizing a relative bundle
is then an isomorphism 
  \be
  \alpha:\quad \underline H\times_{\!H} G \stackrel \cong\longrightarrow
  \underline G
  \ee
of torsors. If we fix base points $*_{\!H}\iN\underline H$ and
$*_G\iN\underline G$, then $\alpha$ is determined by the group element
$\gamma_\alpha\iN G$ such that $\alpha([*_{\!H},e]) \eq *_G\,.\gamma_\alpha$. 
\\[2pt]
Morphisms $(\underline G,\underline H,\alpha) \To (\underline G',
\underline H',\alpha')$ are pairs of morphisms 
$\varphi_H^{}\colon \underline H\to \underline H'$ and 
$\varphi_G^{}\colon \underline G\to \underline G'$ of torsors. Using the base 
points $*_{\!H}^{}$ and $*_{\!H}'$ of $\underline H$ and $\underline H'$,
respectively, and similarly base points of the $G$-torsors,
morphisms are described by group elements $g\iN G$ and $h\iN H$ such that
  \be
  \varphi_H^{}(*_{\!H}ее{})= *_{\!H}'\,.\, h \qquad\text{and}\qquad
  \varphi_G^{}(*_G^{})= *_{G}'\,.\, g \,.
  \ee
The commuting diagram (\ref{morphdia}) requires that
  \be
  \varphi_G^{}(\alpha[*_{\!H}^{},e]) = \varphi_G^{}(*_G^{}.\gamma_\alpha)
  = *_{G}'\,.\,(g\gamma_\alpha)
  \ee
equals
  \be
  \alpha'(\Ind_\iota\varphi_H^{}([*_{\!H}^{},e]) = \alpha'([*_{\!H}'h,e])
  = \alpha'([*_{H}',\iota(h)]) = *_{G}'\,.\,(\gamma_{\alpha'}\iota(h)) \,.
  \ee
We thus find the condition 
  \be
  g\, \gamma_\alpha = \gamma_{\alpha'}\, \iota(h)
  \ee
on the pair $(g,h)$ of group elements. As expected, for $\iota$ injective, 
this determines $h$ in terms of $g$. Moreover, given any two relative bundles, 
we can always find group elements $g$ and $h$ such that this relation holds. 
So there is a single isomorphism class of objects. In particular,
we can restrict our attention to just one $H$-torsor $\underline H$ and
one $G$-torsor $\underline G$. Then we get a category with objects labeled 
by $\gamma_\alpha\iN G$ and morphisms being pairs $(g,h)$ such that 
$g\gamma_\alpha \eq \gamma_{\alpha'} \iota(h)$, 
or put differently, the action groupoid
  \be
  G \sll G \srr_{\iota^-} H \,.
  \ee
Here the notation is as follows. We deal with left actions for both $G$ and 
$H$. The left action of the group $G$ is simply left multiplication, while 
the left action of $H$ is right multiplication after applying the group 
homomorphism $\iota$ and taking the inverse, i.e.\ 
$(g,h).\gamma= g\cdot\gamma\cdot \iota(h)^{-1}$.
\end{Example}

\begin{Example}\label{ex1b} ~\\[1pt]
Take for $X$ a closed interval and for $Y$ the subset consisting of its two 
end points, which we label by $1,2$. Since the interval is contractible
and $G$ is finite, the category of $G$-bundles on $X$
is canonically equivalent to the category of $G$-torsors. Similarly
we have $H_1$- and $H_2$-torsors, one over each end point. We fix one such
torsor for each end point and for the interval itself from now on. We also fix 
base points $*_{\!H_1^{}}$, $*_{\!H_2^{}}$ and $*_G$ for these torsors.
Objects in the category are then pairs $(\gamma_{\alpha,1},\gamma_{\alpha,2}) 
\iN G\Times G$ which describe the morphisms of torsors as
  \be
  \alpha_1([*_{H_1^{}},e]) = *_G . \gamma_{\alpha,1} \qquad \text{and}
  \qquad \alpha_2([*_{H_2^{}},e]) = *_G^{} . \gamma_{\alpha,2} \,.
  \ee
The morphisms are described by triples $(h_1,h_2,g)\iN H_1\Times H_2 \Times G$
satisfying
  \be
  \varphi_{H_1}^{}(*_{\!H_1}) = *_{\!H_1}\,.\, h_1 \,, \quad
  \varphi_{H_2}^{}(*_{\!H_2}) = *_{\!H_2}\,.\, h_2 \quad\text{ and }\quad
  \varphi_G(*_G)^{} = *_G\,.\,g \,.
  \ee
Based on the commuting diagram (\ref{morphdia}), we check when a triple 
$(h_1,h_2,g)$ gives a morphism $(\gamma_{\alpha,1}^{},\gamma_{\alpha,2}^{}) 
\To (\gamma'_{\alpha,1},\gamma'_{\alpha,2})$.  As before we compute 
  \be
  \varphi_G^{}(\alpha_i[*_{\!H_i},e]) = \varphi_G^{}(*_G\gamma_{\alpha,i})
  = *_G\,.\,(g\gamma_{\alpha,i})
  \ee
and
  \be
  \alpha_i'(\Ind_\iota\varphi_H^{}([*_{\!H_i},e]) = \alpha_i'([*_{\!H_i}h_i,e])
  = \alpha'([*_{\!H_i},\iota(h_i)]) = *_G^{}\,.\,(\gamma'_{\alpha, i}.\iota(h_i)) \,.
  \ee
We thus arrive at the equalities
  \be
  g\, \gamma_{\alpha,i} = \gamma'_{\alpha,i}\, \iota(h_i) 
  \ee
for $i\eq 1,2$. Hence the action groupoid is
  \be
  G \sll\, G\Times G \srr_{\!\iota_1^-\times\iota_2^-} H_1 \Times H_2 \,,
  \ee
where the $G$-action is the diagonal one.
\end{Example}

\subsection{Groupoid cohomology and gerbes on groupoids}\label{ss:2.4}

The definition of a Dijkgraaf-Witten theory on a three-manifold requires, 
as an additional datum besides a finite group $G$, the choice of a
3-cocycle $\omega\iN Z(G,\complexx)$. This cocycle enters in the linearization.
We now describe how this 3-cocycle can be seen geometrically as a 2-gerbe on 
the groupoid $* \srr G$.

\medskip

We first give a brief outline of groupoid cohomology. Given a  finite groupoid 
$\Gamma=(\Gamma_0,\Gamma_1)$, consider its nerve, which is a simplicial set
  \be
  \Big(\xymatrix{ {\cdots}~
  \ar@<1.3ex>[r]^{\partial_0} \ar@<0.3ex>[r] \ar@<-0.7ex>[r]
  \ar@<-1.7ex>[r]_{\partial_3}
  & \,{\Gamma_2}\,
  \ar@<0.9ex>[r]^{\partial_0} \ar@<-0.1ex>[r] \ar@<-1.1ex>[r]_{\partial_2}
  & \,{\Gamma_1}\,
  \ar@<0.3ex>[r]^{\partial_0} \ar@<-0.7ex>[r]_{\partial_1}
  &
  \,{\Gamma_0} }\Big)=:\Gamma_{\!\bullet} \, ,
  \ee
where for $i \,{\ge}\, 1$, $\Gamma_i$ consists of $i$-tuples of composable
morphisms of $\Gamma$. Applying the functor $\Map(-,\complexx)$ 
and taking alternating combinations of the face maps yields a complex 
  \be
  \Map(\Gamma_0,\complexx) \to \Map(\Gamma_1,\complexx) \to
  \Map(\Gamma_2,\complexx) \to \Map(\Gamma_3,\complexx) \to \,\cdots 
  \labl{Gammacomplex}
of groups. A group $G$ gives rise to the groupoid $*\srr G$ with a single
object. In this case the complex \erf{Gammacomplex}
reduces to the standard bar complex.

It is useful to think about cochains in this complex in a geometric way.

\begin{Definition}\mbox{} \\[1pt]
An $n$-\emph{gerbe} on the groupoid $\Gamma$ is an $(n{+}1)$-cocycle
  \be
  \omega \in Z^{n+1}(\Gamma,\complexx) \,.
  \ee
\end{Definition}

Using standard facts about complexes in small abelian categories one 
deduces that $n$-gerbes on a groupoid $\Gamma$ form an $n{+}1$-category:
 \beginitemize
\item
A $(-1)$-gerbe is an object in degree 0, i.e.\ an
element of the set of objects of $\Gamma$.

\item
A $0$-gerbe consists of a 1-cocycle $\omega\iN Z^1(\Gamma)$.
The morphism sets are
  \be
  \Hom(\omega,\omega')
  = \{\eta\iN \Gamma_0\,|\, \mathrm d\eta= \omega'-\omega\} \,.
  \ee
We thus get a category of 0-gerbes, which we also call
line bundles on $\Gamma$. Its isomorphism classes are
classified by the cohomology group $H^1(\Gamma,\complexx)$. 

\item
$1$-gerbes form a bicategory. Its objects are 2-cocycles, and the set 
of 1-morphisms between two 2-cocycles $\omega$ and $\omega'$ is 
$\{\eta\iN \Gamma_1\,|\, \mathrm d\eta \eq \omega'\,{-}\,\omega\}$.
Given two 1-morphisms $\eta,\eta'\colon \omega\To\omega'$, a 2-morphism 
$\Phi\colon \eta \,{\Rightarrow}\, \eta'$ is an element $\Phi\iN \Gamma_0$ 
satisfying $\mathrm d\Phi \eq \eta' \,{-}\, \eta$.
\\
The isomorphism classes of this bicategory of gerbes are
classified by the cohomology group $H^2(\Gamma,\complexx)$. 
 \end{itemize}

For Dijkgraaf-Witten theories based on a finite group $G$, 2-gerbes
on the groupoid $*\srr G$ are relevant. As we already have pointed out,
they should be thought of as a finite version of a Chern-Simons 2-gerbe.

\subsection{Module categories over the fusion category 
           {\boldmath $\GVect^\omega$}}\label{ss:mod}

We next discuss category-theoretic and algebraic realizations of group 
3-cocycles. A closed $3$-cocycle $\omega$ on a finite group $G$ allows one 
to endow the abelian category $\GVect$ of $G$-graded vector spaces with a 
non-trivial associativity constraint, defined on simple objects by 
  \be
  \begin{array}{rll}
  a_{V_{g_1},V_{g_2},V_{g_3}}:\quad
  \left(V_{g_1}\oti V_{g_2}\right) \otimes V_{g_3}
  &\!\!\to\!\!& V_{g_1}\otimes \left( V_{g_2} \oti V_{g_3}\right)
  \Nxl2
  v_1\otimes v_2\otimes v_3 &\!\!\mapsto\!\!& \omega(g_1,g_2,g_3)\,
  v_1\otimes v_2\otimes v_3 \,.
  \end{array}
  \labl{aomega}
This yields a fusion category, which is denoted by $\GVect^\omega$
(the pentagon axiom is fulfilled because $\omega$ is closed).
Cohomologous 3-cocycles give rise to monoidally equivalent fusion categories.

The modular tensor category relevant for the Dijkgraaf-Witten theory based on 
$(G,\omega)$ is the Drinfeld center $\Z(\GVect^\omega)$. (This has been discussed 
in \cite{dipR}; a helpful more recent exposition is given in \cite{willS}.)
It is thus a topological field theory of Reshetikhin-Turaev type. This
allows us to compare our geometric results with those obtained in the
model independent approach to defects and boundary conditions in \cite{fusV}.     

\medskip
                                 
The indecomposable module categories over the monoidal category
$\GVect^\omega$ have been classified \cite[Example\,2.1]{ostr5}:
Consider a subgroup $H\le G$ and a 2-cochain $\theta$ on $H$ such that 
$\mathrm d \theta \eq \omega|_H$. Note that this requires the restriction 
of $\omega$ to the subgroup $H$ to be exact and thus imposes in general
restrictions on the subgroup. Rephrased in the language of Section \ref{ss:2.4}, 
$\theta$ is a 1-mor\-phism from the trivial 2-gerbe on $*\srr H$ to the pullback 
2-gerbe $\iota^*\omega$.

The twisted group algebra
$A_{H,\theta} \,{:=}\, \complex_\theta[H]$ is then a (haploid special symmetric) 
Frobenius algebra in $\GVect^\omega$. For any 1-cochain $\chi$ on $H$ the 
algebras $A_{H,\theta}$ and $A_{H,\theta+d\chi}$ are isomorphic. Thus, given 
a subgroup $H$ the isomorphism classes of algebras form a torsor over 
$H^2(H,\complexx)$. Indecomposable module categories over $\GVect^\omega$ 
are given by Morita classes of twisted group algebras. They 
are thus in bijection with equivalence
classes of pairs $(H,\theta)$; we denote them by $\calm_{H,\theta}$. 

Actually, any pair consisting of a 
group homomorphism $\iota\colon  H\To G$ and a 2-cochain $\theta$ on $H$ 
such that $\iota^*\omega \eq \rmd\theta$ defines a module category, albeit 
not an indecomposable one unless $\iota$ is injective. For the case that both 
$\omega$ and $\theta$ vanish, this is discussed in the Appendix.

\section{Categories of generalized Wilson lines in Dijkgraaf-Witten theories}
\label{s:3}

We are now ready to discuss Dijkgraaf-Witten theories with boundaries and
defects. Our ultimate goal is 
to consider such a theory as a 1-2-3-extended topological
field theory. Concretely this means:
 \beginitemize
\item
To a decorated smooth oriented one-dimensional manifold, we have to 
assign a finitely semisimple  \complex-linear category.
\\
This category will have the interpretation of a category of (generalized) 
Wilson lines. The one-dimensional manifold is allowed to have boundaries, 
corresponding to physical boundaries of the three-dimensional theories, and 
to have marked points, corresponding to surface defects.
\item
To a decorated smooth oriented two-dimensional manifold we have to assign 
a \complex-linear functor. A two-dimensional manifold can have physical
boundaries and lines corresponding to surface defects. Moreover, it can have
cut-and-paste boundaries which are one-dimensional manifolds of the type
described in the first item. These cut-and-paste boundaries determine the
categories which are the source and target for the functor associated to the
two-manifold.
\item
To a decorated three-manifold with corners, we have to associate
a natural transformation.
\end{itemize}

\subsection{Decorated one-manifolds and categories of generalized bundles}
            \label{ss3.1}

In the present paper we concentrate on examples and restrict
our attention to one-dimensional manifolds.
We should also keep in mind that cut-and-paste boundaries have been
introduced to implement locality. Accordingly we impose the condition
that a cutting is transversal to any additional decoration data such
as surface defects or generalized Wilson lines.

This leaves us with two types of connected one-manifolds only:
 \beginitemize
\item
An interval which is partitioned by finitely many distinct points 
in its interior.
\item
A circle that is partitioned by finitely many distinct points.
\end{itemize}
For the situations shown in \erf{Figure13+Figure14} above, the 
cutting leading to such one-manifolds is indicated in the following picture: 
  \eqpic{Figure17+Figure18}{420}{69} {
  \put(-15,0) {\Includepic{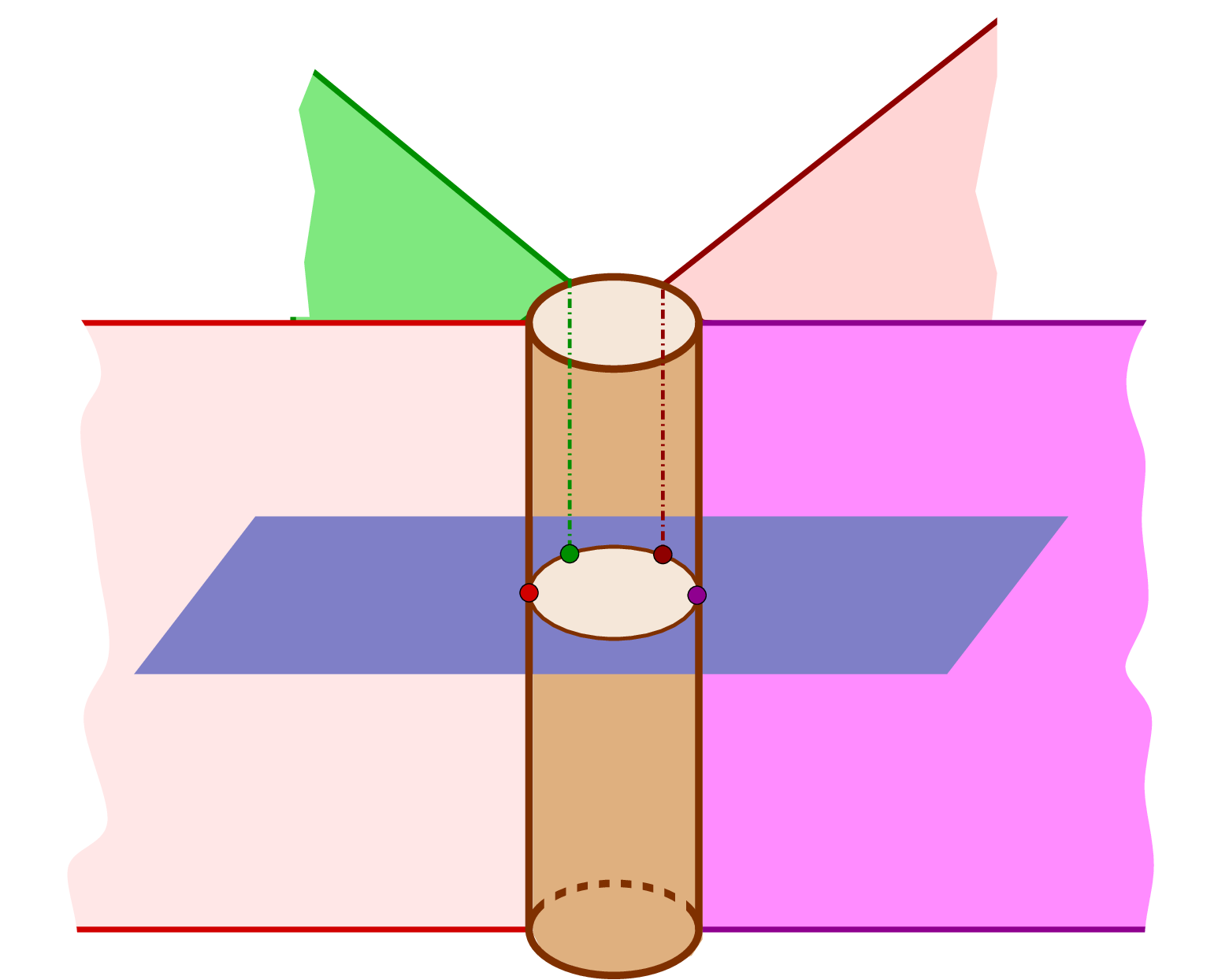}
  }
  \put(205,0) {\Includepic{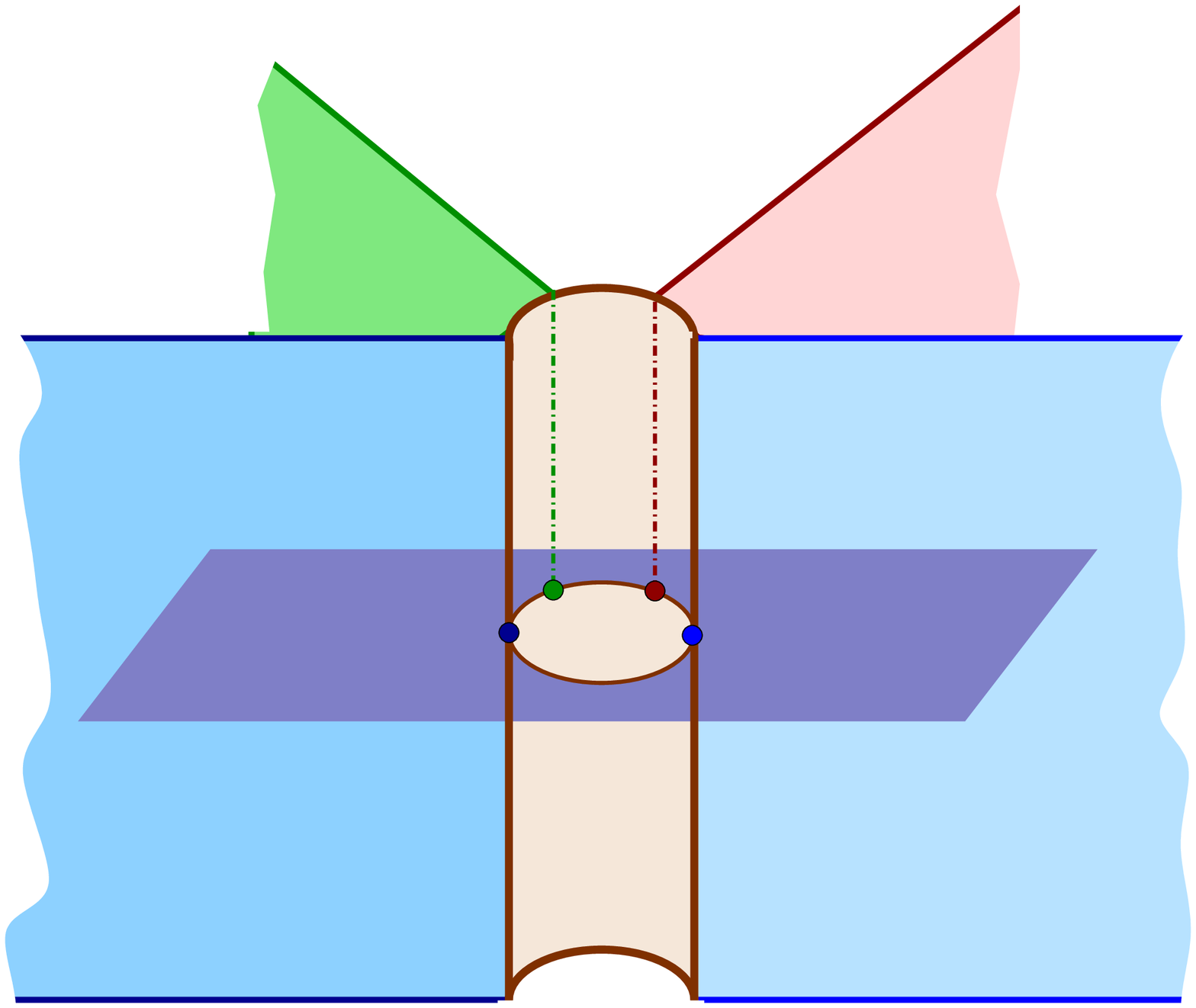}
  } }

Every subinterval of such a one-manifold is decorated by a Dijkgraaf-Witten 
theory. The decoration datum for each subinterval is thus a finite group $G$ 
together with a 3-cocycle $\omega\iN Z^3(G,\complexx)$. The locality of the
geometric construction of Dijkgraaf-Witten theories \cite{free2,mort6} then suggests 
that $G$-bundles on these intervals should appear in our construction.

However, we also must assign data to the end points of a subinterval. Recall 
from Section \ref{ss2.1} that the general construction of Dijkgraaf-Witten
theories consists of two steps: first finding an appropriate stack of bundles, 
leading to spans of groupoids, which then have to be linearized with the help 
of Lagrangian data. In the situation at hand, the relevant categories are 
variants of relative bundles which have been introduced in Section 
\ref{ssec:mb}. In the case of an interval without marked points in the 
interior, the morphism defining the relative manifold is the embedding 
of the end points. 

One might thus pick a group homomorphism $\iota\colon H\To G$ and assign 
$H$-bundles to the two end points. This is, however, not the most general 
situation one can consider -- for complying with locality we must allow for the 
possibility to assign different local conditions to the two end points of the 
interval. Thus we select possibly different groups $H_i$, $i\eq1,2$, and group 
homomorphism $\iota_i\colon H_i\To G$ separately for each end point $p_1,p_2$ 
and consider the following category: an object consists of a $G$-bundle $P_G$ 
over the interval, an $H_1$-bundle $P_{H_1}$ over $p_1$, a morphism
$\Ind_{\iota_1^{}} P_{H_1}\To (P_G)|_{p_1^{}}$ of $G$-bundles on $p_1$, 
an $H_2$-bundle $P_{H_2}$ over $p_2$, and a morphism
$\Ind_{\iota_2^{}} P_{H_2}\to (P_G)|_{p_2^{}}$ of $G$-bundles on $p_2$.

This leads to the following assignment of kinematical data. At the level of 
groups, we associate to an end point of an interval that is labeled by a 
group $G$ a group homomorphism $\iota\colon H\To G$, with $H$ some finite 
group. This prescription still needs to be complemented by group cohomological 
Lagrangian data; these will be introduced in Section \ref{ss3.2}.

\medskip

Example \ref{ex1b} allows us to determine directly a finite action groupoid 
that is relevant for an interval without any marked interior points, labeled 
by a group $G$, and with end points labeled by groups $H_1$, $H_2$ and group 
homomorphisms $\iota_1\colon H_1\To G$ and $\iota_2\colon H_2\To G$
respectively: it is given by
  \be
  G\sll G\Times G \srr_{\!\iota_1^-\times\iota_2^-} H_1 \Times H_2 \,.
  \labl{G-GG-H1H2}
Here $G$ acts from the left as the diagonal subgroup, while $H_1$ is mapped 
via $\iota_1$ to the first copy of $G$ and acts by right multiplication after 
taking the inverse; the action of $H_2$ is analogous, the only difference 
being that it is mapped by $\iota_2$ into the second copy of $G$. 
Let us describe the structure of this groupoid: its set of objects is 
given by a Cartesian product of groups, one factor for each pair consisting
of a marked point and a neighbouring interval. The group is determined by the
interval, since it comes from the morphism of bundles in the corresponding
relative bundle. The morphisms in the groupoid are gauge transformations:
the $G$-action describes gauge transformations of the $G$-bundle on the 
interval and acts by multiplication from the left. The $H_i$-actions are by 
multiplication from the right after having taken the inverse; their origin are 
$H_i$-gauge transformations of the $H_1$-bundles on the respective marked point.

\medskip

This picture generalizes to marked points in the interior, either
of an interval or of a circle. To any such point two intervals are adjacent, 
which are labeled by gauge groups $\Gl$ and $\Gr$, respectively. 
To describe the resulting relative manifold,
consider as an example the closed interval $[0,1]$ with a marked interior 
point $p_1 \,{:=}\, \frac12$. Take for $X$ the disjoint union 
$X \,{:=}\, [0,\frac12] \,{\sqcup}\, [\frac12,1]$. One should appreciate that 
in $X$ the point $p_1$ is ``doubled''. By locality, the category of bundles is 
now defined with separate data for each of the marked points $p_0 \eq 0$, $p_1 
\eq \frac12$ and $p_2 \eq 1$. For $p_0$ and $p_2$ we select again group 
homomorphisms $\iota_0\colon H_0\To \Gl$ and $\iota_2\colon H_2\To \Gr$.
At $p_1$ we take as a datum a finite group $H_1$ and a group homomorphism 
$\iota\colon H_1\To \Gl\Times \Gr$ or, equivalently, a pair of group 
homomorphisms $\iota_{\rm l}\colon H_1\To \Gl$ and 
$\iota_{\rm r}\colon  H_2\To \Gr$.

We consider thus for a given one-manifold $S$ the following geometric category:
an object is the assignment of a $G$-bundle to each subinterval labeled by a 
finite group $G$ and of $H$-bundles to marked points in the interior or end 
points. The final datum are compatible morphisms from induced bundles 
to restrictions of bundles at all marked points.
We denote this geometric category by $\Bun(S)$.

\begin{Definition}\label{def:pre-DW}
~\\[2pt]
(i)\,
A one-dimensional \emph{pre-DW manifold} is a smooth one-dimensional 
manifold $S$, possibly with boundary, together with the following data:
 \beginitemize
\item
A finite set $P_S$ of points of $S$, containing all boundary points of $S$.
\\
We refer to the elements of $P_S$ as \emph{marked} points, and to a connected 
component of $S\,{\setminus}\, P_S$ as a \emph{subinterval} of $S$. We choose 
an orientation for each subinterval.
\item
To each subinterval of $S$ we associate a finite group.
\item
To a marked point $p\iN P_S$ that is a boundary point and is thus adjacent to 
a single subinterval $I$ with associated group $G$, we select a finite group
$H$ and a group homomorphism $\iota\colon H\To G$. 
\\
To a marked point $p\iN P_S$ that is not a boundary point of $S$ and is thus
adjacent to two subintervals $I_1$ and $I_2$, labeled by finite groups $G_1$ 
and $G_2$, respectively, we select a finite group $H$ and a pair of group 
homomorphisms $\iota_i\colon H\To G_i$.
\end{itemize}
 \noindent
(ii)\,
To a one-dimensional
pre-DW manifold $S$, we associate the category $\Bun(S)$ of bundles
described above. This is an essentially finite groupoid.
~\\[4pt]
(iii)\,
Each subinterval of a one-dimensional
pre-DW manifold $S$ is endowed with an orientation.
Thereby any marked point $p\iN P_S$ is either a start point or an end point for 
any interval $I$ adjacent to $p$. In the first case, we set 
$\epsilon(p,I) \,{:=}\, {+}1$, in the latter $\epsilon(p,I) \,{:=}\, {-}1$.
\end{Definition}

To make contact with the results in \cite{fusV} which use the theory of module
categories, we need to find finite groupoids that are equivalent to groupoids
$\Bun(S)$ of relative bundles of pre-DW manifolds. This is the goal of the 
remaining part of this subsection.

As a first example, consider a circle with one marked point, which corresponds 
to a surface defect. If we associate to the interval the group $G$, then
we have to associate to the defect a group homomorphism 
$\iota\colon H\to G\Times G$, and the resulting action groupoid is
  \be
  G\sll G\Times G \srr_{\!\iota^-} H \,.
  \ee
Of particular interest is the case that the group homomorphism $\iota$
is the embedding homomorphism 
of the diagonal subgroup $G \le G\Times G$. We denote by $G \srrad G$ 
the action groupoid for the left adjoint action of $G$ on itself. The functor
  \be
  F:\quad  G\sll G\Times G \srr G \,\to\, G \srrad G
  \ee
that acts on objects as $F(\gamma_1^{},\gamma_2^{}) \eq \gamma_1^{}\gamma_2^{-1}$
and on morphisms as
  \be
  F\left((\gamma_1^{},\gamma_2^{}) \,{\stackrel{(h_1,h_2)} \longrightarrow}\,
  (h_1^{}\gamma_1^{}h_2^{-1},h_1^{}\gamma_2^{}h_2^{-1})\right) =
  \left( \gamma_1^{}\gamma_2^{-1} \,{\stackrel{h_1} \longrightarrow\,}
  h_1^{}\gamma_1^{}\gamma_2^{-1} h_1^{-1}\right)
  \ee
is an equivalence of categories. We will see that the linearization of the
adjoint action groupoid together with the relevant cocycle (see formula 
\erf{325}) produces the appropriate category associated to the circle without
marked points, i.e.\ the category of ordinary bulk Wilson lines. 

As a more involved  example, let us discuss a circle with two marked points.
We describe the circle as $S^1 \eq\{z\iN\complex \,|\, |z| \eq 1 \}$ and take 
the marked points to be $\pm \ii \iN S^1$. For the two intervals that consist of 
points with positive and negative real parts, respectively, we choose groups
$G_>$ and $G_<$, respectively. At the points $\pm\ii$, we choose group 
homomorphisms
  \be
  \iota_+: \quad H_+ \to G_> \Times G_< \qquad\text{ and }\qquad
  \iota_-: \quad H_- \to G_< \Times G_> \, .
  \ee
The relevant action groupoid is then
  \be
  G_>\Times G_< \sll
  G_>\times G_< \Times G_< \Times G_> \srr_{\!\iota_+^-\times\iota_-^-}
  H_+ \Times H_- \,,
  \labl{GG.GGGG.HH}
where the action of $G_>$ and $G_<$ is again diagonal and the left
action of $H_{\pm}$ is again by right multiplication preceded by applying the 
relevant group homomorphism and taking inverses. This description generalizes in 
an obvious manner to circles with an arbitrary finite number of marked points. 
The generalization to intervals with an arbitrary finite number of marked 
points is easy as well. We have thus succeeded in describing for a specific 
type of one-dimensional
pre-DW manifold the category $\Bun(S)$ by a finite action groupoid.

We discuss again a specific case: suppose that $G_> \eq G_< \,{=:}\,G$ and
that $H_+ \,{\cong}\, G \,{\stackrel{d}\to\,} G\Times G$ is the dia\-gonal 
subgroup, while $\iota_- \eq \iota \colon H\To G\Times G$ is an arbitrary 
group homomorphism. Then the relevant action groupoid is
  \be
  G\Times G \sll G\Times G\Times G\Times G \srr_{d^-\times\iota^-} G\Times H
  \labl{GG-GGGG-GH}
with the first copy of $G$ in the gauge group $G \Times G$
acting on the first and forth copies of $G$ in $G \Times G \Times G \Times G$
by left multiplication and the second copy of $G$ acting on the second and 
third copies. The left action of $G$ on the right is as a subgroup
of the first and second copy of $G$. The action groupoid \erf{GG-GGGG-GH}
is equivalent to the action groupoid
  \be
  G\sll G\Times G \srr_{\!\iota^-} H
  \labl{G-GG-H}
via the functor $F$ that acts on objects as
  \be
   F(\gamma_1^{},\gamma_2^{},\gamma_3^{},\gamma_4^{})
  := (\gamma_1^{} \gamma_2^{-1} \gamma_3^{},\gamma_4^{})
  \labl{F12}
and maps the morphism
  \be
   \xymatrix{
  (\gamma_1^{},\gamma_2^{},\gamma_3^{},\gamma_4^{})\,
  \ar^{(g_1^{},g_2^{},g,h)\hspace*{58pt}}[rr] &&
  \,(g_1^{}\gamma_1^{}g^{-1},g_2^{}\gamma_2^{}g^{-1},g_2^{}\gamma_3^{}h^{-1},
  g_1^{}\gamma_4^{}h^{-1})
  } \labl{g1mor}
in the groupoid \erf{GG-GGGG-GH} to the morphism
  \be
  \xymatrix{
  (\gamma_1^{} \gamma_2^{-1} \gamma_3^{},\gamma_4^{})\,
  \ar^{(g_1^{},h)\hspace*{31pt}}[rr]
  && \, (g_1^{} \gamma_1^{} \gamma_2^{-1}\gamma_3^{}h^{-1},
  g_1^{}\gamma_4^{} h^{-1})
  } \labl{g2mor}
in \erf{G-GG-H}. It is straightforward to check that this functor is
surjective and a bijection on morphism spaces and is thus an equivalence
of groupoids.

\subsection{Lagrangian data and linearization of groupoids}\label{ss3.2}

We now proceed to the linearization process. This requires
additional data which come from the cohomology of the
groupoids that have to be linearized. These data have the physical 
interpretation of (topological) Lagrangians and appropriate boundary terms.

We introduce such additional data as follows.
To an end point of an interval that is adjacent to a subinterval labeled by
a finite group $G$ and 3-co\-cyc\-le $\omega$ we associate a group homomorphism 
$\iota\colon H\To G$ and a 2-co\-chain $\theta \iN C^2(H,\complexx)$ such that 
${\rm d} \theta \eq \iota^*\omega$. It is appropriate to think about $\theta$ 
as a morphism ${\rm triv}\To \iota^*\omega$ of 2-gerbes on the groupoid 
${*}\srr H$.
The situation can be regarded as a higher categorical analogue of the role
played by gerbe modules in the description of boundary conditions in
two-dimensional theories with non-trivial Wess-Zumino terms (see e.g.\ 
\cite[Sect.\,6]{fnsw} for an exposition using gerbes and gerbe modules).
In the two-dimensional situation, one has a gerbe module on a submanifold
$\iota\colon \Sigma\To M$, which amounts to a 1-morphism 
$I_\omega\To \iota^*{\cal G}$ of gerbes on $\Sigma$ from a trivial gerbe 
$I_\omega$ to the restriction of the gerbe $\cal G$ on $M$. In the present 
situation we have a module of a 2-gerbe; technical simplifications come from 
the fact that the groups we deal with are finite 
and that thus any infinitesimal data related to connections are trivial.

In the case of two intervals adjacent to one another, labeled by 
$(G_1,\omega_1)$ and $(G_2,\omega_2)$, respectively, we choose a group 
homomorphism $\iota \eq (\iota_1,\iota_2)\colon H\to G_1\Times G_2$ and 
a 2-cochain $\theta$ on $H$ such that $d\theta \eq (\iota_2^*\omega_2^{})
\cdot (\iota_1^*\omega_1^{})^{-1}$. Again the situation has an analogue in 
two dimensions: defects in backgrounds with non-trivial Wess-Zumino term are 
described by gerbe bimodules and bibranes,
see \cite{fusW} and \cite[Sect.\,7]{fnsw} for a review.

We summarize these prescriptions in the following

\begin{Definition}\label{def:DW}
A one-dimensional \emph{DW manifold} is a one-dimensional pre-DW manifold 
$S$ together with the following choice of Lagrangian data:
\beginitemize
\item
To each subinterval of $S$ with finite group $G$, we associate a closed
3-cochain on $G$.

\item
To a marked boundary point $p\iN P_S \,{\cap}\, \partial S$ adjacent to a 
subinterval with group $G$ and 3-cocycle $\omega\iN Z^3(G,\complexx)$ and 
labeled with a group homomorphism $\iota\colon H\To G$, we assign a 2-cochain 
$\theta\iN C^2(H,\complexx)$ such that 
  \be
  \rmd\theta = \iota^*\omega^{\epsilon(p,I)} ,
  \ee
with $\epsilon(p,I)$ as defined in Definition \ref{def:pre-DW}(iii).
 
\item
To a marked interior point $p\iN P_S \,{\setminus}\, \partial S$ adjacent to 
subintervals $I_1$ and $I_2$ with group homomorphisms $\iota_i\colon H\to G_i$
we assign a cochain $\theta\iN C^2(H,\complexx)$ such that 
  \be
  \rmd\theta = 
  \iota_1^*\omega_1^{\epsilon(p,I_1)}\cdot \iota_2^*\omega_2^{\epsilon(p,I_2)} .
  \ee
\end{itemize}
\end{Definition}

We now use the data of a DW manifold to define twisted linearizations 
of the groupoids that we constructed in the previous subsection.  Let us
describe the general idea of a twisted linearization of a finite groupoid
$H \sll G$ given by a left action of a group $H$ on a set $G$.
The ordinary linearization is the functor category $[H \sll G,\Vectc]$.
An object of this category is given by
 \beginitemize

\item 
A finite-dimensional vector space $V_\gam$ for each element $\gam\iN G$.

\item
For each $\gam\iN G$ and $h\iN H$ a linear map 
$\rho_h\colon V_\gam\To V_{h.\gam}\,$
such that the diagram
  \be
  \xymatrix{
  & V_{h_2.\gam} \ar^{\rho_{h_1^{}}^{}} [dr] & \\
  V_\gam \ar^{\rho_{h_2^{}}^{}}[ur]\ar_{\rho_{h_1h_2}^{}}^{} [rr] && 
  V_{h_1h_2.\gam} 
  } \ee
commutes for all $\gam\iN G$ and $h_1,h_2\iN H$.
\end{itemize}
Morphisms in the functor category are natural transformations; explicitly, 
they are $G$-homo\-ge\-ne\-ous maps commuting with the $H$-action.

The additional input datum for a \emph{twisted} linearization is a 2-cocycle 
$\tau$ on the groupoid $H \sll G$.
This gives rise to the following twisted version of the functor category 
$[H \sll G,\Vectc]$ (see also \cite[Sect.\,5.4]{mort6}):

\begin{defi}
The $\tau$-\emph{twisted linearization} of the groupoid $H \sll G$, denoted by
${[H \sll G,\Vectc]}^\tau$, is the following category. An object of 
${[H \sll G,\Vectc]}^\tau$ consists of
 \beginitemize

\item 
A finite-dimensional vector space $V_\gam$ for each $\gam\iN G$.

\item
For each $h\iN H$ a linear map $\rho_h\colon V_\gam \To V_{h.\gam}$ such 
that the composition law of the $H$-action is realized projectively, i.e.\ up 
to the scalar factor $\tau(h_1,h_2;\gam) \iN \complexx$. Diagrammatically,
  \be
  \xymatrix @R+17pt{
  & V_{h_2^{}.\gam} \ar^{\rho_{h_1^{}}} [dr] 
  \ar@{=>}^{\,\begin{turn}{270}~ $\scriptstyle\tau(h_1,h_2;\gam)$
  \end{turn}}[d]& \\
  V_\gam \ar^{\rho_{h_2}^{}}[ur]\ar_{\rho_{h_1h_2}^{}} [rr] 
  && V_{h_1h_2.\gam} 
  } \ee
As a formula,
  \be
  \rho_{h_1 h_2}^{}= \tau(h_1,h_2;\gam)\, \rho_{h_1}^{}\, \rho_{h_2}^{} \,.
  \ee
\end{itemize}
Morphisms of ${[H \sll G,\Vectc]}^\tau$ are $G$-homogeneous 
maps commuting with the $H$-action.
\end{defi}

\subsection{2-cocycles from Lagrangian data}\label{ss3.3}

Our next task is thus to use the Lagrangian data that 
are part of the data of a one-dimensional DW-manifold. We have assigned them
in Definition \ref{def:DW} to intervals and circles with marked points
to produce 2-cocycles for the groupoids discussed in Section \ref{ss3.1}.
For brevity we consider in this subsection Lagrangian data for boundaries 
only; the discussion for surface defects is similar.

Any homomorphism $\iota\colon H\To G$ of finite groups provides 
a morphism $\iota\colon BH\To BG$ of the corresponding classifying spaces.
Assume now 
that we are given a 3-cocycle $\omega\iN Z^{3}(BG,\complexx)$ and a 2-cochain
$\theta\iN C^{2}(BH,\complexx)$ such that
  \be
  i^{*}\omega={\rm d}\theta \,.
  \ee

We recall that a $G$-bundle on a manifold $M$ can be described by a map from 
$M$ to the classifying space $BG$. Morphisms of bundles can be described by 
homotopies between such maps. Thus for $\Sigma$ an oriented one-dimensional 
manifold with boundary, a relative bundle on the relative manifold 
$(\Sigma,\partial\Sigma)$ leads to the following data (up to homotopy):
 \beginitemize
\item
A map $f\iN \Map(\Sigma,BG)$ describing a $G$-bundle on $\Sigma$.

\item
A map $g\iN\Map(\partial\Sigma,BH)$ describing an $H$-bundle on $\partial\Sigma$.

\item
A homotopy describing the morphism of bundles, i.e.\ a map
$h\iN \Map([0,1],\Map(\partial\Sigma,BG))$, with $[0,1]$ the standard interval.
\end{itemize}
We will later need the subset $X_\circ$ consisting of such triples $(f,g,h)$ 
subject to the condition that $h$ is a homotopy relating the maps 
$f|_{\partial\Sigma}$ and $\iota\cir g$ from $\partial\Sigma$ to $BG$,
  \be
  X _\circ := \{\, (f,g,h) \,|\,
  f\big|_{\partial\Sigma} \simeq^{h} i\circ g \,\} \,.
  \labl{defX}
Each point of $X _\circ$ describes a relative bundle, i.e.\ an object of 
$\Bun_{(G,H)}(\partial\Sigma{\to}\Sigma)$. Isomorphism classes of relative 
bundles are in bijection with the set $\pi_0(X _\circ)$ of connected 
components of $X _\circ$.

{}From the cohomological data $\omega$ and $\theta$ we now build a 2-cocycle 
in $Z^{2}(X_\circ,\complexx)$. To this end we use the evaluation map
  \be
  \ev: \quad \Sigma\times\Map(\Sigma,BG) \to {BG}
  \ee
to define a cochain 
$\tau_{\Sigma}(\omega)\iN C^{2}(\Map(\Sigma,BG),\complexx)$ by
  \be
  \tau_{\Sigma}(\omega) := \int_{\Sigma} \ev^{*}\omega \,,
  \ee
where $\int_{\Sigma}$ denotes the pushforward along the fibration 
$p_2\colon \Sigma\Times\Map(\Sigma,BG) \To \Map(\Sigma,BG)$. 
As $\Sigma$ can have a non-empty boundary, there is, in general, 
no reason that the cochain $\tau_{\Sigma}(\omega)$ should be closed.

By the same procedure we obtain a 2-cochain $\tau_{\partial\Sigma}(\theta)
\iN C^2(\Map(\partial\Sigma,BH)),\complexx)$, as well as
a 2-cochain $\tau_{[0,1]}(\tau_{\partial\Sigma}(\omega))
\iN C^{2}(\Map([0,1],\Map(\partial\Sigma,BG)),\complexx)$.
We then consider the product space 
  \be
  X := \Map(\Sigma,BG) \times \Map(\partial\Sigma,BH) \times
  \Map([0,1],\Map(\partial\Sigma,BG)) \,.
  \labl{MaMaMa}
The pullbacks along the canonical projections $p_i$ to the three factors of 
\erf{MaMaMa} supply us with a 2-cochain on $X$:
  \be
  \varphi := p_{1}^{*}\tau_{\Sigma}(\omega) - p_{2}^{*}\tau_{\partial\Sigma}
  (\theta) -p_{3}^{*} \tau_{[0,1]}(\tau_{\partial\Sigma}(\omega)) \,.
  \labl{varphidef}

The space $X_\circ$ introduced in \erf{defX} to describe relative bundles is
by definition a subspace of $X$ \erf{MaMaMa}. 
The central insight is now that the 2-cochain that is obtained by restricting 
$\varphi$ to the subspace $X_\circ$ of $X$ is closed,
  \be
  {\rm d}{\varphi|_{X_\circ^{}}}^{} = 0 \,.
  \labl{dphiX=0}
In other words, we have obtained a 2-cocycle $\varphi|_{X_\circ^{}}^{} \iN 
Z^{2}(X_\circ,\complexx)$ on the space $X_\circ$ describing relative bundles.

To see that \erf{dphiX=0} holds, we work for the moment with differential forms 
and consider an arbitrary manifold $U$. Consider
$\alpha\iN\Omega_{\rm cl}^{3}(\Sigma\Times{U},\mathbb{R})$ 
and $\beta\iN\Omega^2(\partial\Sigma\Times U),\mathbb{R})$ obeying
$ \alpha|^{}_{\partial{\Sigma}\times{U}} \eq \rmd\beta $. 
Taking into account that $\Sigma$ has a boundary, we have
  \be
  \rmd(\int_{\Sigma}\alpha)
  = \int_{\Sigma}\rmd\alpha 
  + \int_{\partial\Sigma}\alpha\big|^{}_{\partial\Sigma\times{U}}
  = \int_{\Sigma}\rmd\alpha + \int_{\partial\Sigma} \rmd\beta
  = \int_{\partial\Sigma}\rmd\beta \,.
  \ee
This means that the form
  \be
  \phi := \int_{\Sigma}\alpha-\int_{\partial\Sigma}\beta
  ~\in \Omega^2(U,\mathbb{R})
  \ee
is closed, $\rmd\phi \eq 0$. The same argument applies to
elements in $Z^{3}(\Sigma\Times{U},\complexx)$ where slant products
are used as the analogue of integration along the fiber.

The argument can now be applied to the situation of our interest: The role of 
$\int_\Sigma\alpha$ is then played by $p_{1}^{*}\tau_\Sigma^{}(\omega)|^{}
_{X_\circ}$ and the role of $\int_{\partial\Sigma}\beta$ by $(p_{2}^{*}\tau
_{\partial\Sigma}^{} (\varphi)+p_{3}^{*} \tau_{[0,1]}^{}
(\tau_{\partial\Sigma}^{}(\omega)))|^{}_{X_\circ}$. Their difference is 
precisely the combination $\varphi$ introduced in \erf{varphidef}. From
the relation $ \alpha|^{}_{\partial{\Sigma}\times{U}} \eq \rmd\beta $
we thus obtain the desired equality \erf{dphiX=0}.

\subsection{Graphical calculus for groupoid cocycles}\label{ss:gcc}

Generalizing the approach of \cite{willS}, we can achieve a more combinatorial 
description of the 2-co\-cyc\-les on the group\-oids derived in Section 
\ref{ss3.1}. We formulate it with the help of an algorithm which is based on
three-dimensional diagrams and their decomposition into simplices. The diagrams
are obtained from a graphical representation of the groupoids involved.

We start with a one-dimensional diagram, drawn vertically, which represents 
a one-dimen\-si\-o\-nal  pre-DW manifold to which we wish to associate a 
category by linearization. 
These manifolds are circles or intervals with finitely many marked points, 
including boundary points in the case of intervals. Each subinterval is
marked by a finite group $G_i$ and a 3-cocycle $\omega_i\iN Z^3(G_i,\complexx)$.
For each marked point we have a group $H_j$ and group homomorphisms to the
groups associated with the adjacent intervals. The data characterizing 
an object in the associated groupoid described in Section \ref{ss3.1}
are then elements in the groups $G_i$ associated to the 
subinterval, one for each point adjacent to the subinterval. 
 
Our convention is now to draw an empty circle for a marked point and to replace
the original subintervals by filled circles. Between these circles we draw edges
which are labeled by elements of the groups $G_i$ that are part of the data 
describing a relative bundle. An example is depicted in the following picture:
  \eqpic{Figure1}{130}{59} {
  \put(0,0)   {\Includepic{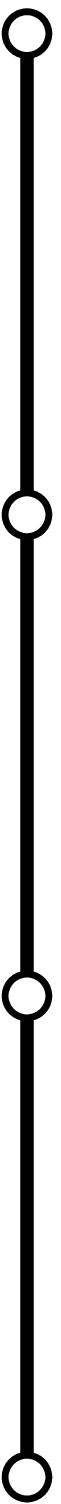}
  \put(5.8,-.2) {\scriptsize $ H_3^{}\,,\iota_3^{} $}
  \put(5.4,23)  {\scriptsize $ G_3^{} $}
  \put(5.8,44.4){\scriptsize $ H_{23}^{}\,,\iota_{23}^2\,,\iota_{23}^3 $}
  \put(5.4,67.6){\scriptsize $ G_2^{} $}
  \put(5.8,89)  {\scriptsize $ H_{12}^{}\,,\iota_{12}^1\,,\iota_{12}^2 $}
  \put(5.4,112) {\scriptsize $ G_1^{} $}
  \put(5.8,133.6) {\scriptsize $ H_1^{}\,,\iota_1^{} $}
  }
  \put(78,64)   {{\large$ \rightsquigarrow $}}
  \put(130,0) {\Includepic{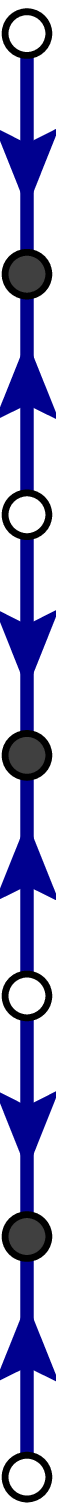}
  \put(-8.5,17) {\scriptsize $ \gamma_6^{} $}
  \put(-8.5,33) {\scriptsize $ \gamma_5^{} $}
  \put(-8.5,62) {\scriptsize $ \gamma_4^{} $}
  \put(-8.5,78) {\scriptsize $ \gamma_3^{} $}
  \put(-8.5,106){\scriptsize $ \gamma_2^{} $}
  \put(-8.5,123){\scriptsize $ \gamma_1^{} $}
  } }
The figure on the left hand side of \erf{Figure1} shows the pre-DW-manifold $S$
which is an interval with two interior marked points, together with the relevant 
groups and group homomorphisms. The labels in the figure on the right hand side 
are group elements $\gamma_1^{},\gamma_2^{} \iN G_1$, 
$\gamma_3^{},\gamma_4^{} \iN G_2$ and $\gamma_5^{},\gamma_6^{} \iN G_3$i 
that specify an object in $\Bun(S)$..

A morphism in the groupoid consists of elements of the groups $H_j$ and
$G_i$ describing gauge transformations of the involved bundles. We represent
such morphisms by two-dimensional diagrams with oriented edges as follows:
  \eqpic{Figure2}{55}{59} {
  \put(0,0)     {\Includepic{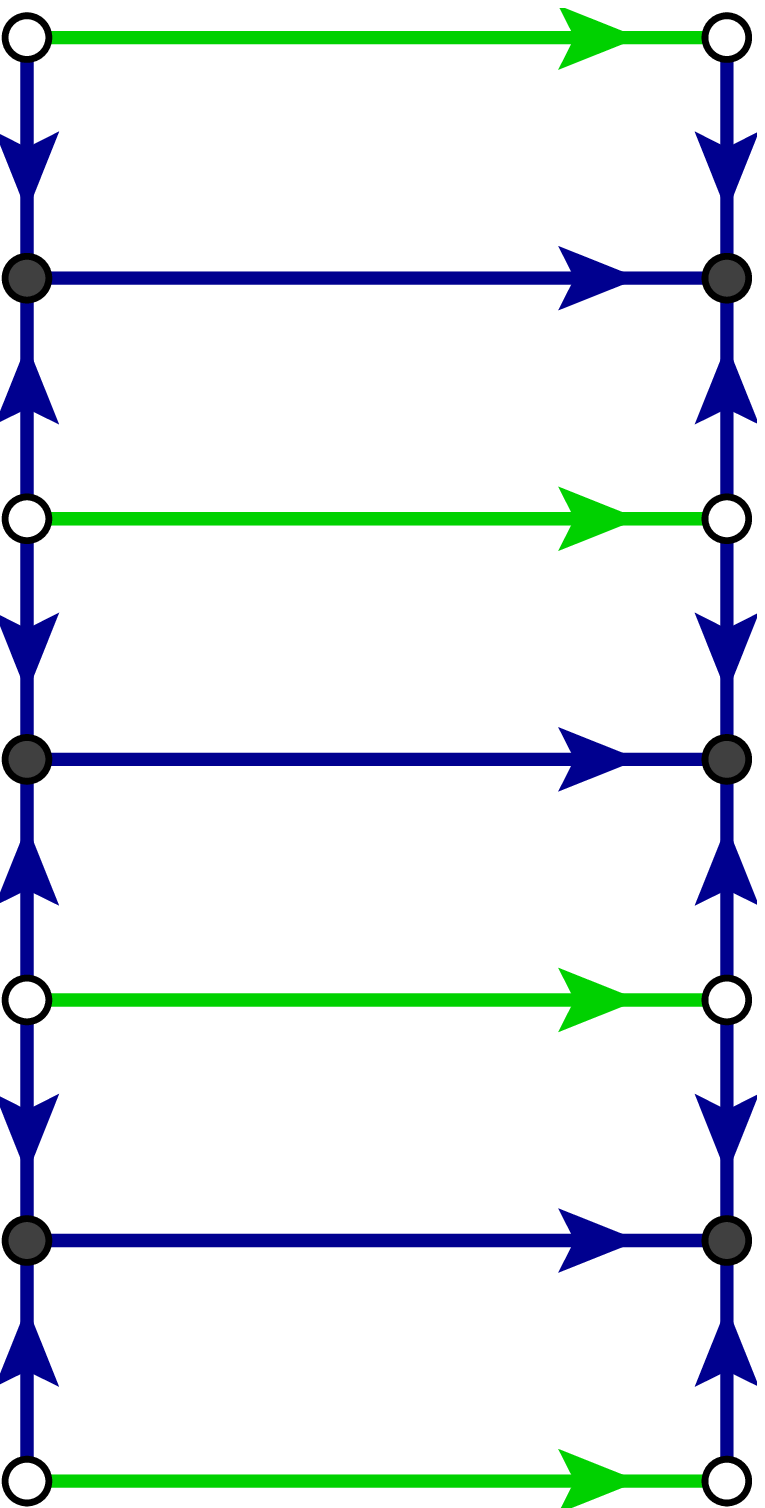}}
  \put(-8.5,17) {\scriptsize$ \gamma_6^{} $}
  \put(-8.5,33) {\scriptsize$ \gamma_5^{} $}
  \put(-8.5,62) {\scriptsize$ \gamma_4^{} $}
  \put(-8.5,78) {\scriptsize$ \gamma_3^{} $}
  \put(-8.5,106){\scriptsize$ \gamma_2^{} $}
  \put(-8.5,123){\scriptsize$ \gamma_1^{} $}
  \put(43,-5.5) {\scriptsize$ h_3^{} $}
  \put(43,18.8) {\scriptsize$ g_3^{} $}
  \put(43,39.2) {\scriptsize$ h_{23}^{} $}
  \put(43,63.5) {\scriptsize$ g_2^{} $}
  \put(43,83.8) {\scriptsize$ h_{12}^{} $}
  \put(43,108.2){\scriptsize$ g_1^{} $}
  \put(43,128.4){\scriptsize$ h_1^{} $}
  \put(70,16)   {\scriptsize$ \gamma_6' $}
  \put(70,32)   {\scriptsize$ \gamma_5' $}
  \put(70,61)   {\scriptsize$ \gamma_4' $}
  \put(70,77)   {\scriptsize$ \gamma_3' $}
  \put(70,105)  {\scriptsize$ \gamma_2' $}
  \put(70,122)  {\scriptsize$ \gamma_1' $}
  } 
Here horizontal edges connecting empty circles are labeled by elements of the
groups $H_j$, while horizontal edges connecting filled circles are labeled by 
elements of the groups $G_i$. For each square in the diagram there is a 
consistency condition relating the labels of its edges. To formulate this 
condition, we adopt the convention that orientation reversal amounts to 
inversion of the group element that labels the edge:
  \eqpic{Figure15}{60}{13} {
  \put(0,0)   {\Includepic{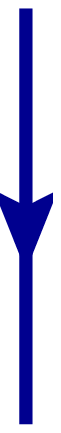}
  \put(-8.5,19) {$ \gamma $}
  }
  \put(37,14)   {\large $ \hat= $}
  \put(77,0) {\Includepic{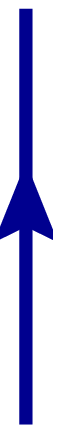}
  \put(5.8,19.6){$ \gamma^{-1} $}
  } }
With this convention the product of all group elements (possibly after applying
an appropriate group homomorphism $H_j \To G_i$) along a closed curve equals 
the neutral element; we refer to this relation as the \emph{holonomy condition}. 
For instance, the holonomy condition for the top square in \erf{Figure2} 
is the equality
  \be
  \gamma_1'\cdot \iota_1(h_1)= g_1\cdot \gamma_1
  \ee
in $G_1$.
This determines the element $\gamma_1'$ of $G_1$, or alternatively $\gamma_1$ 
or $g_1$, as a function of the three other group elements. Also, in case the 
homomorphism $\iota_1$ is injective it alternatively fixes $h_1 \iN H_1$ in 
terms of the three other elements., 

We wish to obtain a 2-cocycle on the groupoid we have just described.
For a general groupoid $\Gamma=(\Gamma_{\!0},\Gamma_{\!1})$ with sets 
$\Gamma_{\!0}$ of objects and $\Gamma_{\!1}$ of morphisms we define the 
2-cocycle by its values $\tau(g_1^{},g_1';\gamma)$ for an object 
$\gamma\iN\Gamma_{\!0}$ and two compatible morphisms $g_1^{},g_1'\iN 
\Gamma_{\!1}$. We depict these values graphically as triangles,
  \eqpic{Figure3}{190}{17} {
  \put(0,22)    {$ \tau(g_1^{},g_1';\gamma) ~= $}
  \put(82,2)    {\Includepic{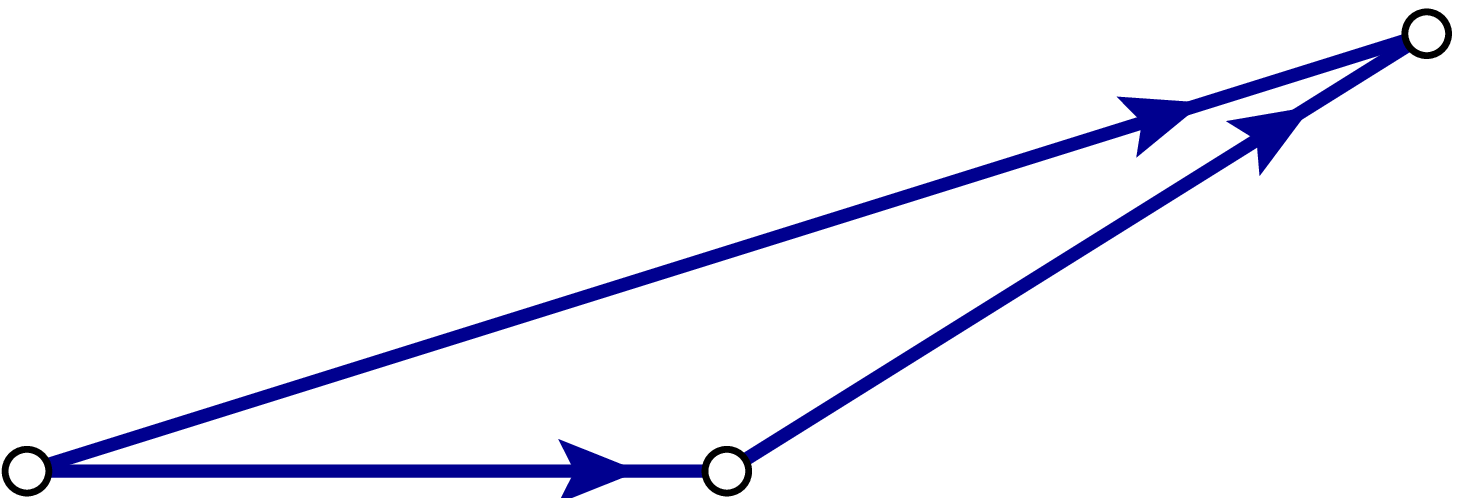}
  \put(-6,.7)   {\scriptsize $ \gamma $}
  \put(42,-3.8) {\scriptsize $ g_1^{} $}
  \put(71.3,-.2){\scriptsize $ g_1^{} \gamma $}
  \put(85,37)   {\scriptsize $ g_1'g_1^{} $}
  \put(108,21.5){\scriptsize $ g_1' $}
  \put(136.3,42){\scriptsize $ g_1'g_1^{} \gamma $}
  } }
(Again the holonomy condition is in effect: we have
$(g_1'g_1^{})^{-1} g_1'g_1^{} \eq e$.)

Now in the situation of our interest, in which we represent objects 
and morphisms of the groupoid by one-dimensional and two-dimensional graphical 
elements, respectively, we obtain a graphical representation 
of the 2-cocycle by a piecewise-linear three-manifold. In the case 
of an interval considered in \erf{Figure2} -- but now, for simplicity, with 
only a single interior marked point -- this three-manifold looks as follows:
  \eqpic{Figure4}{80}{72} {
  \put(0,0)     {\Includepic{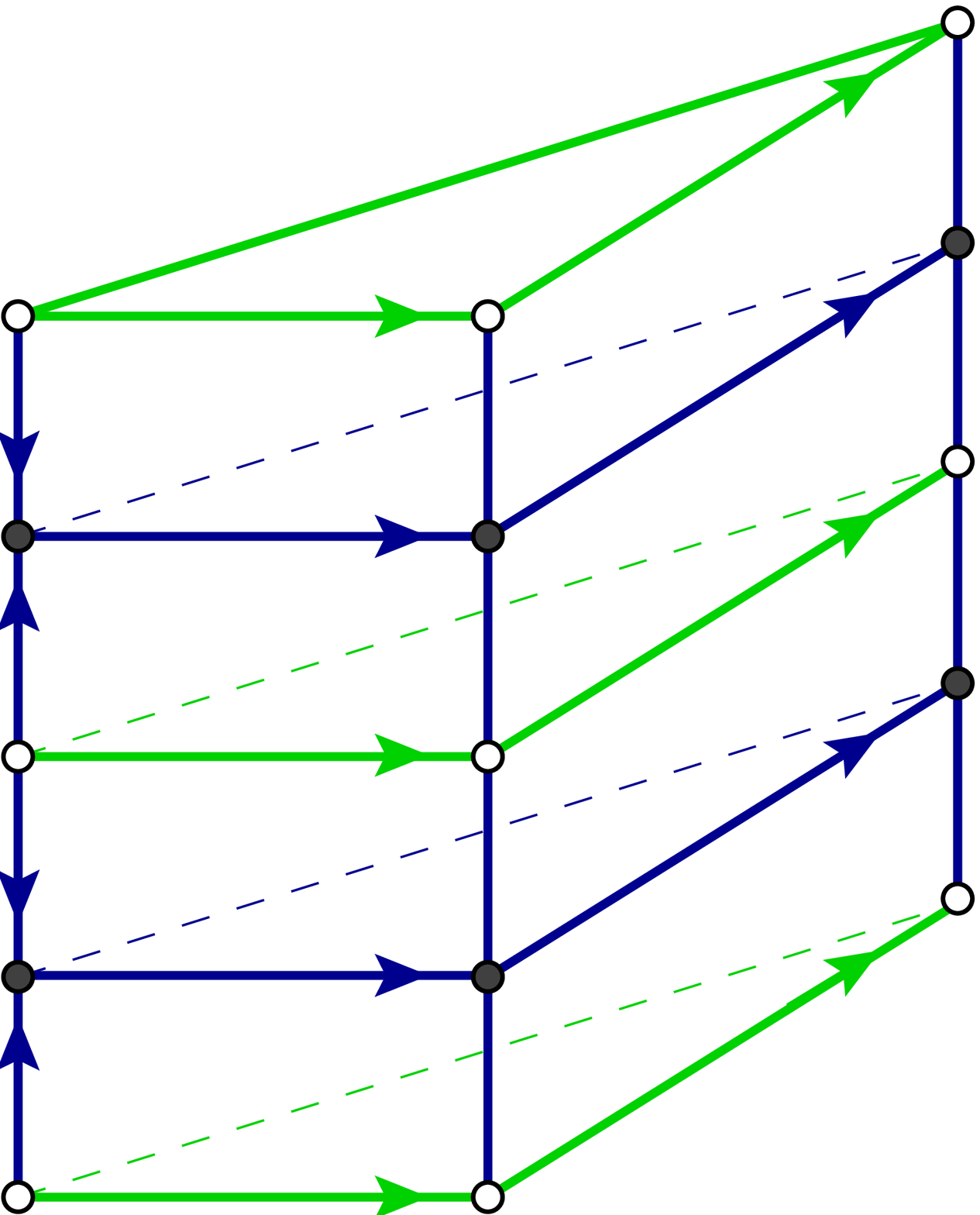}}
  \put(-8.5,17) {\scriptsize $ \gamma_4^{} $}
  \put(-8.5,51) {\scriptsize $ \gamma_3^{} $}
  \put(-8.5,77) {\scriptsize $ \gamma_2^{} $}
  \put(-8.5,112){\scriptsize $ \gamma_1^{} $}
  \put(43,-5.5) {\scriptsize $ h_2^{} $}
  \put(43,27.1) {\scriptsize $ g_2^{} $}
  \put(43,55.5) {\scriptsize $ h_{12}^{} $}
  \put(43,87.8) {\scriptsize $ g_1^{} $}
  \put(43,116)  {\scriptsize $ h_1^{} $}
  \put(108,21)  {\scriptsize $ h_2' $}
  \put(108,53.2){\scriptsize $ g_2' $}
  \put(108,81.7){\scriptsize $ h_{12}' $}
  \put(108,114) {\scriptsize $ g_1' $}
  \put(108,142) {\scriptsize $ h_1' $}
  }
Here the labeling of all lines for which the labels are not indicated explicitly
is fixed as a function of the displayed labels by the holonomy condition.

\medskip

Following the strategy in \cite{willS}, 
our goal is now to cut the so obtained three-manifolds into standard pieces to 
which we can naturally assign values in $\complexx$. The value of the groupoid
2-cocycle is then given by the product of the numbers associated with the 
various standard pieces into which the three-manifold is decomposed.
In our situation, in which also physical boundaries and surface defects are
present, there are \emph{two} types of standard pieces:
 \beginitemize
\item
First, a 3-simplex whose edges are all labeled by elements $g_1,g_2,g_3,...$
of a group $G$ with 3-cocycle $\omega \iN Z^3(G,\complexx)$, subject to the 
holonomy condition. To such a 3-simplex
  \eqpic{Figure5}{120}{47} {
  \put(0,0)     {\Includepic{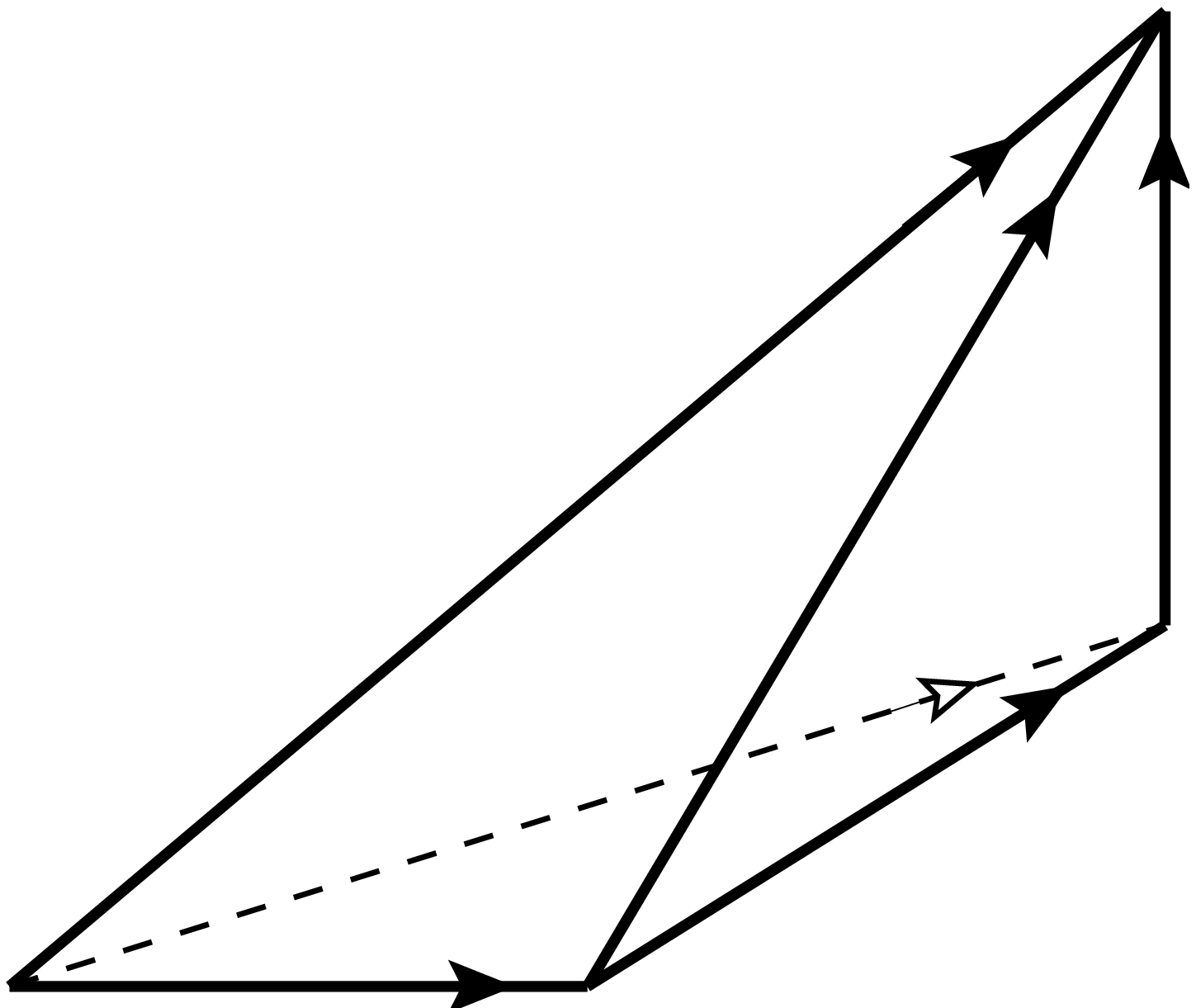}
  \put(42,-3.6) {\scriptsize$ g_1^{} $}
  \put(94,37.6) {\begin{turn}{20}\scriptsize$ g_2^{} g_1^{} $\end{turn}}
  \put(82,78.2) {\begin{turn}{42}\scriptsize$ g_3^{} g_2^{} g_1^{} $\end{turn}}
  \put(99,71.5){\begin{turn}{62}\scriptsize$ g_3^{} g_2^{} $\end{turn}}
  \put(108,23.3){\scriptsize$ g_2^{} $}
  \put(133.5,88){\scriptsize$ g_3^{} $}
  } }
we associate the number
  \be
  \tilde\omega(g_1,g_2,g_3)
  := \omega(g_1^{-1},g_2^{-1},g_3^{-1}) \,\in \complexx .
  \ee
\item
Second, a horizontal triangle whose edges are correspondingly labeled by 
elements of a group $H$ with 2-cochain $\theta$. To such a triangle 
  \eqpic{Figure6}{110}{14} {
  \put(0,2)     {\Includepic{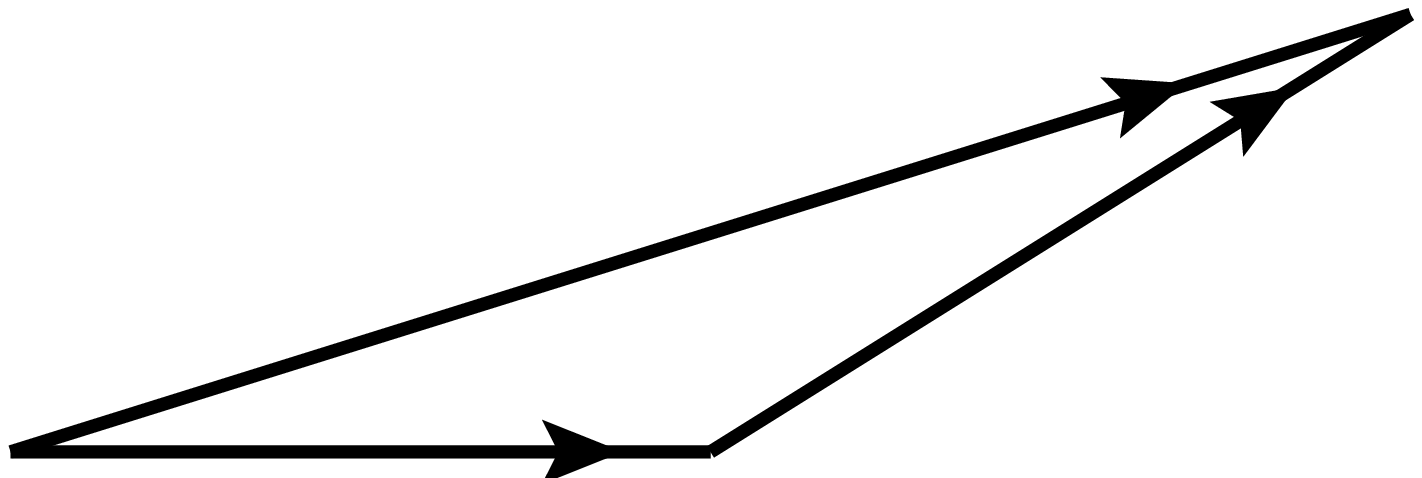}
  \put(43,-5.5) {\scriptsize$ h_1^{} $}
  \put(82.4,37) {\scriptsize$ h_2^{} h_1^{} $}
  \put(110,22.3){\scriptsize$ h_2^{} $}
  } }
we associate the number 
  \be
  \tilde\theta(h_1^{},h_2^{})
  := [\, \theta(h_1^{-1},h_2^{-1}) \,]^{-1}_{} \,\in \complexx .
  \ee
\end{itemize}

We require that any horizontal triangle having only empty circles as vertices
that is contained in a three-dimensional diagram of our interest must be taken
as a face of the simplicial decomposition. The symmetric groups $S_4$ and $S_3$ 
which consist of permutations of the vertices in \erf{Figure5} and \erf{Figure6}, 
respectively, are realized on $\tilde\omega$ and $\tilde\theta$ by
a sign that depends on the relative orientations of the two bases involved,
i.e.\ we have equalities such as
  \be
  \tilde\omega(g_1^{},g_2^{},g_3^{})
  = \tilde\omega(g_1^{-1}g_2^{-1}g_3^{-1},g_1^{},g_2^{})^{-1}_{}
  = \tilde\omega(g_3^{-1},g_2^{-1},g_1^{-1})
  \labl{rel1}
and
  \be
  \tilde\theta(h_1^{},h_2^{}) = \tilde\theta(h_1^{-1}h_2^{-1},h_1^{})
  = \tilde\theta(h_2^{-1},h_1^{-1})
  \labl{theta-ide}
etc. We require that $\tilde\omega$ and $\tilde\theta$ are normalized, i.e.
  \be
  \tilde\omega(e,g,g') = 1 \qquand  \tilde\theta(e,h) = 1 \,.
  \labl{rel3}
We will freely use the identities \erf{rel1}\,--\,\erf{rel3} below.

\medskip

A simplicial decomposition obtained this way is not unique. We therefore must
still verify that the value of the 2-cocycle on the groupoid that is obtained
by our prescription is well-defined. When no boundaries or defects (and thus
no triangular standard pieces) are involved, there are two situations to be 
dealt with: First, a gone with 5 vertices, 8 edges, 4 triangles and 1 
quadrangle. This gone can be decomposed into tetrahedra in two different 
ways; the first is a decomposition
  \eqpic{Figure7}{410}{47} {
  \put(0,0)     {\Includepic{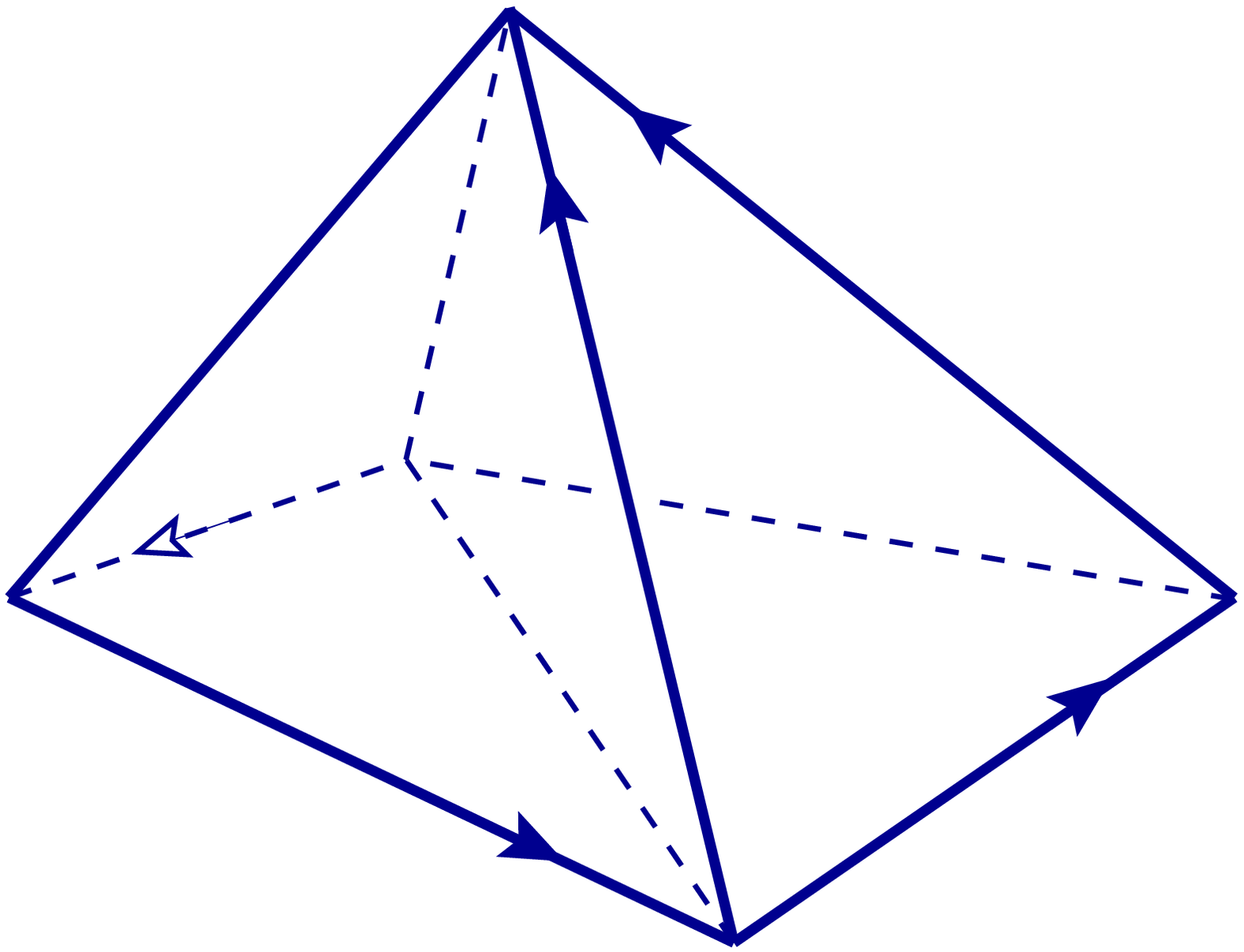}
  \put(26,45.2) {\scriptsize $ g_1^{} $}
  \put(52,7.5)  {\scriptsize $ g_2^{} $}
  \put(71,77.8) {\scriptsize $ g_4^{}g_3^{} $}
  \put(83,96.2) {\scriptsize $ g_4^{} $}
  \put(121,20.3){\scriptsize $ g_3^{} $}
  }
  \put(163,48)  {{\Large$ \rightsquigarrow $}}
  \put(200,0)   {\Includepic{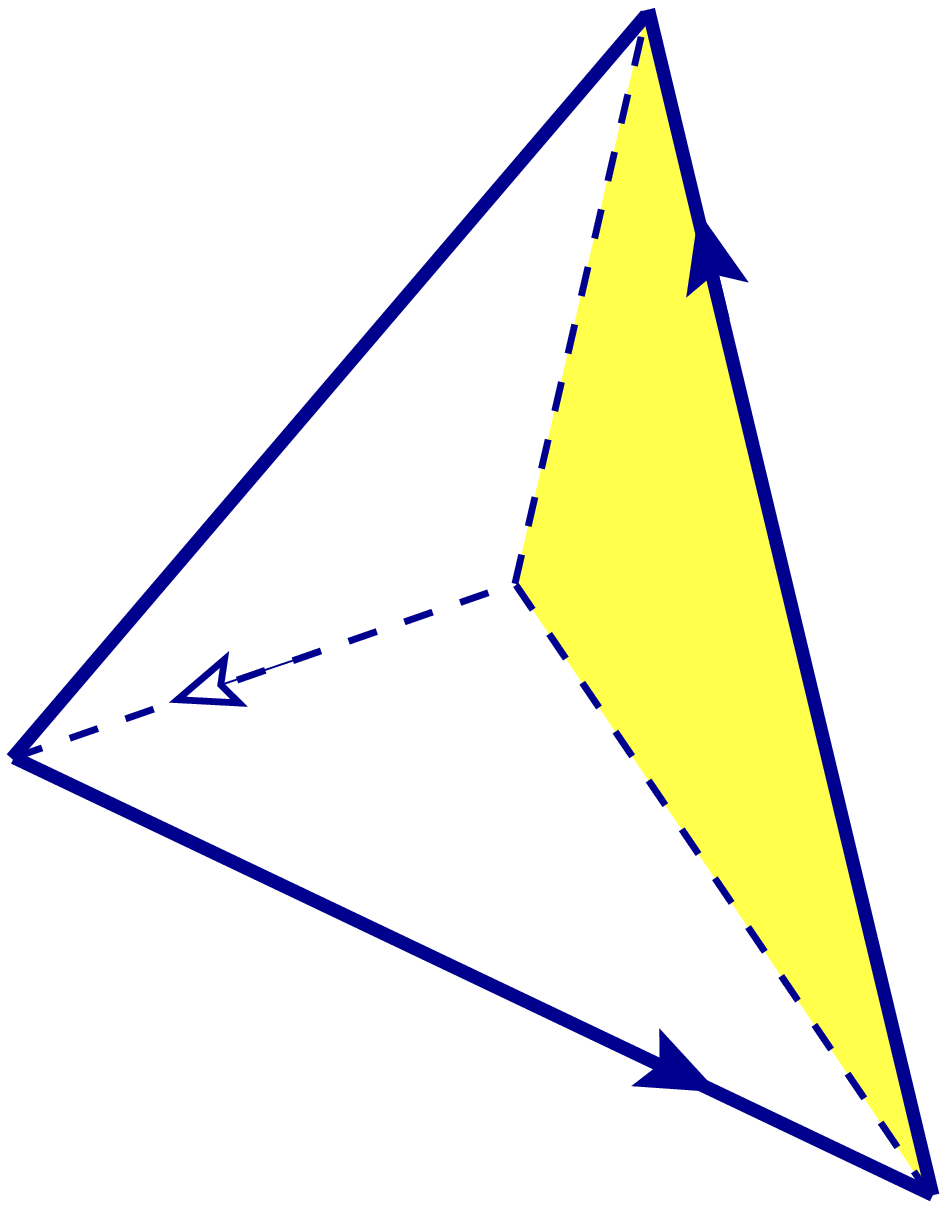}
  \put(53,7.5)  {\scriptsize $ g_2^{} $}
  \put(65,102.2){\scriptsize $ g_4^{}g_3^{} $}
  }
  \put(285,0)   {\Includepic{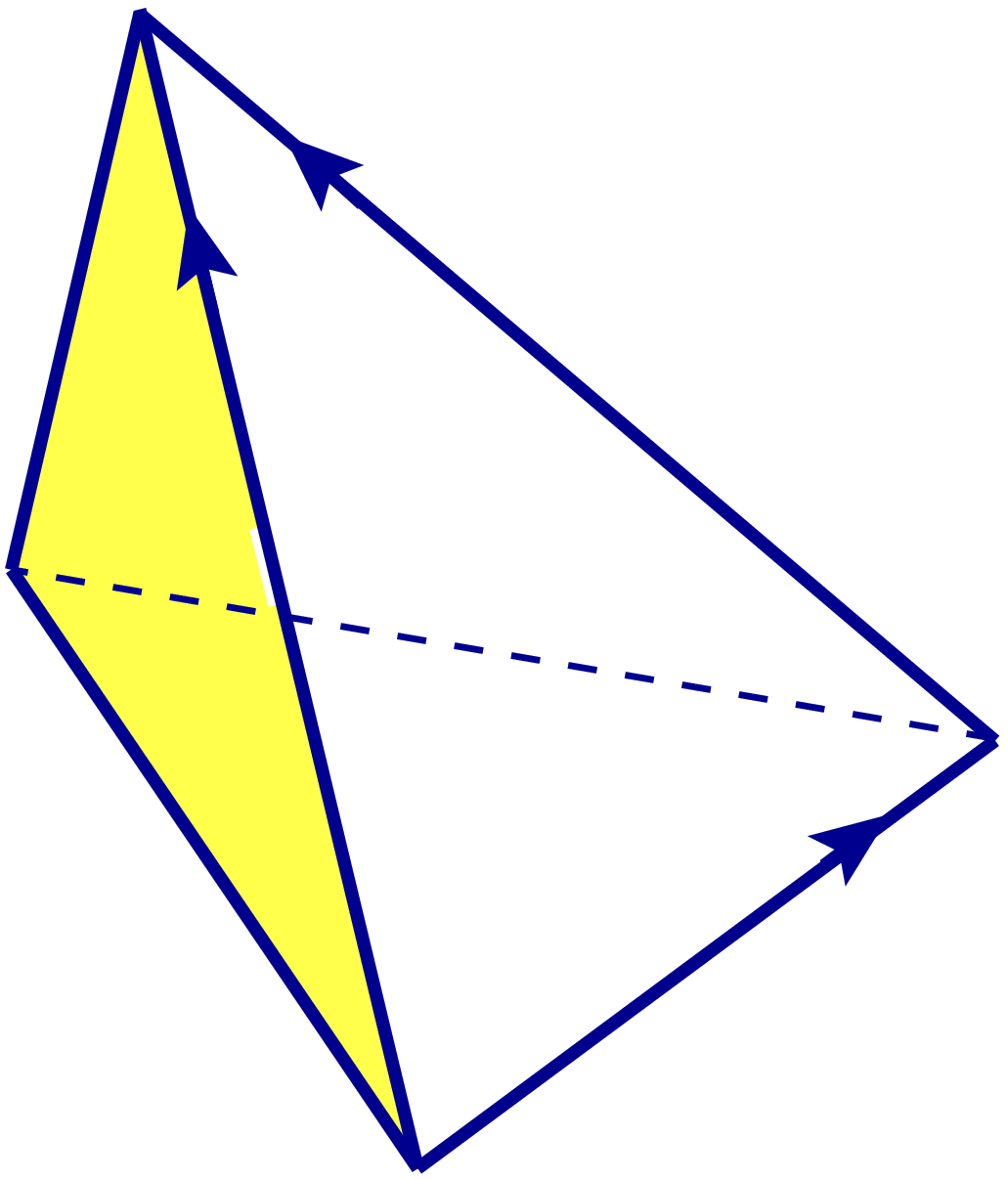}
  \put(24,78.8) {\scriptsize $ g_4^{}g_3^{} $}
  \put(37,94.4) {\scriptsize $ g_4^{} $}
  \put(76,23)   {\scriptsize $ g_3^{} $}
  } }
into two tetrahedra that share a face (shaded in the picture).
The other is a decomposition is into three tetrahedra according to
  \Eqpic{Figure8}{400}{66} {
  \put(-31,8)   {\Includepic{8}
  \put(26,45.2) {\scriptsize $ g_1^{} $}
  \put(52,7.5)  {\scriptsize $ g_2^{} $}
  \put(83,96.2) {\scriptsize $ g_4^{} $}
  \put(121,20.3){\scriptsize $ g_3^{} $}
  }
  \put(131,57)  {{\Large$ \rightsquigarrow $}}
  \put(158,34)  {\Includepic{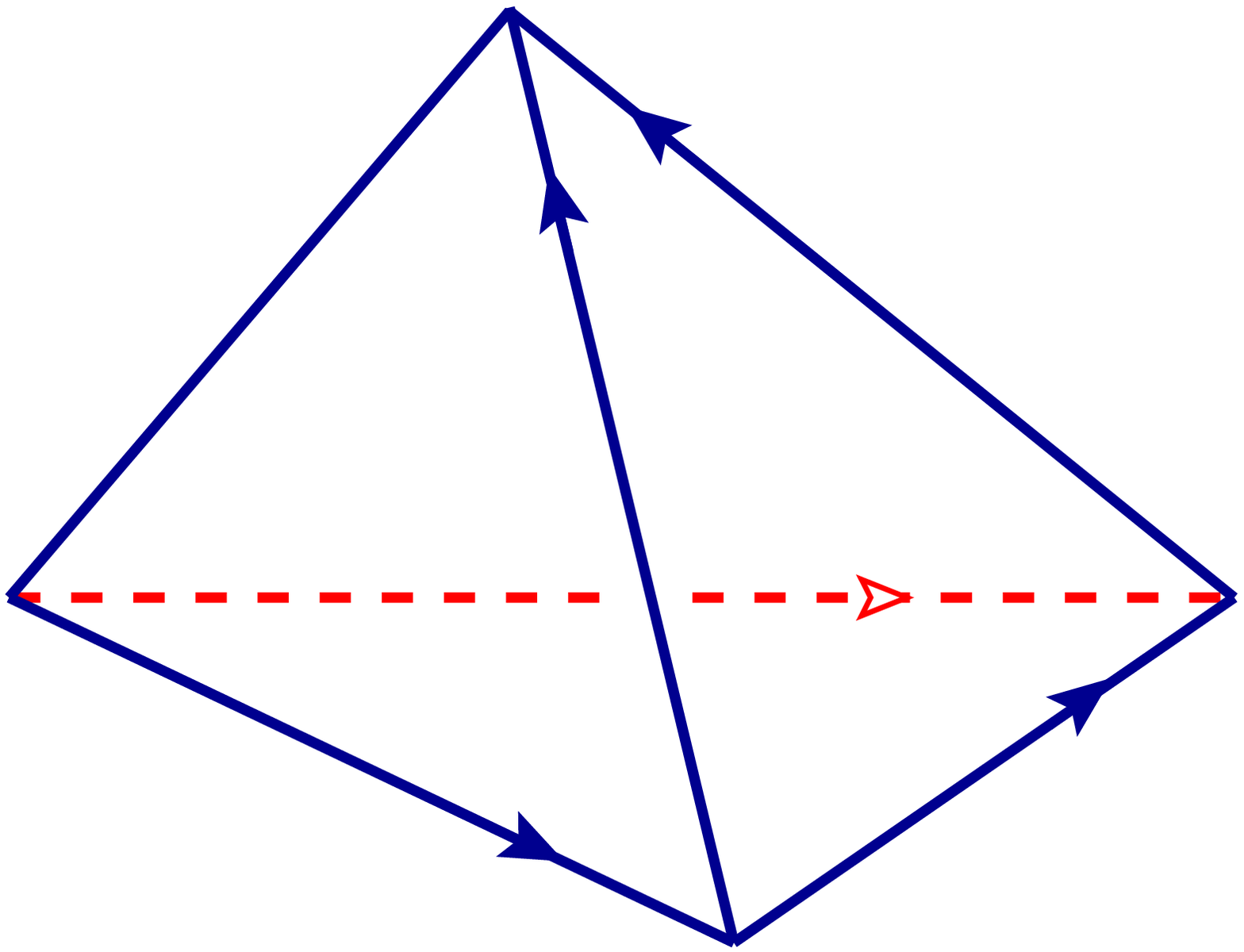}
  \put(53,7.2)  {\scriptsize $ g_2^{} $}
  \put(84,95.4) {\scriptsize $ g_4^{} $}
  \put(105,35.8){\scriptsize $ g_3^{}g_2^{} $}
  \put(121,20)  {\scriptsize $ g_3^{} $}
  }
  \put(261,-15) {\Includepic{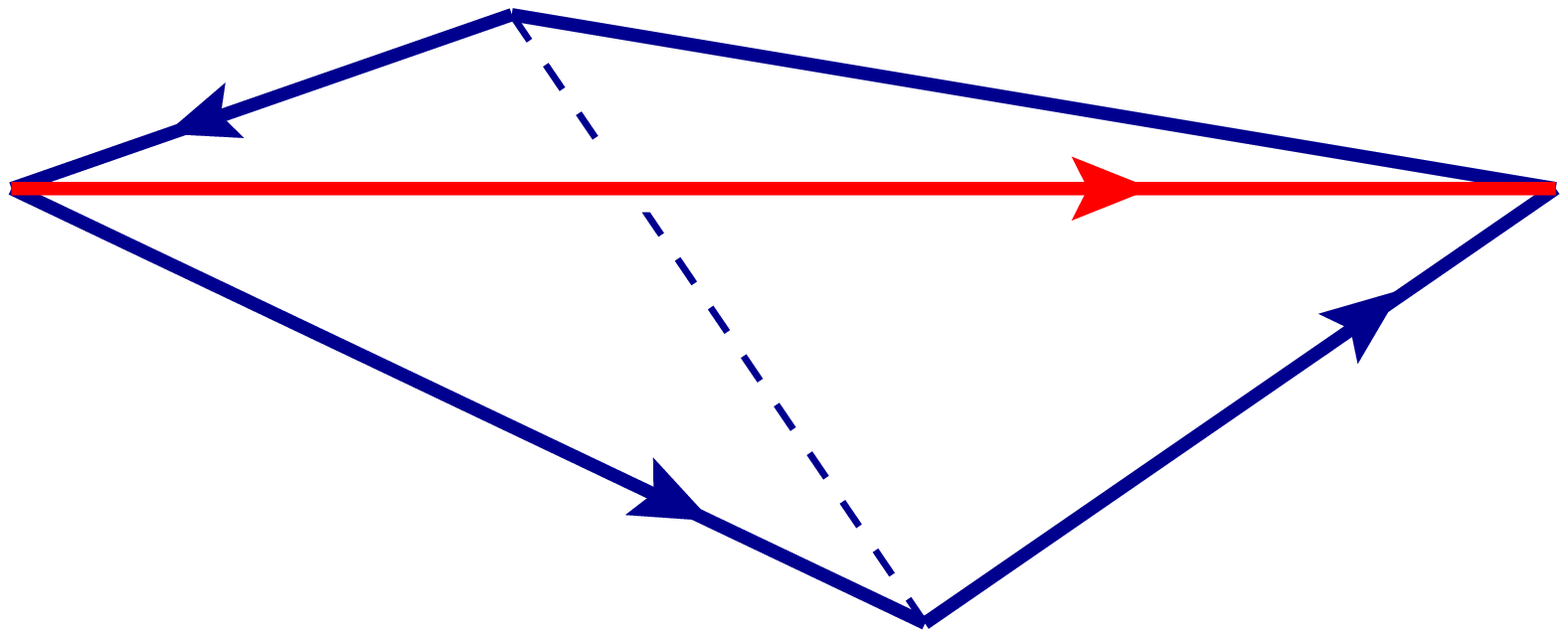}
  \put(24,56.7) {\scriptsize $ g_1^{} $}
  \put(53,7.2)  {\scriptsize $ g_2^{} $}
  \put(84,35.7) {\scriptsize $ g_3^{}g_2^{} $}
  \put(122,20.4){\scriptsize $ g_3^{} $}
  }
  \put(300,54)  {\Includepic{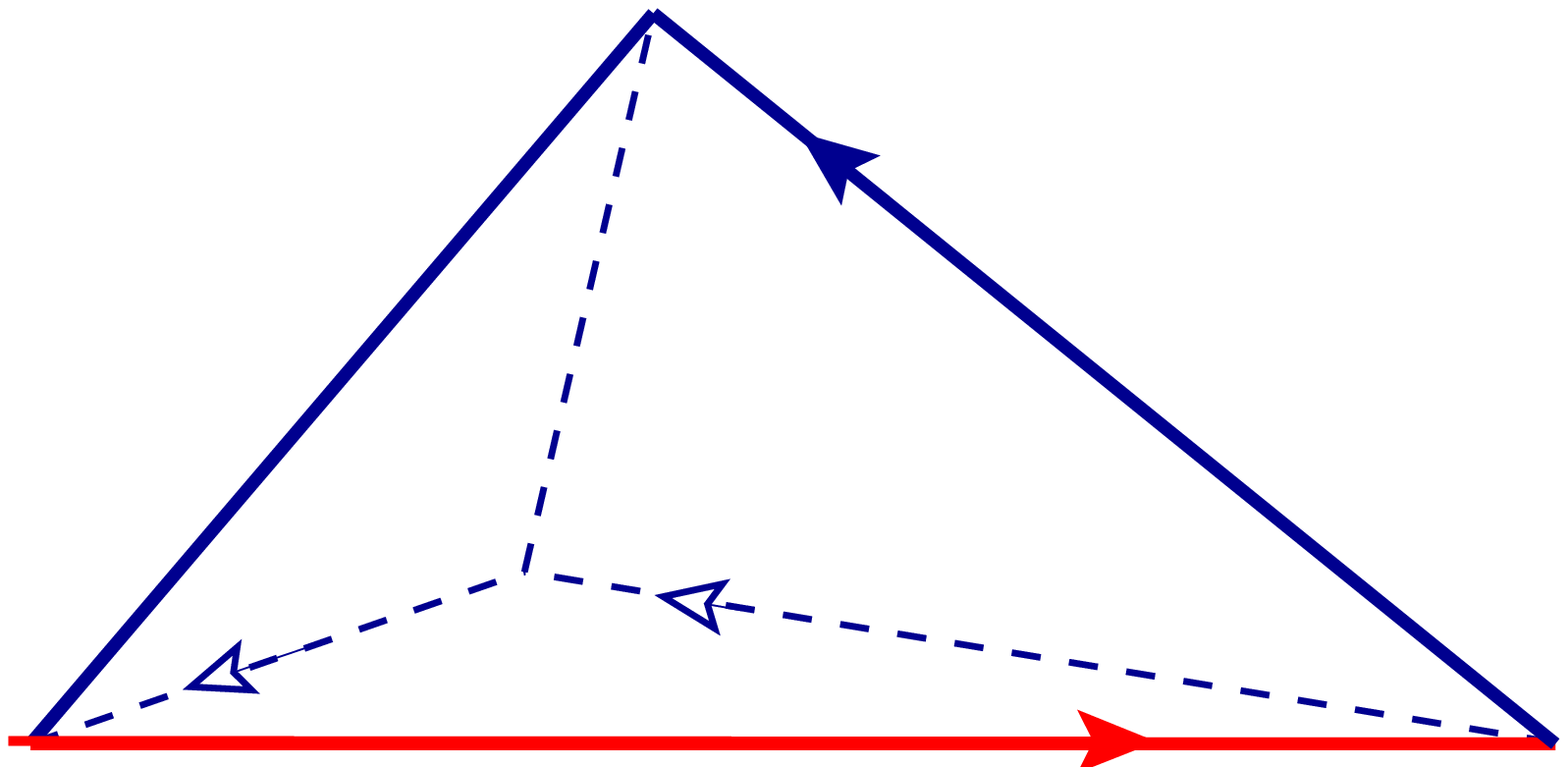}
  \put(23,15.7) {\scriptsize $ g_1^{} $}
  \put(84,57.2) {\scriptsize $ g_4^{} $}
  \put(71,20.2) {\begin{turn}{-11}\scriptsize$ g_3^{}g_2^{}g_1^{} $ \end{turn}}
  \put(86,-3.8) {\scriptsize $ g_3^{}g_2^{} $}
  } }
i.e.\ the three tetrahedra share an edge (the one labeled by $g_3^{}g_2^{}$)
which intersects transversally the shaded face in \erf{Figure7} and pairwise 
share one of three faces which have the shared edge as a boundary segment. 

The two decompositions are related by a 3-2 Pachner move. As is well known, 
invariance under this move is guaranteed by the closedness
of $\omega$. Indeed we have

\begin{lem}
The groupoid cocycles obtained from the two decompositions \erf{Figure7} and
\erf{Figure8} coincide.
\end{lem}

\begin{proof}
The decomposition \erf{Figure7} gives the number
  \be
  \tau_1 := \tilde\omega(g_1,g_2,g_4 g_3)\cdot \tilde\omega(g_2 g_1,g_3,g_4) \,,
  \labl{Xi1}
while the decomposition \erf{Figure8} yields
  \be
  \tau_2 := \tilde\omega(g_1,g_2,g_3)\cdot \tilde\omega(g_2,g_3,g_4) \cdot
  \tilde\omega(g_1,g_3g_2,g_4) \,,
  \labl{Xi2}
with the three factors being the contributions from the lower, the front, 
and the back tetrahedron, respectively. Equality of 
$\tau_1$ and $\tau_2$ amounts to
  \be
  \begin{array}{l}
  \omega(g_1^{-1},g_2^{-1},g_3^{-1} g_4^{-1})
    \cdot \omega(g_1^{-1}g_2^{-1},g_3^{-1},g_4^{-1})
  \Nxl3 
  \hsp 8
  =\omega(g_1^{-1},g_2^{-1},g_3^{-1}) \cdot \omega(g_2^{-1},g_3^{-1},g_4^{-1}) \cdot
  \omega(g_1^{-1},g_2^{-1} g_3^{-1},g_4^{-1}) 
  \,.
  \end{array}
  \labl{oo=ooo}
This is nothing but the statement that $\omega$ is closed, and is thus
indeed satisfied.
\end{proof}
 
The second situation to be analyzed corresponds to a 4-1 Pachner move. It can
be treated in an analogous manner as the 3-2 move; we leave the details to 
the reader.

\medskip

Let us briefly comment on the particular case of the circle without insertions.
According to Section \ref{ss3.1}, in this case the action groupoid is 
$G\srrad G$ with the adjoint action. This situation is described by the simplex
  \eqpic{Figure9}{120}{33} {
  \put(0,0)   {\Includepic{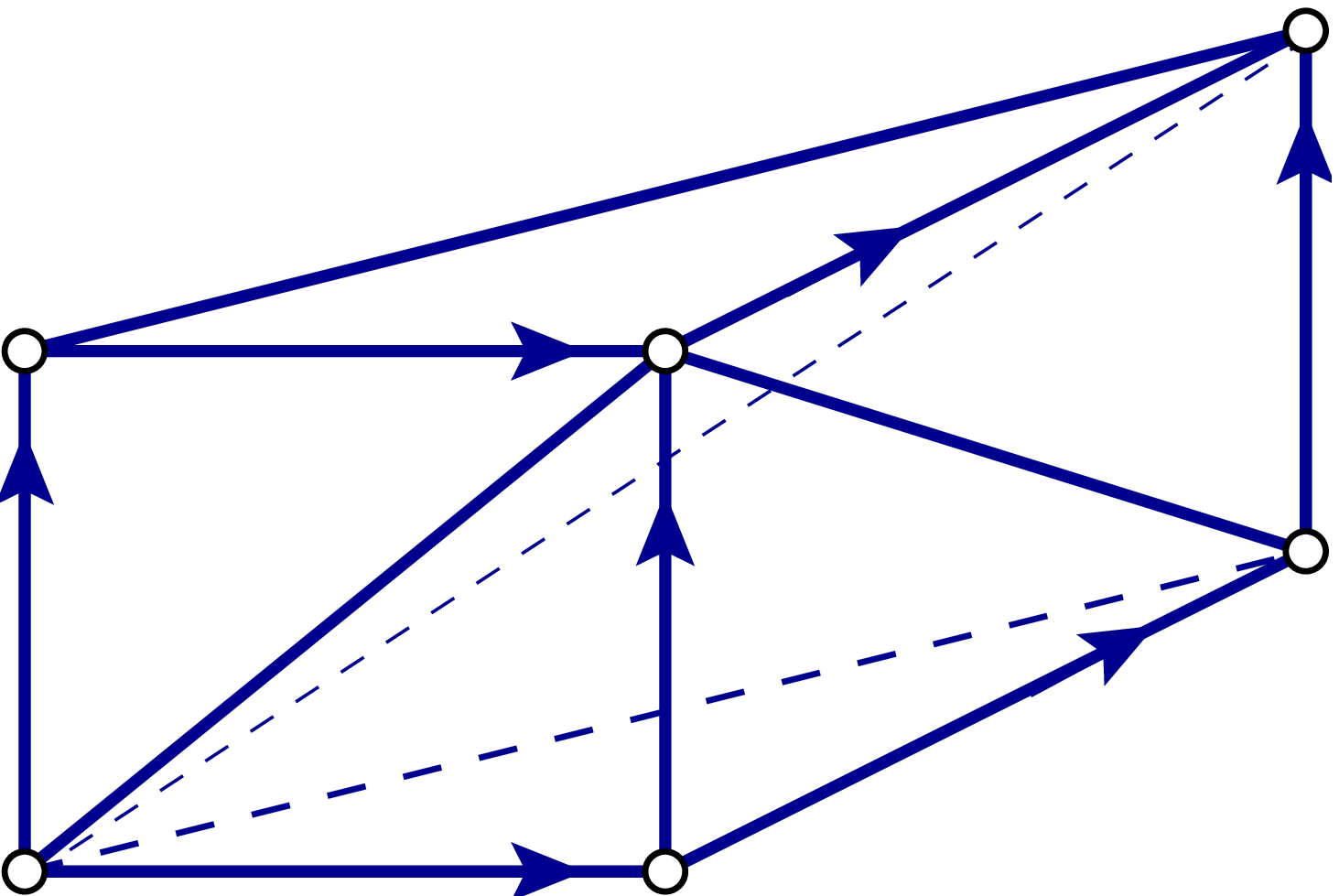}
  \put(-5,33)   {\scriptsize $ \gamma $}
  \put(41,49.2) {\scriptsize $ g_1^{} $}
  \put(42,-3.5) {\scriptsize $ g_1^{} $}
  \put(69.3,36) {\scriptsize $ {}^{g_1^{}\!}\gamma $}
  \put(76,66.6) {\scriptsize $ g_2^{} $}
  \put(104,16)  {\scriptsize $ g_2^{} $}
  \put(134,65)  {\scriptsize $ {}^{g_2^{}g_1^{}\!}\gamma $}
  } }
where we indicate the adjoint left action by a superscript,
${}^g\gamma \eq g\gamma g^{-1}$. This yields the cocycle
  \be
  \tau(g_1,g_2;\gamma) = \tilde\omega(g_1,g_2,{}^{g_2^{}g_1^{}\!}\gamma) \,
  \tilde\omega(g_1,{}^{g_2^{}\!}\gamma,g_2)^{-1}\, \tilde\omega(\gamma,g_1,g_2) \,.
  \labl{325}
This way we precisely recover the argument given in \cite{willS} 
that leads to the 2-cochain found in \cite[(3.2.5)]{dipR}.
Our formalism thus produces the correct category of bulk Wilson lines.

\medskip

We next consider the case of an interval with no marked interior points. The 
interior is labeled by a finite group $G$ and $\omega\iN Z^3(G;\complexx)$, 
while the end points are labeled by group homomorphisms
$\iota\colon H_i \To G$ and by 2-cochains $\theta_i\iN C^2(H_i,\complexx)$ 
such that $\rmd\theta_i \eq \iota_i^*\omega$.

Again there is the issue of non-uniqueness of simplicial decomposition,
with the new aspect that the boundary of the interval leads to the 
presence of triangles of the form \erf{Figure6} in the decompositions.
Thus we must consider tetragons
  \eqpic{Figure10o}{120}{18} {
  \put(0,0)     {\Includepic{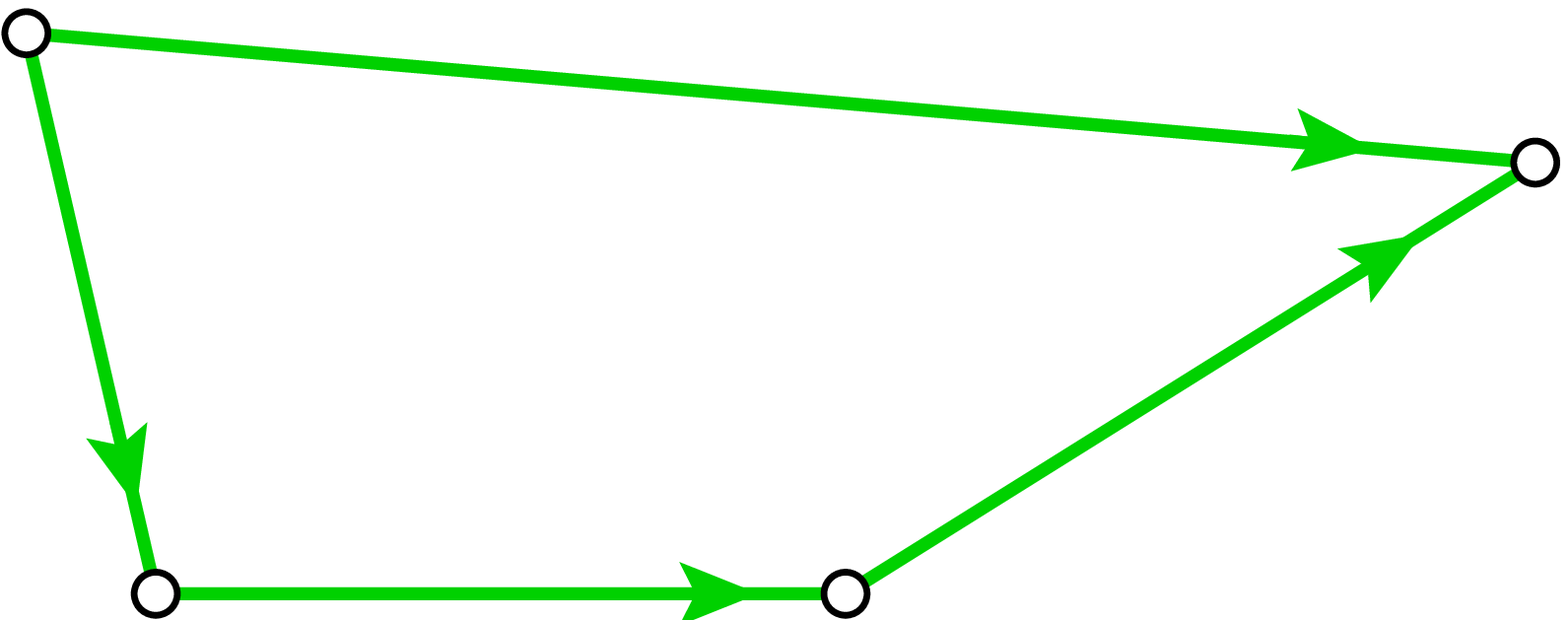}
  \put(0,20)    {\scriptsize $ h_1^{} $}
  \put(57,-5.3) {\scriptsize $ h_2^{} $}
  \put(122,22)  {\scriptsize $ h_3^{} $}
  } }

Such a boundary tetragon can be decomposed into triangles in two different ways: as
  \eqpic{Figure10}{380}{19} {
  \put(0,0)     {\Includepic{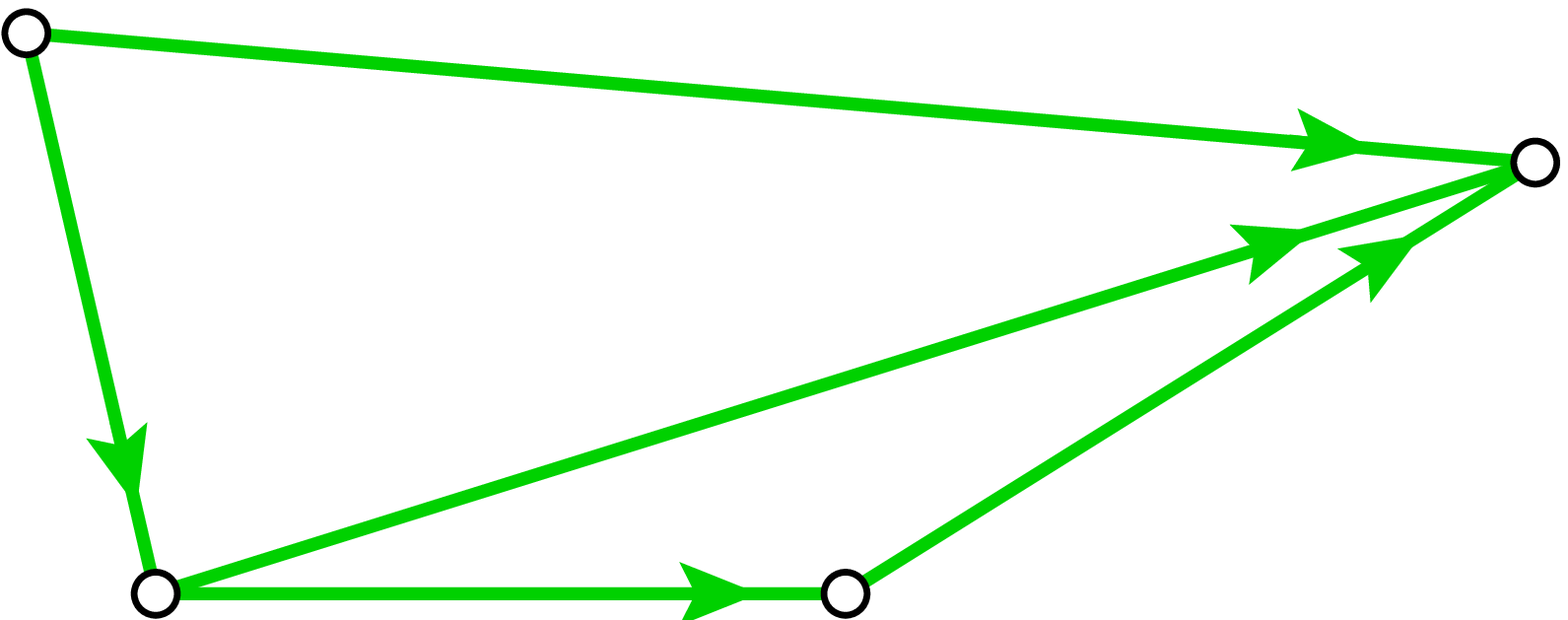}
  \put(0,20)    {\scriptsize $ h_1^{} $}
  \put(57,-5.3) {\scriptsize $ h_2^{} $}
  \put(94.4,32) {\begin{turn}{18.4}\scriptsize $ h_3^{} h_2^{} $\end{turn}}
  \put(93.7,52.5) {\begin{turn}{-6.5}\scriptsize $ h_3^{} h_2^{} h_1^{} $\end{turn}}
  \put(122,22)  {\scriptsize $ h_3^{} $}
  }
  \put(173,20)  {and as}
  \put(232,0)   {\Includepic{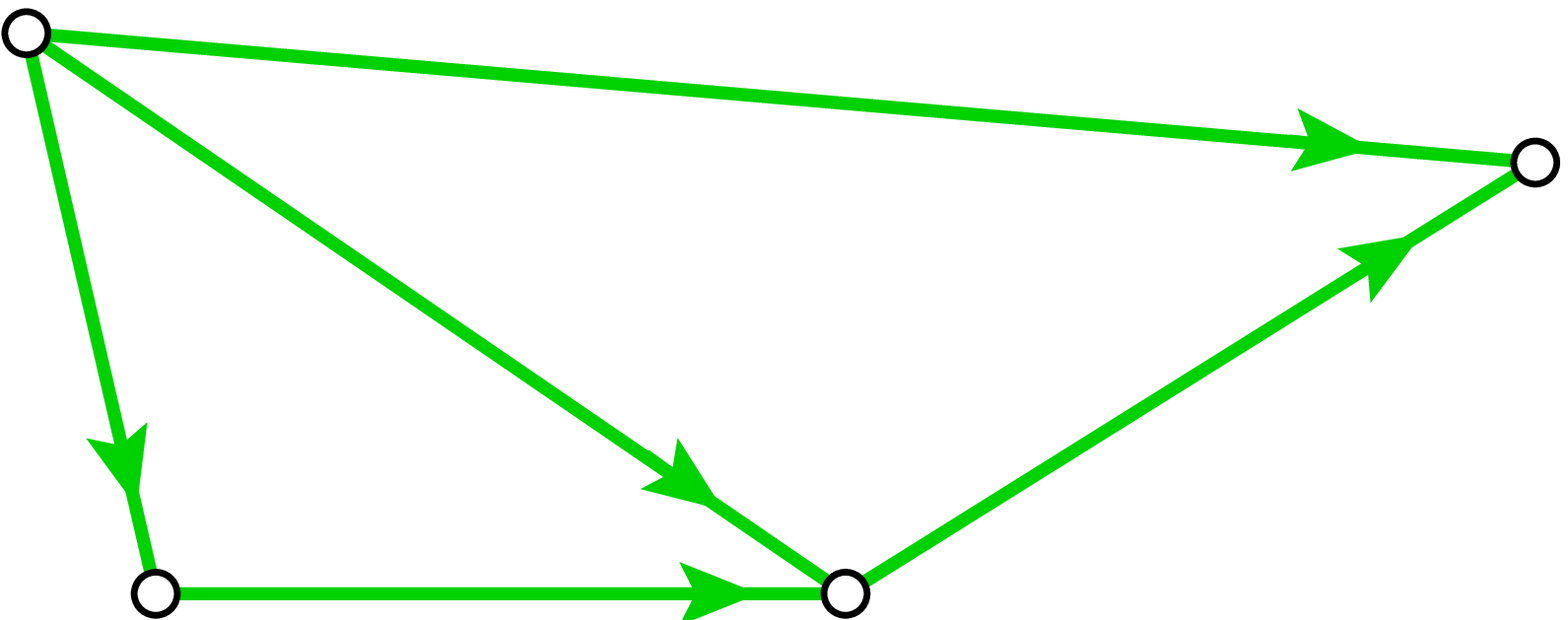}
  \put(0,20)    {\scriptsize $ h_1^{} $}
  \put(47.1,29) {\begin{turn}{-34.5}\scriptsize $ h_2^{} h_1^{} $\end{turn}}
  \put(57,-5.3) {\scriptsize $ h_2^{} $}
  \put(93.7,52.5) {\begin{turn}{-6.5}\scriptsize $ h_3^{} h_2^{} h_1^{} $\end{turn}}
  \put(122,22)  {\scriptsize $ h_3^{} $}
  } }

We compare these two decompositions by continuing the situation 
to the interior of the interval. This leads to the two
simplicial decompositions
  \Eqpic{Figure11Figure12}{410}{66} {
  \put(-15,-14) {\Includepic{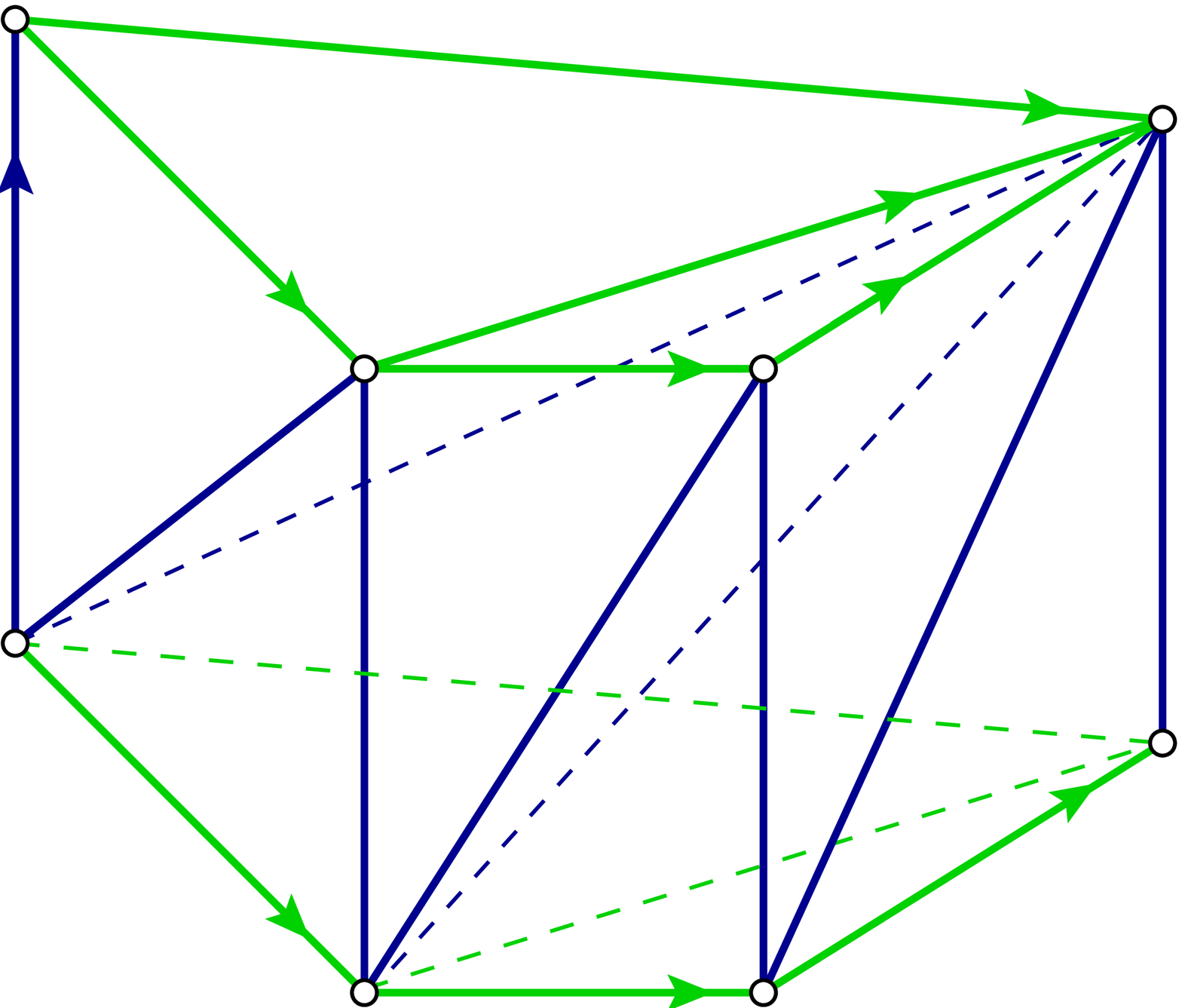}
  \put(-5.1,126){\scriptsize $ \gamma $}
  \put(37,9.8)  {\scriptsize $ g_1^{} $}
  \put(45,121.7){\scriptsize $ h_1^{} $}
  \put(121,128.8){\begin{turn}{18}\scriptsize $ h_3^{}h_2^{} $ \end{turn}}
  \put(102,-4.2){\scriptsize $ g_2^{} $}
  \put(104,95.3){\scriptsize $ h_2^{} $}
  \put(139,105.8){\scriptsize $ h_3^{} $}
  \put(169,24.9){\scriptsize $ g_3^{} $}
  }
  \put(201,65)  {and}
  \put(248,-14) {\Includepic{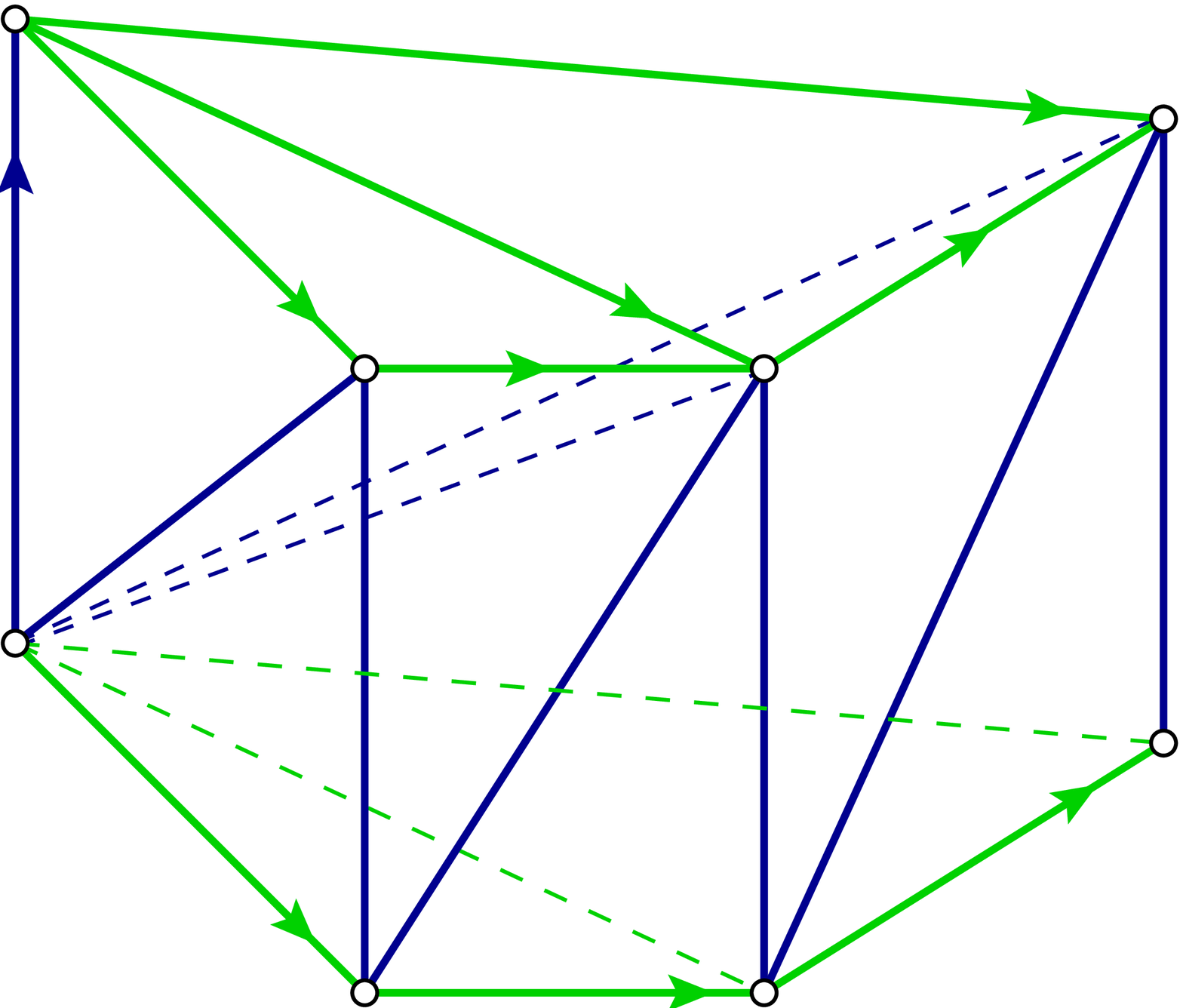}
  \put(-5.1,126){\scriptsize $ \gamma $}
  \put(37,9.8)  {\scriptsize $ g_1^{} $}
  \put(45,121.7){\scriptsize $ h_1^{} $}
  \put(83,127.7){\begin{turn}{-24}\scriptsize $ h_2^{}h_1^{} $ \end{turn}}
  \put(102,-4.2){\scriptsize $ g_2^{} $}
  \put(73,109.2){\scriptsize $ h_2^{} $}
  \put(152,113.2){\scriptsize $ h_3^{} $}
  \put(169,24.9){\scriptsize $ g_3^{} $}
  } }
respectively, each consisting of six tetrahedra and of two triangles at the top. 

\begin{prop}
The complex numbers
obtained from the two decompositions in 
\erf{Figure11Figure12} coincide.
\end{prop}

\begin{proof}
Of the six tetrahedra appearing in the two simplices \erf{Figure11Figure12},
only two are different: the ones attached to the top. The simplex on the right 
hand side of \erf{Figure11Figure12} gives factors $\tilde\theta(h_1,h_2)$ and 
$\tilde\theta(h_2h_1,h_3)$ from the triangles at the top and
  \be
  \tilde\omega(\gamma,\iota(h_1),\iota(h_2)) \cdot 
  \tilde\omega(\gamma,\iota(h_2 h_1),\iota(h_3)) \,.
  \ee
from the two tetrahedra attached to the top triangle, while for the simplex 
on the left hand side we get $\tilde\theta(h_2,h_3) \, \tilde\theta(h_1,h_3h_2)$
from the top triangles and 
  \be
  \tilde\omega(\gamma,\iota(h_1),\iota(h_3 h_2)) \cdot 
  \tilde\omega(\iota(h_1)\gamma,\iota(h_2),\iota(h_3))
  \ee
from the attached tetrahedra. Equality of the two expressions yields, after
implementing the closedness \erf{oo=ooo} of $\omega$,
  \be
  \tilde\theta(h_1,h_2) \, \tilde\theta(h_2h_1,h_3)
  = \tilde\theta(h_2,h_3) \, \tilde\theta(h_1,h_3h_2) \,
  \tilde\omega(\iota(h_1),\iota(h_2),\iota(h_3)) \,,
  \ee
or, what is the same
  \be
  \mathrm d\theta(h_1^{-1},h_2^{-1},h_3^{-1})
  = \omega(\iota(h_1^{-1}),\iota(h_2^{-1}),\iota(h_3^{-1})) \,. 
  \ee
This indeed holds true, owing to $\rmd \theta \eq \iota^*\omega$.
\end{proof}

\subsection{Wilson line categories for the interval}\label{ss:int}

As already pointed out, by invoking fusion of defects (and of defects to
boundaries), among the one-dimensional manifolds there are two fundamental 
building blocks, the interval without interior marked points and the circle
with a single marked point. We now turn to the computation of the categories 
for these building blocks 
and then compare them to the model-independent results of \cite{fusV}.
In the present subsection we consider an interval 
without interior marked points. The interior is
labeled by $(G,\omega)$ with $G$ a finite group and $\omega$ a
3-cocycle. For the two boundary points we have group homomorphisms
$\iota_i\colon H_i\to G$ and 2-cochains $\theta_i$ on $H_i$ such that
$\iota^*_i\omega \eq \rmd\theta_i$ for $i \eq 1,2$. 

Before computing the linearization, we outline what the general formalism of
\cite{fusV} predicts for the situation at hand: The data associated to a 
boundary leads to module categories $\calm_i$ over the fusion category 
$\GVect^\omega$. Such a module category can be decomposed into indecomposable 
module categories. As described in Section \ref{ss:mod}, an indecomposable 
module category over $\GVect^\omega$ can, in turn,
be concretely described \cite{ostr} as the category of modules over
an algebra in $\GVect^\omega$. Thus for the description
of $\calm_i$ it suffices to know such an algebra $A_{H,\theta}$ for any
subgroup $H\le G$ and 2-cochain $\theta$ on $H$ satisfying
$\rmd \theta \eq \omega|_H$. As seen in Section \ref{ss:mod}. such 
algebras can be described as follows. Isomorphism classes of simple objects
in $\GVect^\omega$ are in bijection with elements $g\iN G$; we fix
a set of representatives $(U_g)_{g\in G}$. Then $A_{H,\theta}$ is the object
$\bigoplus_{h\in H\!} U_h$ endowed with the multiplication morphism that is
furnished by the cochain $\theta$. This multiplication is associative, due to 
the relation $\rmd\theta \eq \iota^*\omega$. Then the category
$\calm_{H,\theta} \,{:=}\, A_{H,\theta}\Mod$ is a right module category
over $\GVect^\omega$.

By the results of \cite{fusV}, such a module category corresponds
to an indecomposable boundary condition of the Dijkgraaf-Witten theory 
based on $(G,\omega)$. Given two such boundary conditions, 
consider the abelian category
  \be
  \calf := \Fun_{\GVect^\omega}(A_{H_2,\theta_2}\Mod,A_{H_1,\theta_1}\Mod)
  \ee
of module functors. It has the following physical interpretation:
Objects of \calf\ label boundary Wilson lines separating the boundary condition
$\calm_{H_1,\theta_1}$ from $\calm_{H_2,\theta_2}$. Morphisms of \calf\
label point-like insertions on such Wilson lines. \calf\ can be described as the
category of $A_{H_1,\theta_1}$-$A_{H_2,\theta_2}$-bimodules in $\GVect^\omega$.

The objects $M \eq \bigoplus_{g\in G}M_g$ of the category of 
$A_{H_1,\theta_1}$-$A_{H_2,\theta_2}$-bimodules have been described explicitly 
in \cite[Prop.\,3.2]{ostr5}: Taking into account that the tensor product on
$\GVect^\omega$ realizes the group law strictly, i.e.\ $U_h\oti U_g \eq U_{hg}$,
the restriction of the left action of $A_{H_1,\theta_1}$ on
$M$ to $U_{h_1}\oti U_g$ leads to an endomorphism of $U_{h_1g}$ which is a 
multiple $\rho(h_1,g) \iN \complex$ of the identity. Analogously the right 
action of $A_{H_2,\theta_2}$ gives us scalars $\ohr(g,h_2)\iN\complex$.
These scalars obey the following conditions.
 \beginitemize
 \item
That we have a left $A_{H_1,\theta_1}$-action amounts to the relation
  \be
  \rho(h_1' h_1^{},g)
  = \theta_1^{}(h_1',h_1^{})^{-1}_{} \, \omega(h_1',h_1^{},g) \,
  \rho(h_1,g) \, \rho(h_1',h_1^{}g)
  \labl{rhoohr1}
for all $g\iN G$ and all $h_1^{},h_1'\iN H_1^{}$.

 \item
Similarly the right $A_{H_2,\theta_2}$-action gives
  \be
  \ohr(g,h_2^{} h_2')
  = \theta_2^{}(h_2^{},h_2')^{-1}_{} \, \omega(g,h_2^{},h_2')^{-1}_{} \,
  \ohr(g,h_2^{}) \, \ohr(gh_2^{},h_2')
  \labl{rhoohr2}
for all $g\iN G$ and all $h_2^{},h_2'\iN H_2^{}$.

 \item
The condition that left and right actions commute amounts to
  \be
  \rho(h_1,g) \, \ohr(h_1 g,h_2)
  = \omega(h_1,g,h_2) \,\, \ohr(g,h_2) \, \rho(h_1,g h_2)
  \labl{rhoohr3}
for all $g\iN G$, $h_1\iN H_1$ and $h_2\iN H_2$. 

 \item
Finally the unitality of the actions implies the two constraints
  \be
  \rho(e,g) = 1 = \ohr(g,e) 
  \labl{rhoohr4}
for all $g\in G$.
 \end{itemize}
(Note that $\theta_1$ and $\theta_2$ are normalized because the algebras are 
strictly unital; \erf{rhoohr4} corresponds to $\omega$ being normalized as
well.) The objects in the category \calf\ of 
$A_{H_1,\theta_1}$-$A_{H_2,\theta_2}$-bimodules  are thus $G$-graded vector 
spaces together with two functions $\rho$ and $\ohr$ 
that obey the constraints \erf{rhoohr1}\,-\,\erf{rhoohr4}. Morphisms of 
\calf\ are $G$-homogeneous maps, commuting with the actions.

We may also consider, for given $\gamma \iN G$, the group
  \be
  H_\gamma := \{ (h_1,h_2) \iN H_1\Times H_2 \,|\, h_1\gamma \eq \gamma h_2 \} \,.
  \ee
We can identify $H_\gamma$ with a subgroup of $H_1$, which in turn is a subgroup 
of $G$. Then $h \iN H_\gamma$ acts on the homogeneous component  $M_\gamma$ of $M$ 
as a scalar multiple
  \be
  \rog(h) := \rho(h,\gamma)\, \ohr(\gamma,\gamma^{-1}h\gamma)^{-1}_{}
  \ee
of the identity. In view of \erf{rhoohr1}\,-\,\erf{rhoohr3} this gives rise to
a 2-cocycle $\thetagamma$ on $H_\gamma$, given by
  \be
  \bearll
  \thetagamma(h,h') \!\!& := \rog(hh')^{-1}_{} \, \rog(h) \, \rog(h')
  \Nxl3&
  = \rho(hh',\gamma)^{-1}_{} \rho(h,\gamma)\, \rho(h',\gamma)\,\, 
  \ohr(\gamma,\gamma^{-1}hh'\gamma)\, \ohr(\gamma,\gamma^{-1}h\gamma)^{-1}_{}\,
  \ohr(\gamma,\gamma^{-1}h'\gamma)^{-1}_{}
  \Nxl3&
  = \theta_1(h,h') \, \theta_2(\gamma^{-1}h\gamma,\gamma^{-1}h'\gamma)^{-1}_{}\,
  \omega(h,h',\gamma)^{-1}_{}\,
  \omega^{-1}_{}(\gamma,\gamma^{-1}h\gamma,\gamma^{-1}h'\gamma)
  \Nxl2& \hsp9
  \rho(h,\gamma)\, \rho(h,h'\gamma)^{-1}_{}\,
  \ohr(\gamma,\gamma^{-1}h'\gamma)^{-1}_{}\, \ohr(h\gamma,\gamma^{-1}h'\gamma)
  \Nxl3&
  = \theta_1(h,h') \, \theta_2(\gamma^{-1}h'{}^{-1}\gamma,\gamma^{-1}h^{-1}\gamma)\,
  \Nxl2& \hsp9
  \omega(h,h',\gamma)^{-1}_{}\,
  \omega(\gamma,\gamma^{-1}h\gamma,\gamma^{-1}h'\gamma)^{-1}_{}\,
  \omega(h,\gamma,\gamma^{-1}h'\gamma)
  \eear
  \ee
(compare formula (3.1) of \cite{ostr5}). Here in the third equality we invoke 
\erf{rhoohr1} and \erf{rhoohr2}, while the last equality uses \erf{rhoohr3}.

\medskip

We now show that the prescription \erf{G-GG-H1H2} indeed produces
the expected result:

\begin{prop}
Consider the groupoid 
$\Gamma \eq G\sll G\Times G \srr_{\!\iota_1^-\times\iota_2^-} H_1{\times}H_2$
that according to formula \erf{G-GG-H1H2} is assigned to the interval without 
interior marked points. If the group homomorphisms $\iota_i\colon H_i \To G$ 
are subgroup embeddings, then the category that is obtained by the projective 
linearization of $\Gamma$ for the Lagrangian data $\theta_1$, $\theta_2$ and 
$\omega$ is equivalent, as a $\complex$-linear abelian category,
to the category of $A_{H_1,\theta_1}$-$A_{H_2,\theta_2}$-bimodules,
  \be
  [\,G\sll G\Times G\srr_{\!\iota_1^-\times\iota_2^-} H_1{\times}H_2\,,\Vectc\,]
  ^{\theta_1,\theta_2,\omega}
  \,\simeq\, A_{H_1,\theta_1} \mbox- A_{H_2,\theta_2} \mbox- \mathrm{Bimod}
  _{\GVect^\omega} \,.
  \ee
\end{prop}

\begin{proof}
The objects of the groupoid in question are pairs $(\gamma_1,\gamma_2)$ of
elements of $G$; they label the vertical edges in the following figure:
  \eqpic{FigureA1}{120}{57} {
  \put(0,0)     {\Includepic{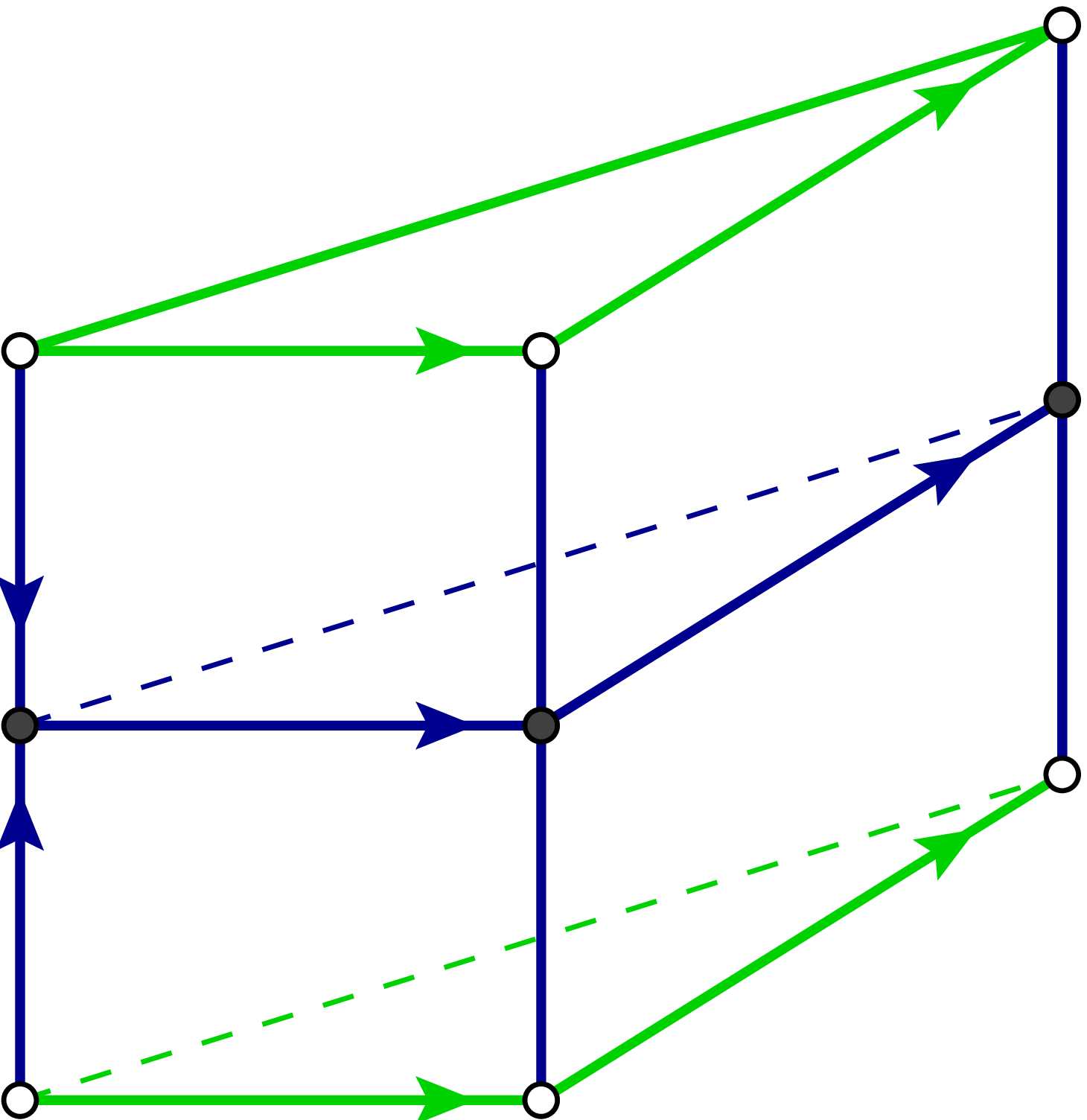}
  \put(-7.5,71) {\scriptsize $ \gamma_1^{} $}
  \put(-8,29)   {\scriptsize $ \gamma_2^{} $}
  \put(29,-6.3) {\scriptsize $ \iota_2^{}(h_2^{}) $}
  \put(28,87.7) {\scriptsize $ \iota_1^{}(h_1^{}) $}
  \put(44,43)   {\scriptsize $ g $}
  \put(97,13.5) {\begin{turn}{32.5}\scriptsize $ \iota_2^{}(h_2') $\end{turn}}
  \put(97,107.5){\begin{turn}{32.5}\scriptsize $ \iota_1^{}(h_1') $\end{turn}}
  \put(108,69)  {\scriptsize $ g' $}
  } }
Morphisms are gauge transformations in $H_1$, $H_2$ and in $G$ -- labeling 
horizontal edges that connect empty circles and filled circles in 
\erf{FigureA1}, respectively. Again we consider a pair of compatible morphisms 
leading to horizontal edges forming the shape of a triangle
to get the relevant 2-cocycle on the groupoid $\Gamma$.
In the sequel we suppress the embedding homomorphisms $\iota_1$ and $\iota_2$.
\\
Observe that the functor 
  \be
  G\sll G\Times G \srr_{\!\iota_1^-\times\iota_2^-} H_1 \Times H_2
  \,\longrightarrow\,
  H_1 \,{} _{\iota_1^{\phantom-}} \hspace{-5.5pt} \sll G \srr_{\!\iota_2^-} H_2
  \ee
that is defined on objects by 
$(\gamma_1^{},\gamma_2^{}) \,{\mapsto}\, \gamma_1^{-1} \gamma_2^{}$
is actually an equivalence of groupoids. 
Accordingly we set $\gamma \,{:=}\, \gamma_1^{-1}\gamma_2^{}$ and obtain from
\erf{FigureA1} a number $\tau(\gamma;h_1^{},h_1';h_2^{},h_2')$ that can be read 
off from the following slice of pie:
  \eqpic{FigureA2-0}{120}{38} {
  \put(0,0)     {\Includepic{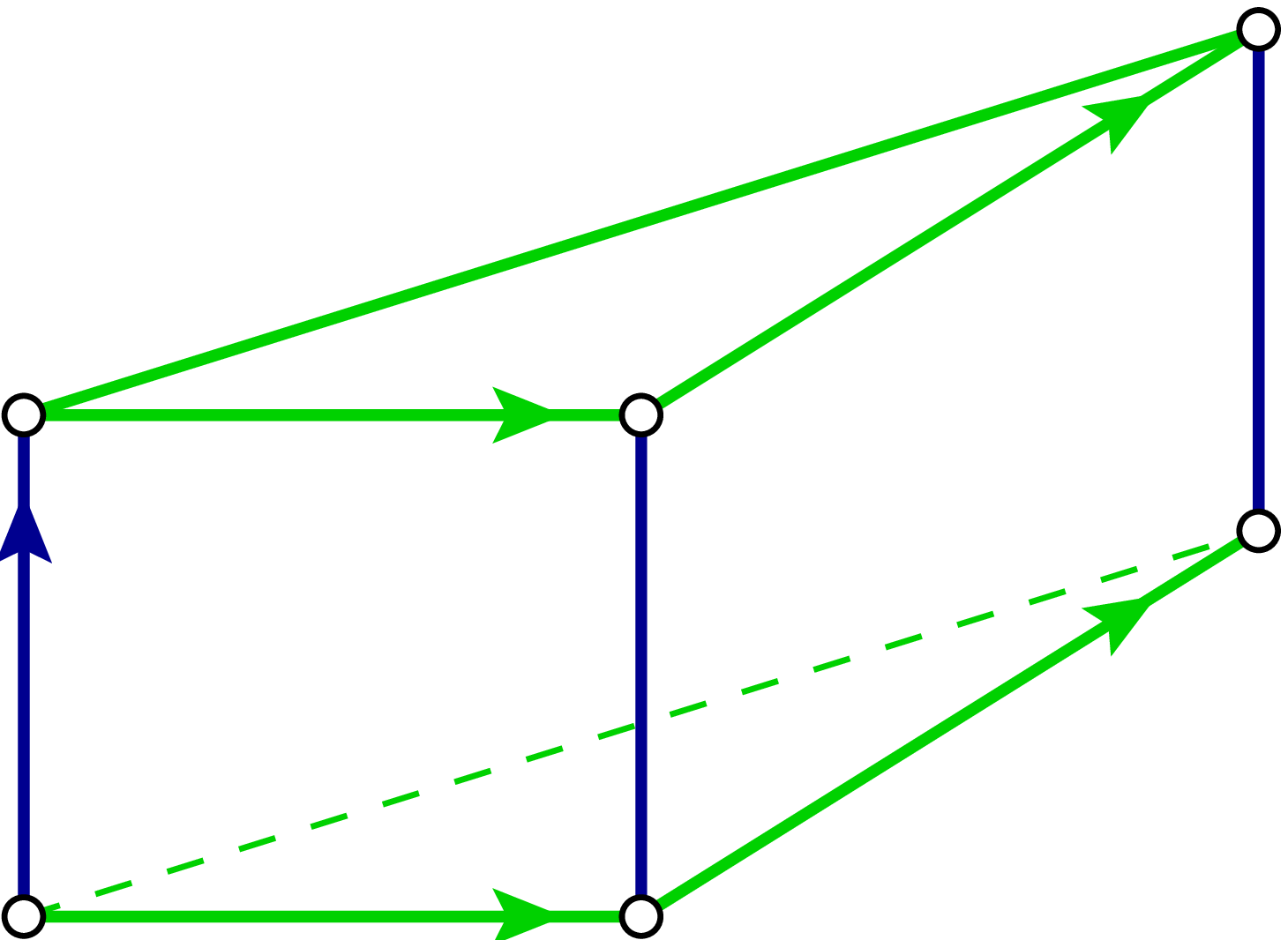}
  \put(-6,29)   {\scriptsize $ \gamma $}
  \put(41,-6)   {\scriptsize $ h_2^{} $}
  \put(41,47)   {\scriptsize $ h_1^{} $}
  \put(109,20)  {\scriptsize $ h_2' $}
  \put(109,73)  {\scriptsize $ h_1' $}
  } }
where $\gamma\iN G$, $h_1^{},h_1 \iN H_1$ and $h_2^{},h_2' \iN H_2$.
There are many equivalent ways to express the so defined numbers in terms
of the 2-cocycles $\theta_i$ and the 3-cocycle $\omega$; they are related 
by the various properties of $\theta_i$ and $\omega$.
Let us choose one such expression that corresponds to the decomposition
  \eqpic{FigureA2}{120}{36} {
  \put(0,0)     {\Includepic{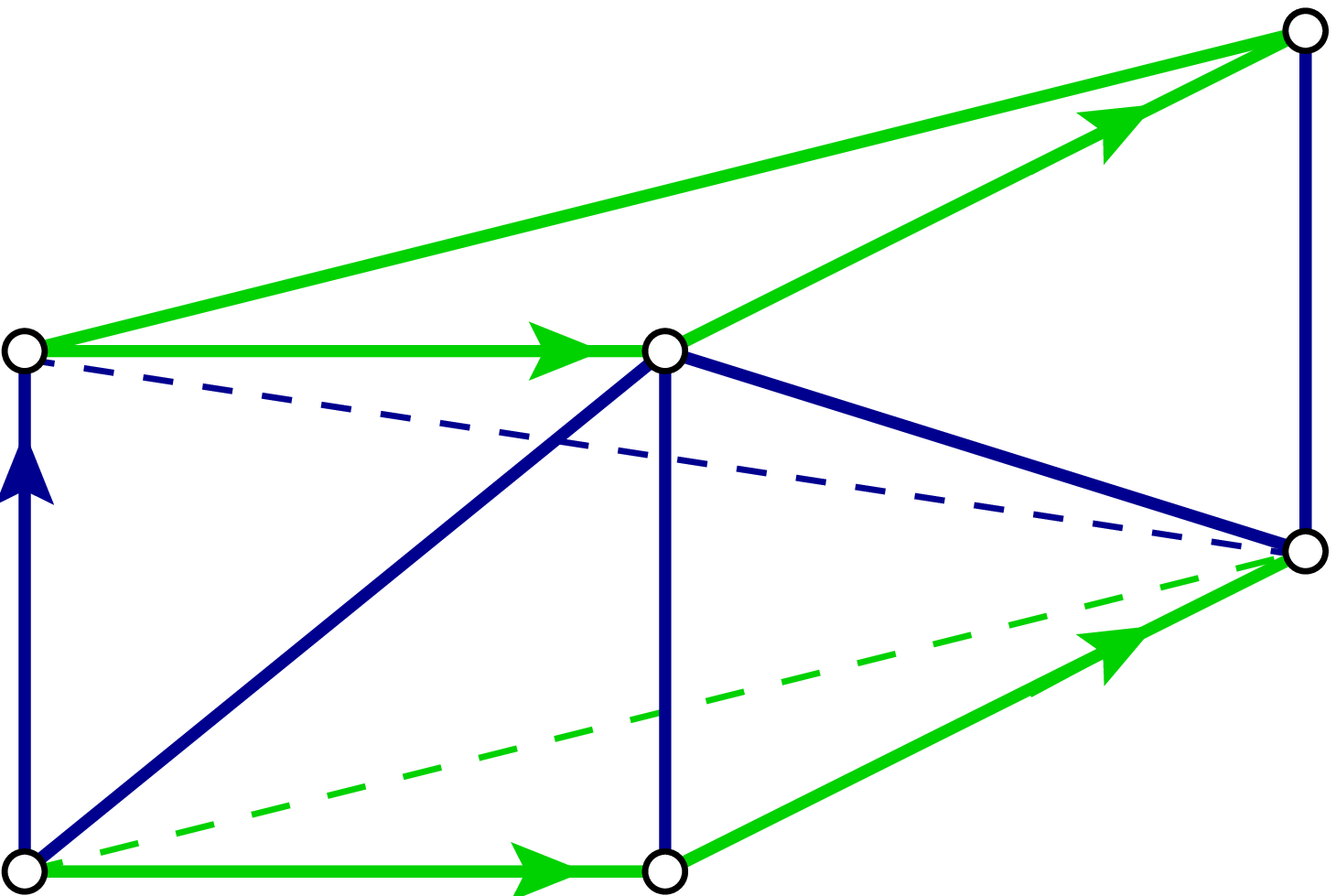}
  \put(-5,34)   {\scriptsize $ \gamma $}
  \put(44,-6)   {\scriptsize $ h_2^{} $}
  \put(46,58.4) {\scriptsize $ h_1^{} $}
  \put(103,13)  {\scriptsize $ h_2' $}
  \put(103,66)  {\scriptsize $ h_1' $}
  } }
of the slice \erf{FigureA2-0} into three tetrahedra. This yields
  \be
  \bearll
  \tau(\gamma;h_1^{},h_1';h_2^{},h_2') \!\!&
  = \tilde\theta_1^{}(h_1^{},h_1')\, {\tilde\theta_2^{}(h_2^{},h_2')} \,
  \tilde\omega(h_1^{},h_1',h_2'h_2^{}\gamma^{-1}h_1^{-1}h_1'{}^{-1}) \,
  \Nxl3 & \hsp5
  \tilde\omega(h_2^{},h_2',h_1^{}\gamma h_2^{-1}h_2'{}^{-1}) \,
  \tilde\omega(\gamma,h_1^{},h_2'h_2^{}\gamma^{-1}h_1^{-1}) \,.
  \eear
  \labl{deftau}
To make contact to the relations \erf{rhoohr1}\,-\,\erf{rhoohr3} for the
category of $A_{H_1,\theta_1}$-$A_{H_2,\theta_2}$-bimodules, we consider
three special cases of $\tau(\gamma;h_1^{},h_1';h_2^{},h_2')$. 
 \beginitemize
 \item
First we set $h_2^{} \eq e \eq h_2'$; then \erf{deftau} reduces to
  \be
  \bearll
  \tau(\gamma;h_1^{},h_1';e,e) \!\!&
  = \tilde\theta_1^{}(h_1'{}^{-1},h_1^{-1})\,
  \tilde\omega(h_1^{},h_1',\gamma^{-1} h_1^{-1}h_1'{}^{-1})\,
  \tilde\omega(\gamma,h_1^{},\gamma^{-1} h_1^{-1})
  \Nxl3 &
  = \tilde\theta_1^{}(h_1'{}^{-1},h_1^{-1})
  \, \tilde\omega(h_1'{}^{-1},h_1^{-1},\gamma^{-1})
  = \theta_1^{}(h_1',h_1^{})^{-1}_{} \, \omega(h_1',h_1^{},\gamma) \,.
  \eear
  \labl{spc1}
This reproduces the factor in the relation \erf{rhoohr1} for the left action
of $H_1$, with $g \eq \gamma$.

\item
Next consider the case $h_1^{} \eq e \eq h_1'$; then we get
  \be
  \bearll
  \tau(\gamma;e,e;h_2^{},h_2') \!\!& = \tilde\theta_2^{}(h_2^{},h_2')\,
  \tilde\omega(h_2^{},h_2',\gamma h_2^{-1}h_2'{}^{-1})
  \Nxl3 &
  = \tilde\theta_2^{}(h_2^{},h_2')\, \tilde\omega(\gamma^{-1},h_2^{},h_2')^{-1}_{}
  = \theta_2^{}(h_2^{-1},h_2'{}^{-1})^{-1}_{} 
  \, \omega(\gamma,h_2^{-1},h_2'{}^{-1})^{-1}_{} .
  \eear
  \labl{spc2}
This is the factor in \erf{rhoohr2}, provided we replace the group
elements $h_2^{}$ and $h_2'$ in \erf{rhoohr2} by their inverses, 
which is precisely what is needed to turn the right action of $H_2$
in \erf{rhoohr2} to the left action considered here.

\item
Finally take $h_1' \eq e \eq h_2'$. This results in
  \be
  \tau(\gamma;h_1^{},e,h_2^{},e)
  = \tilde\omega(\gamma,h_1^{},h_2^{}\gamma^{-1}h_1^{-1})
  = \tilde\omega(h_1^{-1},\gamma^{-1},h_2^{})
  = \omega(h_1^{},\gamma,h_2^{-1}) \,,
  \labl{spc3}
thus reproducing the factor appearing in the bimodule relation \erf{rhoohr3} 
(again upon putting $g \eq \gamma$ and inverting $h_2$).
 \end{itemize}
~\\[-3em] \end{proof}

Notice that the number $\tilde\omega(h_1'{}^{-1},h_1^{-1},\gamma^{-1})$
appearing in the expression \erf{spc1} corresponds to a tetrahedron that
can be viewed as the degeneration of the slice \erf{FigureA2-0} that 
results from the degeneration of its bottom triangle to a single point.
Similarly, $\tilde\omega(\gamma^{-1},h_2^{},h_2')^{-1}_{}$ in \erf{spc2} 
corresponds to the degeneration of the top triangle of \erf{FigureA2-0}
to a point. And the tetrahedron corresponding to 
$\tilde\omega(h_1^{-1},\gamma^{-1},h_2^{})$ in \erf{spc3} can be
obtained by gluing together two quadrangles along their edges which
are obtained from the slice \erf{FigureA2-0} by degenerating both
the top and the bottom triangle to a single edge.

\subsection{The transparent defect}\label{ss:tpd}

We now address aspects of categories associated to DW manifolds with the 
topology of a circle. Recall that one expects that surface defects can be 
fused and should thus form a monoidal bicategory. We refer to the monoidal 
unit of this monoidal bicategory as the \emph{transparent}, or \emph{invisible}
surface defect.  We have already mentioned in Section \ref{sec:b+defTFT} 
that in the framework of \cite{fusV} the transparent surface defect should 
correspond to the canonical Witt trivialization \erf{Wittcan}. In the present 
subsection we are interested in the Lagrangian realization of this 
distinguished surface defect.

To understand what group homomorphism and 2-cocycle 
furnish the transparent defect, we consider a circle with any
number $n$ of surface defects, one of which is transparent. By fusing all 
other surface defects to a single one, we can reduce the situation to the case 
$n \eq 2$. This situation has already been studied in Section \ref{ss3.1}; 
it leads to the groupoid \erf{GG.GGGG.HH}. To realize the transparent defect 
for one of the two marked points we must moreover set $G_> \eq G_< \,{=:}\, G$ 
and take the same 3-cocycle $\omega$ on either side. 
Now we claim that the group homomorphism for the transparent defect 
is the diagonal subgroup embedding, i.e.\ we have to set $H_+ \eq G$ with 
$\iota_+ \eq d \colon G\To G\Times G$ the diagonal embedding. This way we
arrive at the action groupoid 
  \be
  \Gamma_{\!1} := 
  G\Times G \sll G\Times G\Times G\Times G \srr_{d^-\times\iota^-} G\Times H
  \ee
which we already considered in \erf{GG-GGGG-GH}. We further claim that the 
relevant 2-cochain on $H \eq G$ is the constant 2-cochain 
$\theta_d \,{\equiv}\, 1$. Note that this is a valid cochain, as it satisfies
$\rmd \theta_d \eq 1 \eq \omega \cdo \omega^{-1}$.

To see that the defect defined by $\iota \eq d$ and $\theta \eq 1$ indeed has 
the relevant properties of the transparent defect, recall first that
in \erf{F12} we have obtained an equivalence 
$F\colon \Gamma_{\!1} \,{\stackrel\simeq\to}\, \Gamma_{\!2}$ between 
$\Gamma_{\!1}$ and the action groupoid 
  \be
  \Gamma_{\!2} := G\sll G\Times G \srr_{\!\iota^-} H
  \ee
introduced in \erf{G-GG-H}, and that the latter groupoid is precisely the 
one relevant for the circle with a single surface defect of arbitrary type.
Our prescription also yields 2-cocycles 
$\tau_1$ on $\Gamma_{\!1}$ and $\tau_2$ on $\Gamma_{\!2}$. We need to show 
that we still get an equivalence after linearization with respect to 
Lagrangian data. To this end, describe the second defect by $(H,\theta)$
with group homomorphisms $\iota_i\colon H\To G$ and a 2-cochain $\theta$ on 
$H$ satisfying $\rmd\theta \eq (\iota_1^*\omega)\,(\iota_2^*\omega)^{-1}$.
We then have

\begin{prop}
The pullback along the functor $F\colon \Gamma_{\!1}\To\Gamma_{\!2}$ described 
in \erf{F12} yields an equivalence
  \be\begin{array}{rll}
  F^*_{}:\quad [\Gamma_{\!2},\Vect]^{\tau_2}_{}
  &\!\!\stackrel\simeq\longrightarrow\!\!& [\Gamma_{\!1},\Vect]^{\tau_1}_{} \\[3pt]
  \varphi &\!\!\longmapsto\!\!& \varphi\circ F
  \end{array} \ee
of $\complex$-linear abelian categories.
\end{prop}

\begin{proof}
Morphisms in the groupoid $\Gamma_{\!1}$ have the form \erf{g1mor}. Pick two 
such morphisms $(g_1^{},g_2^{},g,h)$ and $(g_1',g_2',g',h')$. Their 
images under $F$ are morphisms $(g_1,h)$ and $(g_1',h')$ in $\Gamma_{\!2}$,
of the form \erf{g2mor}. We must show that
  \be
  \tau_1^{}(\gamma_1^{},\gamma_2^{},\gamma_3^{},\gamma_4^{} ;
  g_1^{},g_2^{},g,h;g_1',g_2',g',h')
  = \tau_2^{}(\gamma_1^{}\gamma_2^{-1}\gamma_3^{},\gamma_4^{};g_1^{},h;g_1',h')
  \labl{tautau}
for all quadruples $(\gamma_1^{},\gamma_2^{},\gamma_3^{},\gamma_4^{})$ of 
elements of $G$.
Both sides of \erf{tautau} are obtained by evaluating appropriate diagrams of
the form of slices of pie with top and bottom faces identified. The diagram 
relevant to $\Gamma_{\!1}$ is similar 
to the one of figure \erf{Figure4}, but now with identified top and bottom,
so that $h_1^{} \eq h_2^{} \,{=:}\, h$ and $h_1' \eq h_2' \,{=:}\, h'$, as 
well as with $h_{12}^{} \eq g$ and $h_{12}' \eq g'$ being now elements of 
$G$; this diagram is shown on the left hand side of the picture \erf{Figure16} 
below. In the case of $\Gamma_{\!2}$ there is, besides the identified top and 
bottom faces, only one horizontal face, with edges labeled by elements $g_1^{}$ 
and $g_1'$ of $G$; this diagram is shown on the right hand side of the picture:
  \eqpic{Figure16}{370}{74} {
  \put(0,0)     {\Includepic{4}
  \put(-8.8,17) {\scriptsize $ \gamma_4^{} $}
  \put(-8.8,51) {\scriptsize $ \gamma_3^{} $}
  \put(-8.8,77) {\scriptsize $ \gamma_2^{} $}
  \put(-8.8,112){\scriptsize $ \gamma_1^{} $}
  \put(46,-5.5) {\scriptsize $ h $}
  \put(43,27.1) {\scriptsize $ g_2^{} $}
  \put(46,57.5) {\scriptsize $ g $}
  \put(43,87.8) {\scriptsize $ g_1^{} $}
  \put(46,116)  {\scriptsize $ h $}
  \put(108,21)  {\scriptsize $ h' $}
  \put(108,53.2){\scriptsize $ g_2' $}
  \put(109,83.7){\scriptsize $ g' $}
  \put(108,114) {\scriptsize $ g_1' $}
  \put(108,142) {\scriptsize $ h' $}
  }
  \put(220,0)     {\Includepic{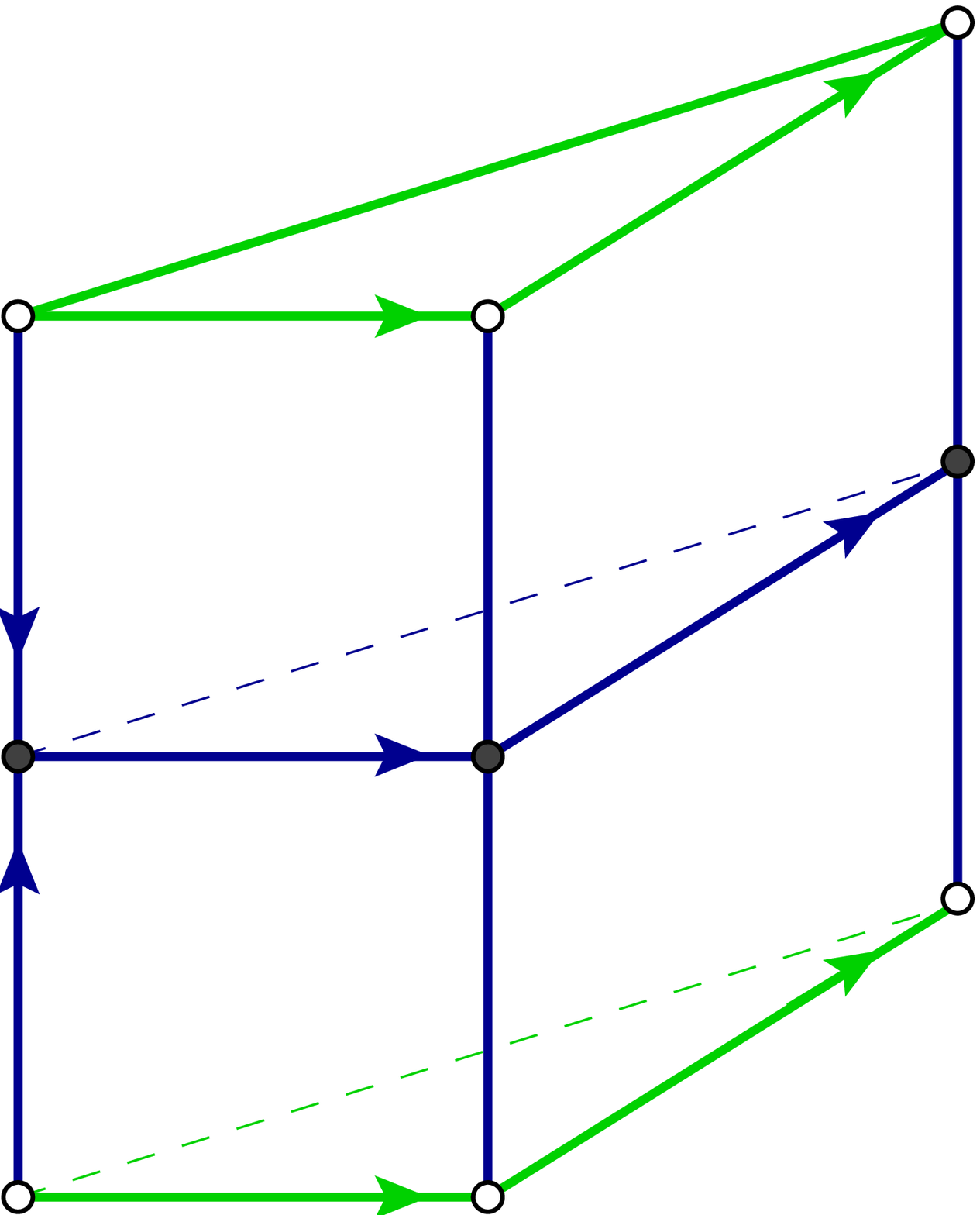}
  \put(-9.8,107){\begin{turn}{-90} \scriptsize
                 $ \gamma_1^{}\gamma_2^{-1}\gamma_3^{} $\end{turn}}
  \put(-8.5,41) {\scriptsize $ \gamma_4^{} $}
  \put(46,55.5) {\scriptsize $ g_{1}^{} $}
  \put(46,-5.5) {\scriptsize $ h $}
  \put(46,116)  {\scriptsize $ h $}
  \put(108,21)  {\scriptsize $ h' $}
  \put(108,82.5){\scriptsize $ g_{1}' $}
  \put(108,142) {\scriptsize $ h' $}
  } }
It should be appreciated that the two diagrams only differ in a part that 
is of the same topology and only involves edges labeled by $G$. 
It is easily seen that there is a sequence of Pachner moves relating the
decompositions of the two diagrams in \erf{Figure16}. And as discussed in
Section \ref{ss:gcc}, invariance under Pachner moves holds (as a direct 
consequence of the axioms of group cohomology)  
for the decomposition of simplices into tetrahedra. Together it follows
that indeed the 2-cocycles on the left and right hand sides of \erf{tautau} 
have the same value.
\end{proof}

To summarize our findings: The surface defect labeled by $\iota \eq d$ and 
$\theta \eq 1$ can be omitted without changing the category that our 
linearization procedure associates to the circle. In other words, it
has the characteristic property of the monoidal unit for the fusion of 
surface defects, and thus of the transparent defect.

\subsection{Wilson line categories for the circle}\label{ss:cir}

A one-dimensional
DW manifold with the topology of a circle can contain finitely many marked 
points, corresponding to surface defects. Invoking fusion of defects, the 
situation with any number of marked points can be reduced to the one with a 
single marked point, which thereby constitutes one of the two
fundamental building blocks. In this subsection we finally compute the category 
of generalized Wilson lines corresponding to this building block and compare 
it with the results of \cite{fusV} for defects in topological field theories
of Reshetikhin-Turaev type.

Let, as before, the subinterval be labeled by $(G,\omega)$ and the defect 
by a group homomorphism $\iota\colon H\To G\Times G$ and a 2-cochain $\theta$ 
on $H$ satisfying $\rmd\theta \eq (\iota_1^*\omega)\,(\iota_2^*\omega)^{-1}$.
We can restrict our attention to indecomposable defects and therefore assume 
that $\iota$ is injective. For this situation
our formalism yields in a straightforward manner the groupoid
  \be
  G\sll G\times G \srr_{\!\iota^-_{}} H
  \labl{:G-GG-H}
that we already encountered in \erf{G-GG-H}. Its (projective) linearization,
which we denote by \wht, is the abelian category
of $G\Times G$-graded vector spaces with two commuting left actions (which are,
in general, projective): a left action of $G$ such that 
$g\,.\,V_{\gamma_1^{}\gamma_2^{}} \,{\subseteq}\,V_{g\gamma_1^{},g\gamma_2^{}}$
and a left $H$-action such that $h\,.\,V_{\gamma_1^{}\gamma_2^{}}\,{\subseteq}\, 
V_{\gamma_1^{}\iota_1^{}(h)^{-1}_{},g\gamma_2^{}\iota_2^{}(h)^{-1}_{}}$.

The category \wht\ has the interpretation of the category of generalized Wilson 
lines separating the defect labeled by $\iota$ and $\theta$ from the
transparent defect that we studied in the previous subsection. Pictorially,
fusion of surface defects replaces the configuration depicted on the right 
hand side of figure \erf{Figure17+Figure18}, in which four surface defects 
meet in a generalized Wilson line, by the configuration shown in the following
picture, in which the single non-trivial defect is on the right and the 
transparent defect on the left:
  \eqpic{Figure19}{140}{48} {
  \put(0,0) {\Includepic{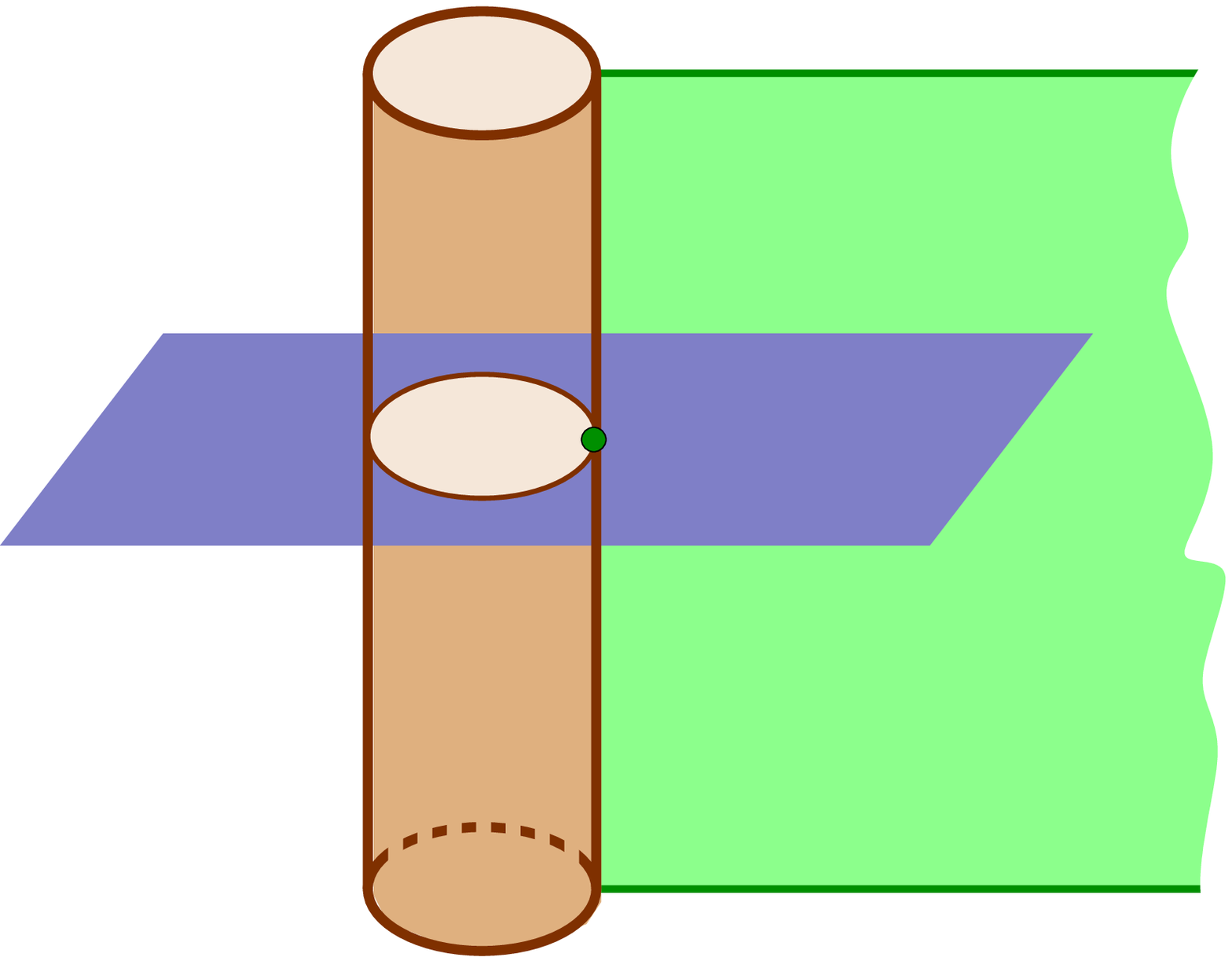}   
  \put(79,85) {$ (H,\theta) $}
  } }

\medskip

We claim that the category produced by our geometric prescription is the same 
as the Wilson line category that is obtained in the formalism of \cite{fusV}. 
Let us thus compute the latter. According to formula \erf{Wtriv}, in the 
framework of \cite{fusV} a surface defect is described by a Witt 
trivialization. Now in the case of Dijkgraaf-Witten theories 
already the modular tensor category \C\ of bulk Wilson lines is, by definition,
Witt trivial. Indeed, $\C \eq \Z(\cala)$, where for the theory
based on $(G,\omega)$, \cala\ is the fusion category of finite-dimensional 
$G$-graded vector spaces with associativity constraint given by $\omega$ as in 
\erf{aomega}. It is not difficult to verify that the Witt trivialization of
\C\ implies the Witt trivialization
  \be
  \C\boti \C^\rev_{} \stackrel\simeq\longrightarrow \Z(\cala\boti\cala\op_{}) \,,
  \labl{WtAA}
where $\cala\op_{}$ is the fusion category \cala\ with opposite tensor product.

\medskip

Indecomposable surface defects separating the  modular tensor category 
$\C \eq \Z(\cala)$ from itself correspond \cite{fusV} to indecomposable module 
categories over $\cala\boti\cala\op_{}$ which is, as an abelian category, the 
category of $G\Times G$-graded vector spaces. According to the results 
reported in Section \ref{ss:mod}, such a module category is described by a 
subgroup $H \,{\le}\, G\Times G$ and a 2-cochain $\theta$ on $H$. This 
category can be realized as $\calm_{H,\theta} \eq A_{H,\theta}\Mod$, with the 
algebra $A_{H,\theta}$ as introduced in Section \ref{ss:mod}. The category 
$\calm_{H,\theta}$ of $A_{H,\theta}$-modules, seen as a module
category over $\cala\boti\cala\op_{}$, describes the non-transparent 
surface defect in the situation we are considering.

The analogous algebra in $\cala\boti\cala\op_{}$ that is relevant for the 
transparent defect can be deduced from the discussion in Section \ref{ss:tpd}: 
it is the algebra $A_d$ for the diagonal subgroup $G \,{\le}\, G\times G$ with 
trivial 2-cocycle $\theta \eq 1$. The category of Wilson lines described by the 
linearization of the groupoid \erf{:G-GG-H} should therefore be matched to the 
category 
  \be
  \Hom_{\cala\boti\cala\op_{}} (\cala,\calm_{H,\theta})
  \ee
of module functors or, equivalently, to the category of 
$A_d$-$A_{H,\theta}$-bimodules in $\cala\boti\cala\op_{}$. But the latter is 
nothing else than the category of $G\Times G$-graded vector spaces together 
with projective actions of $H$ and $G$.
 
This concludes the match of the categories that are obtained, for the case of 
the circle, in the present geometric approach and in \cite{fusV}.

\newpage

\appendix

\section{Module categories for non-injective group homomorphisms}

As described in Section \ref{ss:mod}, indecomposable module categories over 
the fusion category $\GVect$ are given by subgroups $H\,{\le}\, G$ and 
group cochains. On the other hand, in 
the definition of relative bundles a group homomorphism $\iota\colon H\To G$
enters. In the geometric context, it is not natural, and for many purposes,
e.g.\ for the discussion of fusion of surface defects, not appropriate, to 
require $\iota$ to be injective. This raises the question how corresponding 
module categories decompose into indecomposable ones if the group homomorphism 
$\iota$ is not injective. We discuss this issue in the simplest setting, 
in particular dropping Lagrangian data.

\medskip

We consider a morphism $\iota\colon H\To G$ of finite groups and the action 
groupoid $G \srr_{\!\iota^-} H$ with left action $h.\gamma \eq \gamma\, 
\iota(h)^{-1}$. The functor category $\calm\,{:=}\,[G \srr_{\!\iota^-} H,\Vect]$
is a module category over the monoidal category $\GVect$ as follows.
Objects in $\calm$ are $G$-graded vector spaces
$V \eq \bigoplus_{g\in G} V_g$ endowed with a left action of $H$ such that
  \be
  h. V_g\subset V_{g\cdot \iota(h)^{-1}} .
  \ee
The simple object $W_\gamma$ of $\GVect$ acts on such an object of $\calm$ 
by shifting the degrees of the homogeneous components by left 
multiplication by $\gamma$ and keeping the action of $H$:
  \be
  (W_\gamma\otimes V)_g= V_{\gamma\cdot g} \,.
  \ee

Any module category over $\GVect$ can be decomposed into indecomposable module 
categories. Let us see how this works for the module categories arising in
the way considered here. To this end we consider the normal subgroup 
$K \,{:=}\, \ker\iota\leq H$ and the exact sequence 
  \be
  1\rightarrow K\rightarrow H \stackrel\pi\rightarrow J \to 1
  \labl{1KHJ1}
of groups. This sequence is, in general, not split, and $H$ is thus not a
semidirect product. Still, we can choose a set-theoretic section 
$s\colon J\To H$ of $\pi$, which for convenience we require
to respect neutral elements, $s(e_J) \eq e_H$. We keep the section
$s$ fixed from now on. For each $j\iN J$ consider the group automorphism
  \be
  \alpha_j:= \mathrm{ad}_{s(j)}|_K^{} \,\in\mathrm{Aut}(K) \,.
  \ee
The automorphism $\alpha_j$ is not necessarily inner; its class
$[\overline\alpha_j]\iN \mathrm{Out}(K)=\mathrm{Aut}(K)/\mathrm{Inn}(K)$
does not depend on the choice of $s$. Moreover, introduce group elements
  \be
  c_{i,j} := s(i)\, s(j)\, s(ij)^{-1} \,\in K
  \ee
for each pair $i,j\iN J$. Then one has the relation
  \be
  \alpha_j\circ\alpha_{j'}= \mathrm{ad}_{c_{j,j'}}\circ \alpha_{jj'}
  \ee
and obvious coherence conditions on the elements $c_{ij}\iN K$; thus
$(\alpha_j,c_{i,j})$ defines a \emph{weak action} of the group $J$ on 
the group $K$.
We use this observation to rewrite the group $H$ in a convenient way. The map
  \be
  \begin{array}{rll}
  \psi:\quad J\times K&\!\!\to\!\!& H \\
  (j,k)&\!\!\mapsto\!\!&k\cdot s(j)
  \eear
  \ee
has the inverse
  \be
  \begin{array}{rll}
  \psi^{-1}:\quad H&\!\!\to\!\!& J\times K \\ 
  h &\!\!\mapsto\!\!& (\pi(h),h\cdot (s\pi(h))^{-1}) \,.
  \eear
  \ee
Define on the set $J\Times K$ a composition map
  \be
  (i,k)\cdot (j,k'):= (ij,k\alpha_i(k') c_{ij}) \,.
  \labl{.JK}
A direct calculation shows that the map $\psi$ is compatible with the product 
\erf{.JK} and with the product on $H$. Thus \erf{.JK} endows the set 
$J\Times K$ with the structure of a finite group isomorphic to $H$.
We denote this group structure by $J \,{\times_\alpha}\, K$, 
suppressing the group elements $c$ in the notation.
We will identify $J \,{\cong}\, G/K$ with a subgroup of $G$ in the sequel.

Thus we now replace $H$ by the isomorphic group $J \,{\times_\alpha}\, K$.
Then the left $J \,{\times_\alpha}\, K$-action on
$V \eq \bigoplus_{g\in G} V_g$ satisfies
  \be
  (j,k)(V_g)\subset V_{g.j^{-1}} .
  \ee
Moreover, each homogeneous component $V_g$ has a natural structure of a 
$K$-module from the action of elements of the form
$(e_J,k)\iN J \,{\times_\alpha}\, K$.

It is crucial to note that the so obtained $K$-module structures on different
homogeneous components $V_g$ are in general not isomorphic. They are related
by the action of elements of the form $(j,k)$ that
are \emph{twisted} intertwiners rather than morphisms
of $K$-modules. Comparing the group elements $ (e,k)\cdot (j,k') \eq (j,kk') $ 
and $ (j,k')(e,k'') \eq (j,k'\alpha_j(k'')) $ we deduce that
  \be
  (e,k)\cdot (j,k') = (j,k')(e,k'') \quad\text{ with }\quad
  k''= \alpha_j^{-1}((k')^{-1} k k') \,.
  \ee
Thus the action by $(j,k')$ is a twisted intertwiner relating a $K$-module in 
the isomorphism class $[V_g]$ to a $K$-module in the class
$ [V_{g.j}] \eq \overline\alpha_j^{\,-1}[V_g]$.
These two isomorphism class are different if $\alpha_j$ is outer.

\medskip

To find the simple objects of the category $[G \srr_{\!\iota^-} H,\Vect]$,
fix representatives $(\gamma_1,\gamma_2,...\,,\gamma_r)$ for the
orbits of the right action of $J$ on $G$. Then the isomorphism classes of
simple objects are in bijection with pairs $(\gamma_i,\chi)$
with $\chi\iN\widehat K$ a simple character of $K$. 
The action of $\GVect$ on the set of isomorphism classes of simple objects
$(\gamma_i,\chi)$ of the category $[G \srr_{\!\iota^-} H,\Vect]$ 
and thus its decomposition as a module category over
$\GVect$ can now be computed explicitly.

An instructive example is the group homomorphism 
$\iota\colon H \eq S_3\To \zet_2 \eq G$\,,
with $S_3$ the symmetric group on three letters, that is given 
by the sign function. The exact sequence \erf{1KHJ1} of groups is then
  \be
  1\longrightarrow A_3\cong\zet_3\longrightarrow S_3
  \stackrel{\rm sign}\longrightarrow \zet_2 \longrightarrow 1 \,.
  \ee
The simple objects of the resulting linearization $[\zet_2 \srr S_3 , \Vect]$ 
are labeled by the single orbit of the right action of $\zet_2$ on itself
and by one of the three irreducible characters $\{1,\zeta,\zeta^\vee\}$
of $\zet_3$. Since $S_3$ is a semidirect product, any section 
$s\colon \zet_2 \To S_3$, e.g.\ the one mapping the generator of $\zet_2$
to the permutation $\tau_{12}\iN S_3$, gives a genuine action
of $\zet_2$ on $\zet_3$, rather than only a weak action. Here the generator 
of $\zet_2$ acts as the outer automorphism of $\zet_3$ which exchanges the
non-trivial irreducible characters $\zeta$ and $\zeta^\vee$.
This fixes the $\zet_3$-representation on the homogeneous component
$V_1$ in terms of the $\zet_2$-rep\-re\-sentation on $V_0$ as shown in the
following table:
  \be
  \begin{array}{cc}
  \mbox{rep. on }V_0~ & ~\mbox{rep. on }V_1 \\[3pt] \hline {}\\[-8pt]
  1 & 1 \\ \zeta& \zeta^\vee\\ \zeta^\vee&\zeta
  \eear
  \labl{tabA}
We conclude that the abelian category $[\zet_2 \srr S_3,\Vect]$ has three 
isomorphism classes of simple objects, corresponding to the three lines of 
the table. 

To determine the structure of $[\zet_2 \srr S_3,\Vect]$ as a module category 
over $\zet_2$-\Vect\ we note that the action of the simple object $X_g$ in a 
non-trivial homogeneous component exchanges the two homogeneous components
$V_0$ and $V_1$. It therefore preserves the isomorphism class of simple
$[\zet_2 \srr S_3,\Vect]$-objects in the first 
line of \erf{tabA} and exchanges the two classes in the other two lines.
Thus the first line of \erf{tabA} gives us one indecomposable module category 
over $\zet_2$-\Vect\ with a single simple object, which corresponds to $\zet_2$ 
seen as a subgroup of itself. From the second and third lines of \erf{tabA}
we get another indecomposable module category having two simple objects,
corresponding to the trivial subgroup $\{e\}$ of $\zet_2$.

\vskip 3em

{\small
\noindent{\sc Acknowledgments:}
We thank Domenico Fiorenza, Jeffrey Morton and Jan Priel for helpful 
discussions. JF is still to some extent supported 
by VR under project no.\ 621-2009-3993. CS and AV are partially supported 
by the Collaborative Research Centre 676 ``Particles, Strings and the Early 
Universe - the Structure of Matter and Space-Time'' and by the DFG Priority 
Programme 1388 ``Representation Theory''. JF is grateful to Hamburg 
University, and in particular to CS, Astrid D\"orh\"ofer and Eva Kuhlmann, 
for their hospitality when part of this work was done.
}

\newpage

 \newcommand\wb{\,\linebreak[0]} \def\wB {$\,$\wb}
 \newcommand\Bi[2]    {\bibitem[#2]{#1}} 
 \newcommand\inBO[9]  {{\em #9}, in:\ {\em #1}, {#2}\ ({#3}, {#4} {#5}), p.\ {#6--#7} {\tt [#8]}}
 \newcommand\J[7]     {{\em #7}, {#1} {#2} ({#3}) {#4--#5} {{\tt [#6]}}}
 \newcommand\JO[6]    {{\em #6}, {#1} {#2} ({#3}) {#4--#5} }
 \newcommand\JP[7]    {{\em #7}, {#1} ({#3}) {{\tt [#6]}}}
 \newcommand\BOOK[4]  {{\em #1\/} ({#2}, {#3} {#4})}
 \newcommand\PhD[2]   {{\em #2}, Ph.D.\ thesis #1}
 \newcommand\Prep[2]  {{\em #2}, preprint {\tt #1}}
 \def\aagt  {Alg.\wB\&\wB Geom.\wb Topol.}     
 \def\adma  {Adv.\wb Math.}
 \def\comp  {Com\-mun.\wb Math.\wb Phys.}
 \def\imrn  {Int.\wb Math.\wb Res.\wb Notices}
 \def\jgap  {J.\wb Geom.\wB and\wB Phys.}
 \def\joal  {J.\wB Al\-ge\-bra}
 \def\jpaa  {J.\wB Pure\wB Appl.\wb Alg.}
 \def\momj  {Mos\-cow\wB Math.\wb J.}
 \def\npbp  {Nucl.\wb Phys.\ B (Proc.\wb Suppl.)}
 \def\nupb  {Nucl.\wb Phys.\ B}
 \def\phrb  {Phys.\wb Rev.\ B}
 \def\phrx  {Phys.\wb Rev.\ X}
 \def\toap  {Topology\wB Applic.}
 \def\trgr  {Trans\-form.\wB Groups}

\end{document}